\documentclass[12pt]{article}
\usepackage{float}
\usepackage{amsmath}
\usepackage{graphicx,psfrag,epsf}
\usepackage{enumerate}
\usepackage{natbib}
\usepackage{url} 
\usepackage{tikz}
\usepackage{pgfplots}
\pgfplotsset{compat=1.17}

\usepackage[utf8]{inputenc}
\usetikzlibrary{patterns}
\usetikzlibrary{shapes}
\usepackage[linesnumbered,ruled]{algorithm2e}

\usepackage{bm}
\usepackage{dsfont}
\usepackage{amssymb}
\usepackage{amsfonts}
\usepackage{amsthm}
\usepackage{mathtools}
\usepackage{xcolor}
\usepackage{multirow}

\newtheorem{theorem}{Theorem}[section]

\newtheorem{proposition}[theorem]{Proposition}

\newtheorem{remark}{Remark}[section]

\newtheorem{algo}{Procedure}[section]

\title{Simultaneous directional inference}
\date{August 2023}
\newcommand{\mI}{\ensuremath{\mathbb I}}
\newcommand{\mE}{\ensuremath{\mathbb E}}
\newcommand{\mP}{\ensuremath{\mathbb P}}

\newcommand{\pl}{\tilde{\ell}}

\begin{document}

 \author{Ruth Heller\thanks{
    The authors gratefully acknowledge the \textit{Rita Levi-Montalcini prize for  Scientific Cooperation between
Italy and Israel.}}\hspace{.2cm}\\Department of Statistics and Operations Research, Tel-Aviv University\\{ruheller@gmail.com}
    \\
    Aldo Solari\hspace{.2cm}\\Department of Economics, Management and Statistics, University of Milano-Bicocca\\{solari.aldo@gmail.com}}

\maketitle
\begin{abstract}
We consider the problem of inference on the signs of $n>1$ parameters. We aim to provide $1-\alpha$ post-hoc confidence bounds on the number of positive and negative (or non-positive) parameters. The guarantee is simultaneous, for all subsets of parameters. 
Our suggestion is as follows: start by using the data to select the direction of the hypothesis test for each parameter; then, adjust the  $p$-values of the one-sided hypotheses for the selection, and use the adjusted $p$-values for simultaneous inference on the selected $n$ one-sided hypotheses.
The adjustment is straightforward assuming that the  $p$-values of one-sided hypotheses have densities with monotone likelihood ratio,  and are mutually independent. 
We show that the bounds we provide are tighter (often by a great margin) than existing alternatives, and that they can be obtained by  at most a polynomial time. We demonstrate the usefulness of our simultaneous post-hoc bounds  in the evaluation of treatment effects across studies or subgroups. Specifically, we provide a tight lower bound on the number of studies which are beneficial, as well as on the number of studies which are harmful (or non-beneficial), and in addition conclude on the effect direction of individual studies, while guaranteeing that the probability of at least one wrong inference is at most 0.05. 
\end{abstract}

\section{Introduction}

 Let  $\theta = (\theta_1,\ldots,\theta_n)$ be a vector of $n$ unknown real-valued parameters. A conventional analysis is two-sided multiple testing, e.g., testing the family of point null hypotheses $H_1:\theta_1=0, \ldots, H_n:\theta_n=0$ with a procedure that guarantees familywise error rate (FWER) control at a pre-specified level $\alpha$. If the $i$th null hypothesis is rejected, the conclusion is that $\theta _i \neq 0$. 
 \cite{Tukey91} argued that such a conclusion is unsatisfactory, and the analysis should  instead conclude  on the sign of the parameter, 
 i.e., that $\theta_i>0$ or that $\theta_i<0$.  
 Following rejection, it is tempting  to conclude  on the sign of the parameter based on  the data. For example, for a rejected point null hypothesis, conclude that it is positive if its point estimate is positive; otherwise conclude that it is negative. However, a directional error, referred to often as a  
 type III error may occur if we decide after rejection of the point null hypothesis $\theta_i=0$ that $\theta_i>0$ when in fact $\theta_i< 0$ (or that $\theta_i<0$ when in fact $\theta_i> 0$). Therefore, the probability of making at least one type I or type III error, henceforth referred to as the directional FWER (dFWER), 
may be larger than $\alpha$ even though $\textrm{FWER} \leq \alpha$.  See \cite{shaffer1980control} for an example of  lack of control using Holm's procedure \citep{holm1979simple} for the family of point null hypotheses when the test statistics are independent but Cauchy distributed.

 For independent test statistics 
 satisfying monotone likelihood ratio (MLR) for $\theta_i$ \citep{karlin1956theory}, existing two-sided multiple testing procedures   have been shown to control the dFWER: Holm's procedure
 \citep{shaffer1980control},
Hochberg's procedure \citep{finner1994,liu1997control}
and, more generally, any closed testing procedure 
 \citep{finner1999stepwise}. 
For example, for continuous exponential families (where the $p$-value is a monotone transformation of the sufficient statistic),  the  family of $p$-value densities satisfies the MLR property. MLR is also satisfied when the observations are normally distributed  with unknown variance when testing a single mean, or assuming equal unknown variance when comparing two means \citep{lehmann2005testing}.  

Rather than adopting existing procedures for two-sided tests for the purpose of conclusion  on the sign of each parameter, one can consider the problem of testing $2n$  one-sided hypotheses given by the following $n$ pairs: 
 \begin{eqnarray}\label{eq-hypotheses}
 H_i^-:\theta_i \leq 0,  \quad H_i^+: \theta_i \geq 0, \quad i=1, \ldots, n.
 \end{eqnarray}
The directional decisions are built-in in the  rejections of multiple testing procedures for the family of $2n$ hypotheses in \eqref{eq-hypotheses}, so the conventional error rate (FWER) coincides with the directional error rather (dFWER). 
Let $p_i$ and $q_i$ be the  $p$-values for $H_i^-$ and $H_i^+$, respectively, $i=1,\ldots,n$. 

\cite{shaffer1974bidirectional}
considered (\ref{eq-hypotheses}) for $n=1$, referring to it as a directional hypothesis-pair.
Multiple directional hypothesis-pairs were considered in 
\cite{holm1979simple,
shaffer1980control, 
shaffer2002multiplicity,
bauer1986multiple, finner1999stepwise,
guo2015stepwise}, among others.
\cite{bauer1986multiple} showed that without additional distributional assumptions, it is not possible to improve much over Holm's procedure on the $2n$  $p$-values  of the family of hypotheses in \eqref{eq-hypotheses}, for level $\alpha$ FWER control. Unfortunately, the number of rejections can be considerably smaller than Holm's procedure on two sided $p$-values (which requires more distributional assumptions for validity), see details in \S~\ref{sec-background-gr}.

In this work, we provide powerful procedures not only for making positive and negative discoveries, but also for bounding the total number of positive and negative parameters, as well as for bounding the number of positive and negative parameters in  
 any index set $\mathcal{I}\subseteq \{1,\ldots,n\}$.   Let $n^+(\mathcal{I})$ and  $n^-(\mathcal{I})$ be the number of parameters in $\mathcal I$ with positive values and negative values, respectively:
\begin{eqnarray*}\label{eq-n_plus_minus}
n^+(\mathcal{I}) = |\{ i \in \mathcal{I} : \theta_i > 0 \}|, \quad n^-(\mathcal{I}) = |\{ i \in \mathcal{I} : \theta_i < 0 \}|.
\end{eqnarray*}
For simplicity of notation, we use $n^+$ and $n^-$ instead of $n^+([n])$ and  $n^-([n])$, where $[n]$ is the index set $\{1,\ldots,n\}$. 
Our first goal is to provide tight  lower bounds for  $n^+(\mathcal{I})$ and  $n^-(\mathcal{I})$, for any $\mathcal I$. More formally, we aim to 
 provide functionals $\ell_{\alpha}^+(\cdot)$, $\ell_{\alpha}^-(\cdot)$, so that the following  guarantee holds: 
\begin{eqnarray}\label{eq-coverage_plus_minus-unconditional}
\mP_{\theta}\Big(
n^{+}(\mathcal{I}) 
\geq 
\ell_{\alpha}^{+}(\mathcal{I}),\,\,
n^{-}(\mathcal{I}) 
\geq 
\ell_{\alpha}^{-}(\mathcal{I}),
 \mathrm{\,\,for\,\,all\,\,}\mathcal{I}  \Big) \geq 1- \alpha.
\end{eqnarray}
 We shall demonstrate in this work that much tighter  lower bounds  are possible than the bounds based on  the number of discoveries in each direction in $\mathcal I$ (i.e., $|i\in \mathcal{I}: H_i^- \textrm{ is rejected}|$ for the lower bound for $n^+(\mathcal{I})$; $|i\in \mathcal{I}: H_i^+ \textrm{ is rejected} |$ for the lower bound for $n^-(\mathcal{I})$) . 
 
Tight bounds are important in various applications. 
Consider, for example, the case of  evaluating a treatment effect  on multiple subgroups (or  cohorts). So $\theta_i$ is  the average treatment effect for subgroup $i$, and a  positive lower bound on both the number of subgroups for which $\theta_i>0$ ($n^+$), and on the number of subgroups for which $\theta_i<0$ ($n^-$),  conveys the heterogeneity of the treatment effect.  This can be of major clinical importance, since it provides the researcher with the knowledge that at least for some subgroups the treatment is harmful rather than  effective (i.e., that there is a qualitative interaction, \citealt{gail1985testing, Zhao2019}).

In some settings it is enough to infer on positive and non-positive findings, i.e., the interest is only in lower bounds on $n^+(\mathcal I)$ and on $n^-(\mathcal I)+n^0(\mathcal I)$, where $n^0(\mathcal I) = |\mathcal I|-n^+(\mathcal I)-n^-(\mathcal I) = |i\in \mathcal I: \theta_i = 0|$. 
For example,    in systematic reviews in clinical trials, when it is important to evaluate the bounds on the number of studies with a positive treatment effect (i.e., $\theta_i>0$), and a positive lower bound on $n^-+n^0$ is of concern since it suggests that the treatment may have no effect or may even be  harmful \citep{IntHout16}. More generally, in replicability analysis \citep{Bogomolov22}, if interest lies in establishing that there is a positive association (i.e., $\theta_i>0$) in more than one study, then providing with $1-\alpha$ confidence the lower bound on $n^+(\mathcal I)$ as well as on $n^-(\mathcal I)+n^0(\mathcal I)$ is of great interest: the lower bound on $n^+(\mathcal I)$ provides the minimal extent of replicability   in the desired positive direction that can be stated with confidence; the lower bound on $n^-(\mathcal I)+n^0(\mathcal I)$ provides the evidence of the limit on replicability -- a  positive lower bound means that at least some studies do not replicate. 

Thus we would like to provide functionals $\pl_\alpha^+(\cdot)$, $\pl_\alpha^-(\cdot)$ so that the following guarantee holds: 
\begin{equation}\label{eq-positiveonly-simultguarantee}
  \mP_{\theta}\Big(
n^{+}(\mathcal{I}) 
\geq 
\pl_{\alpha}^{+}(\mathcal{I}),\,\,
n^{-}(\mathcal{I})+n^{0}(\mathcal{I}) 
\geq 
\pl_{\alpha}^{-}(\mathcal{I}),
 \mathrm{\,\,for\,\,all\,\,}\mathcal{I}  \Big) \geq 1- \alpha.
\end{equation}
Since $n^{-}_\alpha(\mathcal{I})\leq n^{-}_\alpha(\mathcal{I})+n^0(\mathcal I),$ the functionals $\ell_\alpha^+(\mathcal{I})$ and $\ell_\alpha^-(\mathcal{I})$ that satisfy \eqref{eq-coverage_plus_minus-unconditional} also satisfy \eqref{eq-positiveonly-simultguarantee}. So it should be possible to improve the lower bounds when aiming at the more relaxed requirement \eqref{eq-positiveonly-simultguarantee}.  
For this purpose, we modify $H_i^+$ in \eqref{eq-hypotheses} by removing the equality sign. So the hypotheses considered for $i=1,\ldots,n$ are:
\begin{eqnarray}\label{eq-setup}
 H_i^-: \theta_i \leq 0, \quad K_i: \theta_i > 0, \quad i=1,\ldots,n.
 \end{eqnarray}
Exactly one of the hypotheses pairs is true, i.e., in the family of $2n$ null hypotheses defined in \eqref{eq-setup},  there are exactly $n$ true one-sided null hypotheses. Therefore, for making positive and non-positive discoveries
 more powerful procedures than those that are typically used with two-sided testing can be devised \citep{bauer1986multiple,  guo2015stepwise}.   \cite{bauer1986multiple} showed that for testing the $2n$ hypotheses defined in \eqref{eq-setup} using Bonferroni, discoveries are made with  the Bonferroni cut-off $\alpha/n$, which is twice as large as as the Bonferroni cut-off for the family of hypotheses \eqref{eq-hypotheses}. This can be understood from the work of \cite{shaffer1986modified}, who showed that when performing a Bonferroni procedure, the proper correction factor for FWER control is not the number of hypotheses tested, but the maximum number of null hypotheses that can simultaneously be true \citep{goeman2010}.  In \S~\ref{sec-background-gr} we review additional FWER and FDR controlling procedures for \eqref{eq-setup}.

The bounds we develop are uniformly tighter when considering the pairs of hypotheses in \eqref{eq-setup} rather than in \eqref{eq-hypotheses}. Therefore,  the approach based on the pairs of hypotheses in \eqref{eq-setup} should always be preferred -- unless it is important to infer on the positive as well as the negative parameters. 
Support for testing \eqref{eq-setup} is also given by \cite{Tukey91}. 
With regard to a single pair of hypotheses, \cite{Tukey91} argued as follows: 
``Statisticians classically asked the wrong question -- and were willing to answer with a lie, one that was often a downright lie. They asked ``Are the effects of A and B different?'' and they were willing to answer ``no.''
All we know about the world teaches us that the effects of A and B are always different--in some decimal place--for any A and B. Thus asking ``Are the effects different?'' is foolish.
What we should be answering first is ``Can we tell the direction in which the effects of A differ from the effects of B?'' In other words, can we be confident about the direction from A to B? Is it ``up,'' ``down'' or ``uncertain''?". So in cases where the point null hypotheses are never true, the hypotheses in \eqref{eq-setup} should clearly  be preferred over the hypotheses in \eqref{eq-hypotheses}. According to \cite{Tukey91, Jones2000}, this is often the case.

\subsection{FWER and FDR controlling procedures for testing \eqref{eq-hypotheses} and \eqref{eq-setup}}\label{sec-background-gr}

Let $x_i=\min(p_i,q_i), i=1, \ldots, n$ be the $n$ smallest  $p$-values of the $2n$ one-sided hypotheses. Let $x_{(1)}\leq \ldots \leq x_{(n)}$ be their ordered values. 
For critical values $\alpha_1\leq \ldots \leq \alpha_n$, the procedure for directional inference is as follows: 
\begin{enumerate}
    \item Let 
\begin{eqnarray}\label{eq-stepwiseR}
R =
\begin{cases}
\max \lbrace r: x_{(r)} \leq \alpha_i, \  \forall i\leq r\rbrace & \text{for a step-down procedure } , \\
\max \lbrace r: x_{(r)} \leq \alpha_r\rbrace & \text{for a step-up procedure }.
\end{cases}
\end{eqnarray} 
    \item For $r\leq R$, if  $x_{(r)}$ is the $p$-value for $H_i^-$, declare $\theta_i>0$, otherwise declare $\theta_i< 0$ if testing \eqref{eq-hypotheses} or $\theta_i\leq 0$ if testing \eqref{eq-setup}.
\end{enumerate}

For testing \eqref{eq-hypotheses} with FWER control, a valid procedure for any dependence structure among the  test statistics is Holm's procedure on the $2n$ $p$-values, i.e., the step-down procedure with critical values   $\alpha_i= \alpha/(2n-i+1), i=1,\ldots,n$. 

The next procedures were shown to provide FWER control under the additional distributional  assumption of uniform validity
\citep{guo2015stepwise, Zhao2019, ellis2020gaining}. A valid $p$-value $x_i$ is {\it uniformly valid} if for all $\tau \in (0,1)$, $x_i/\tau$ given $x_i\leq \tau$ is valid, i.e., if the null hypothesis is true then 
$$ \mP(x_i \leq t \mid x_i\leq \tau) \leq \frac t\tau \ \  \forall \ \ 0\leq t\leq \tau \leq 1.$$
A sufficient condition for uniform validity is the MLR property of the test statistic for $\theta_i$.
It is in fact sufficient to require that the inequality above be satisfied for $\tau=1/2$ (this is the weaker condition of conditional validity, defined in $(A1)$ of \S~\ref{sec-setup}). 
These procedures are detailed in the first two rows of Table \ref{tab:procedures_overview}.

Note that the procedure in row 1 has larger critical values than Holm's procedure on the $2n$ $p$-values (which has critical values $\alpha_i = \alpha/(2n-i+1), \ i=1,\ldots,n$),  and hence can lead to considerably more discoveries.  Moreover,  the procedure in row 2 has critical values that are almost twice as large as the procedure in row 1, so it should be preferred if it is enough to discover positive and non-positive (rather than negative) parameters. 

The procedures in the last two rows of Table \ref{tab:procedures_overview} provide FDR control on positive and negative discoveries (row 5) and on positive and non-positive discoveries (row 6). 
The procedure in row 6 provides a uniform improvement, in terms of the number of discoveries, over the procedure in row 5, since the critical values are twice as large. 

\begin{table}
\caption{\label{tab:procedures_overview}Overview of existing and proposed procedures for error criterion and the type of inference considered. The critical values refer to the step-wise procedure outlined in Equation (\ref{eq-stepwiseR}). The procedures in the first four rows require conditional validity and independence of the $n$ test statistics, i.e., conditions (A1)-(A2) in \S~\ref{sec-setup}. The procedures in the last two rows require only independence of the $n$ test statistics.}
\centering
\scalebox{0.9}{
\begin{tabular}{llccl}
\hline
Criterion & Procedure & Type of inference  & Critical values $\alpha_i$\\
 \hline
\multirow{2}{*}{FWER}  & Holm $n$ two-sided \citep{shaffer1980control} & positive/negative  & $\alpha/\{2(n-i+1)\}$\\
  & Guo-Romano \citep{guo2015stepwise} & positive/non-positive  & $\alpha/(n-i+1 + \alpha)$\\
   \hline
\multirow{2}{*}{FDP} & Directional closed testing (\S~\ref{sec-setup})  & positive/negative &  - \\
 & Partitioning with adaptive tests (\S~\ref{sec-partitioning}) & positive/non-positive & - \\
 \hline
\multirow{2}{*}{FDR} & Benjamini-Hochberg \citep{Benjamini05}& positive/negative  & $i\alpha/2n$\\
 & Guo-Romano \citep{guo2015stepwise} & positive/non-positive  & $i\alpha/n$\\
\hline
\end{tabular}
}
\end{table}

We conclude this section by a brief overview of our contributions, which include the two procedures in rows 3 and 4 of Table \ref{tab:procedures_overview}.

\subsection{A brief overview of our contributions}\label{sec-Introduction-overview}

Closed testing procedures  
are  useful not only for making discoveries, but also in order to provide  simultaneous confidence lower bounds on the number of true discoveries in $\mathcal I$, or equivalently simultaneous confidence upper bounds on the false discovery proportion (FDP)  in $\mathcal I$ (\citealt{goeman2011, blanchard20}, reviewed in \S~\ref{sec-background}). These bounds are often referred to as post-hoc confidence bounds. 
When testing point null hypotheses with a closed testing procedure, the number of true discoveries are the number of parameters which are not zero. A lower bound on the number of parameters which are not zero does not reveal information about the number of positive and negative parameters in the set, and therefore may not be satisfactory. More specifically, from having a functional $\ell_{\alpha}(\cdot)$
that satisfies 
\begin{eqnarray}\label{eqn-undirectional-lb-guarantee}
\mP_{\theta}\Big(
n^{+}(\mathcal{I}) +n^{-}(\mathcal{I}) 
\geq 
\ell_{\alpha}(\mathcal{I}),
 \mathrm{\,\,for\,\,all\,\,}\mathcal{I}  \Big) \geq 1- \alpha.
\end{eqnarray}
it is not clear how to extract the lower bounds to satisfy \eqref{eq-coverage_plus_minus-unconditional}.

For finding $\ell_{\alpha}^+(\cdot)$ and $\ell_{\alpha}^-(\cdot)$ that satisfy \eqref{eq-coverage_plus_minus-unconditional}, we suggest the following two step approach, which we refer to henceforth as {\em directional closed testing} (DCT): first, 
select from each pair the hypothesis to test (based on the data); second, on the selected $n$ one-sided hypotheses, apply an $\alpha$ level closed testing procedure (reviewed in \S~\ref{sec-background}). Of course, the second step has  to adjust for the first step of selection from the same data. With this approach, we can provide simultaneous confidence bounds on $n^+(\mathcal I)$ and $n^-(\mathcal I)$ for any $\mathcal I$, building upon the work of \cite{goeman2011} that showed how to obtain simultaneous bounds for a closed testing procedure. 
In particular, we can provide tight lower bounds on $n^+$ and $n^-$, as well as positive and negative discoveries, with the guarantee that  the probability of any wrong inference 
is at most $\alpha$. 

Interestingly, a closed testing procedure on two-sided $p$-values (suggested by \cite{finner1999stepwise} for dFWER control) implicitly starts by the first step in the DCT procedure.  However, viewing closed testing on two-sided $p$-values as a DCT procedure has several advantages. First, by realizing that the first step is the selection of the one-sided hypotheses to test, the follow-up inference is explicit on the selected one-sided hypotheses and therefore controls the usual error rate (FWER), so there is no need for a special definition of directional error rate (dFWER), see \cite{Benjamini10} for a discussion of the potential harms in having too many concepts of error rates.   
 Second, in addition to 
 concluding  on the signs of   individual parameters DCT provides simultaneous lower bounds on $n^+(\mathcal I)$ and $n^-(\mathcal I)$ as described in \eqref{eq-bounds_I_plus_minus}.  
Third,  the sufficient assumption $(A1)$ in \S~\ref{sec-setup} is more general than the distributional assumption  
for dFWER control considered in \cite{shaffer1980control, 
finner1999stepwise, sarkar2004two, guo2015stepwise}. 

Another advantage of DCT is that the second step in the DCT procedure can start by testing for qualitative interactions (i.e., testing the null hypothesis that all parameters are non-positive or all parameters are non-negative) using the state-of-art approach  suggested in \cite{Zhao2019}.  The  power to detect qualitative interactions will thus be as in \cite{Zhao2019}, but in addition DCT provides (at no additional cost for multiplicity) discoveries  and confidence lower bounds on the number of positive and negative parameters for any subset of parameters of interest, see \S~\ref{subsec - QI} for details.

It is worth noting that once we have obtained the confidence lower bounds on the number of positive and negative parameters, we can automatically derive the corresponding upper bounds, i.e. if $n^{+}(\mathcal{I}) \geq \ell_\alpha^+(\mathcal{I}), n^{-}(\mathcal{I}) \geq \ell_\alpha^-(\mathcal{I})$ holds, then also $n^{+}(\mathcal{I}) \leq  |\mathcal{I}| - \ell^-_\alpha(\mathcal{I}), n^{-}_\alpha(\mathcal{I}) \leq |\mathcal{I}| - \ell^{+}_\alpha(\mathcal{I})$ holds.

For finding $\pl_{\alpha}^+(\cdot)$ and $\pl_{\alpha}^-(\cdot)$ that satisfy \eqref{eq-positiveonly-simultguarantee}, we suggest applying the partitioning principle (reviewed in \S~\ref{sec-background}) to the $2^n$ orthants defined by $\{(-\infty,0],(0, \infty)\}^n $.
 Each orthant defines the {\it orthant hypothesis} that $\theta$ is in the orthant. There are $2^n$ orthant hypotheses, but  exactly one orthant  hypothesis is true.  Therefore, it is enough to test each at level $\alpha$. This is in contrast to DCT, where multiple testing adjustments are made (adjustments are necessary, since for the zero vector $\theta = 0$ the $2n$ hypotheses in \eqref{eq-hypotheses} are true).
  
We  show the general approach for extracting the bounds,  as well as the specific polynomial time algorithm that uniformly improves over the DCT bounds $\ell_{\alpha}^+(\cdot), \ell_{\alpha}^-(\cdot)$, so that   $\pl_{\alpha}^+(\cdot)\geq \ell_{\alpha}^+(\cdot)$, $\pl_{\alpha}^-(\cdot)\geq \ell_{\alpha}^-(\cdot)$. 

In \S~\ref{sec-background}, we review the necessary statistical background and some of the notation and definitions we are going to use throughout this work. 
 Then, we explain in detail  DCT  in  \S~\ref{sec-setup}, and partitioning  in  \S~\ref{sec-partitioning}. 
In \S~\ref{sec - Zhao} we provide usages of our suggested methods: for subgroup analysis,  for providing inference following the test of qualitative interactions, and for enhancing meta-analysis. In \S~\ref{sec - sim} we compare and contrast the methods in simulations, and provide some practical recommendations.  
Moreover, we quantify the improvement upon the methods targeting FWER or FDR control in Table \ref{tab:procedures_overview} (their lower bounds are computed merely by counting the number of rejections in each direction).
Final remarks are provided in \S~\ref{sec - disc}. All proofs are in  \S~\ref{Appendix-proofs} of the Supplementary Material (SM). R code reproducing the examples is available at \verb"https://aldosolari.github.io/directionalinference/".

\section{Background on simultaneous inference}\label{sec-background}

Let $H_1,\ldots, H_n$ be $n$ base null hypotheses, or base hypotheses, which form the basic units for the inference. The corresponding $p$-values are $x_1, \ldots, x_n$. Each $x_i$ is a {\it valid $p$-value} for its respective null hypothesis $H_i$, i.e., its distribution  is uniform or stochastically larger than uniform when $H_i$ is true. 
For $\mathcal I \subseteq \{1,\ldots,n\}$, let $H_{\mathcal I}: \bigcap_{i\in \mathcal I}H_i$
be the \emph{intersection hypothesis} of the $|\mathcal I|$  base hypotheses in $\mathcal I$. This hypothesis is true if $\forall i\in \mathcal I, \ H_i$ are true. 
The  level $\alpha$ test of $H_{\mathcal I}$, denoted by $\phi_{\mathcal I}\in \{0,1\}, $ is called the  \emph{local test}. It satisfies that  
 $\mP_{\theta}(\phi_{\mathcal I}=1)\leq \alpha$  when $\theta\in H_{\mathcal I}$.

 \begin{paragraph}{\bf The closed testing procedure} 
The closed testing procedure corrects the local tests for multiple testing by rejecting $H_{\mathcal I}$ only if it is rejected by all intersection hypotheses that include $\mathcal I$, i.e., if  $\bar{\phi}_{\mathcal{I}}=1$, where
$\bar{\phi}_{\mathcal{I}} = \min\{\phi_{\mathcal{J}}: \mathcal{J} \supseteq \mathcal{I}\}.$
\cite{marcus1976} showed that the adjusted tests $\bar{\phi}_{\mathcal{I}}$ have FWER control:
$
\mP_{\theta}\Big(\bar{\phi}_{\mathcal{I}}=0 \mathrm{\,\,for\,\,all\,\,}\mathcal{I} \subseteq \mathcal{T} \Big) \geq 1- \alpha,
$
where $\mathcal T$ is the (unknown) index set of the true base hypotheses. The \emph{discoveries} are the base hypotheses rejected, 
$\mathcal{D}_{\alpha} =  \{i \in [n]: \bar{\phi}_{i}=1\}  
$. 
\end{paragraph}

\begin{paragraph}
{\bf Simultaneous post-hoc bounds}
\cite{goeman2011} showed that the lower bound for the number of true discoveries in $\mathcal{I}$, $|\mathcal I\setminus  \mathcal{T}|$,  is
\begin{equation}\label{eq-undirectional-lb}
\ell_{\alpha}(\mathcal{I}) = \min(|\mathcal{I} \setminus \mathcal{J}|: \mathcal{J} \subseteq \mathcal{I}, \bar{\phi}_{\mathcal{J}}=0), 
\end{equation} with the guarantee that 
$
 \mP_{\theta}\Big(   |\mathcal{I} \setminus \mathcal{T}| \geq  \ell_{\alpha}(\mathcal I), \mathrm{\,\,for\,\,all\,\,}\mathcal{I} \Big) \geq 1- \alpha.
 $ For two-sided hypotheses, this is the guarantee in \eqref{eqn-undirectional-lb-guarantee}. 
Note that  $\ell_{\alpha}(\cdot)$ can be used to upper bound the FDP, which is $|\mathcal I\cap \mathcal{T}|/|\mathcal I|$:  with at least $(1-\alpha)$ confidence,  $|\mathcal I\cap \mathcal{T}|/|\mathcal I|$ is at most $(|\mathcal{I}| -\ell_{\alpha}(\mathcal I))/|\mathcal I|$. 
 \end{paragraph}

\begin{paragraph}{\bf The combining function for testing an intersection hypothesis}
 The {\it combining function} $f: [0,1]^{|\mathcal I|} \longrightarrow [0,1]$, maps the $|\mathcal I|$  $p$-values for $H_i, i\in \mathcal I$, into a valid $p$-value for testing $H_{\mathcal I}$. The local test is 
 $$\phi_{\mathcal{I}} = \mathds{1}\{ f(\{x_i, i\in \mathcal I\}) \leq \alpha \}.$$
For mutually independent $p$-values, a popular combining function is Fisher, which has excellent power properties for a wide range
of signals \citep{benjamini2008screening, hoang2021combining} is $f_{\mathrm{Fisher}}(x_1,\ldots,x_d) = \mP(\chi^2_{2d}\geq -2\sum_{i=1}^d \log x_i)$, where $\chi^2_{2d}$ is a chi-squared random variable with $2d$ degrees of freedom.  Let $x_{(1)}\leq \ldots \leq x_{(d)}$ be the ordered values. Another popular combining function is Simes\footnote{This combining method  produces a valid $p$-value for the intersection hypothesis also if the $p$-values being combined are positive dependent, \cite{goeman2017}.}, $f_{\mathrm{Simes}}(x_1,\ldots,x_d) =\min_{1\leq i\leq d}\left\lbrace \frac{d}{i}x_{(i)}\right\rbrace$.  
It  is essentially based on a single quantile, and thus does not take full advantage of the pooled evidence when the number of positive (or negative) parameters is large. \cite{Bogomolov21} suggested the   the following combination function,
$\min_{1\leq i\leq d}\left\lbrace \frac{2(\sum_{l=1}^d\mathds{1}(x_{l}>0.5)+1)}{i}x_{(i)}\right\rbrace$, 
which uses the plug-in estimator for the number of null hypotheses combined. 
It may be much smaller than Simes, except in the case where  the intersection is only of two hypotheses. Therefore, it makes sense to switch to the Simes combining method in that case. So the modified Simes (mSimes) combining function we shall consider is
: 
\begin{equation}\label{eq-modified-adaptive-Simes}
 f_{\mathrm{mSimes}}(x_{1}, \ldots, x_{d})  =
\begin{cases}
\min_{1\leq i\leq d}\left\lbrace \frac{2(\sum_{l=1}^d\mathds{1}(x_{l}>0.5)+1)}{i}x_{(i)}\right\rbrace
& \text{if } d > 2, \\
 \min_{1\leq i\leq d}\left\lbrace \frac{d}{i}x_{(i)}\right\rbrace & \text{otherwise}.
\end{cases}
\end{equation}

Our assumption regarding  $f(\cdot)$, satisfied by all the combining functions above, is:

\begin{description}
\item[$(A0)$] 
The combining function  
satisfies monotonicity:
$f(x_1, \ldots, x_d) \leq f((x'_1, \ldots, x'_{d}))$
for $x_i \leq x_i'$ for all $i =1,\ldots,d$; and symmetry:
$f(x_1, \ldots, x_d) = f( x_{j_1}, \ldots, x_{j_d})$
for any permutation $(j_1,\ldots,j_{d})$ of $(1,\ldots,d)$.
\end{description}
\end{paragraph}

\begin{paragraph}{\bf The computational complexity for evaluating post-hoc bounds}
 Computing $\ell_{\alpha}(\mathcal I)$ involves the evaluation of exponentially many tests, which hinders its practical application. For specific combining functions, however, there exist shortcuts to derive the closed testing results in polynomial time. Quadratic time shortcuts for combining functions satisfying $(A0)$ have been developed
for FWER control \citep{dobriban2020fast} and simultaneous FDP control \citep{goeman2011, goeman2021}.  
\citet{tian2021large} showed that if the combining function satisfies also the separability condition (see Appendix B of \citet{tian2021large} for the definition), 
both
FDP and FWER shortcuts can be reduced to linear time, after an initial sorting of the $p$-values.  Separability is satisfied by the Fisher's combining function, as well by the combining function based on generalized means of $p$-values \citep{vovk2020combining}. For the Simes' combining function, the shortcut described in \citet{goeman2017} allows calculation of $\ell_\alpha(\mathcal{I})$ in linear time, after an initial preparatory step that takes
linearithmic time. 
\end{paragraph}

\begin{paragraph}{\bf The partitioning principle}
Closed testing is a fundamental principle for controlling familywise error, but there is another equally fundamental principle known as \emph{partitioning principle} \citep{finner2002partitioning}. The key idea  is to partition the parameter space into disjoint regions and simultaneously test the hypotheses that the true parameter value lies within each region. Since the true value lies within exactly one of these regions, there is only one true null hypothesis, and all the others are false. Therefore, if the test for the true null hypothesis controls the Type I error probability, then the FWER is also controlled.
The partition procedure for the base hypotheses $H_1,\ldots, H_n$ defines for every $\mathcal{I} \subseteq [n]$ the \emph{partitioning hypothesis}
$$J_{\mathcal{I}} = \big\{\bigcap_{i\in \mathcal{I}} H_i \big\} \cap \big\{ \bigcap_{j\in [n] \setminus \mathcal{I}  }H_j^c \big\},$$ where $H_j^c$ is the 
$j$th alternative hypothesis, i.e. the complement of $H_j$. The partitioning hypothesis $J_\mathcal{I}$ is true if and only if all $H_i$, $i \in \mathcal{I}$, are true and all $H_i$, $i \in [n] \setminus\mathcal{I}$, are false. The discoveries are all indices $i\in [n]$ for which 
all $J_{\mathcal{I}}$ with $\mathcal{I} \ni i$ are rejected at level $\alpha$. 

Closed testing and partitioning are equivalent principles for FWER control: for every closed testing procedure, there exists a partitioning procedure that rejects exactly the same hypotheses. Likewise, for every partitioning procedure, there exists a closed testing procedure that rejects exactly the same hypotheses \citep{goeman2021}. The partitioning principle has been argued to be  more convenient than closed testing  for obtaining confidence intervals and sets in \cite{stefansson1988confidence, finner2021partitioning}. Since the family of hypotheses in \eqref{eq-setup} has $n$ true null hypotheses which are in exactly one of the $2^n$ possible orthants, we suggest in \S~\ref{sec-partitioning} to use the partitioning principle in order to obtain confidence sets for $n^+(\mathcal I)$ and $n^0(\mathcal I)+n^-(\mathcal I)$. 
\end{paragraph}

\section{Directional closed testing}\label{sec-setup}

Let $p = (p_1,\ldots, p_n)$ denote the $p$-value vector for $(H_1^-, \ldots, H_n^-)$, and  $q = (q_1,\ldots, q_n)$ denote the $p$-value vector for $(H_1^+, \ldots, H_n^+)$. For simplicity, we assume that the test statistics are continuous and $q_i=1-p_i$, $i\in[n]$. Our assumptions for the $p$-values are: 
\begin{description}
 \item[$(A1)$] Each $p$-value is conditionally valid:
\begin{eqnarray*}
&& \sup_{\theta \in H^-_i} \mP_\theta(2p_i \leq x \mid p_i\leq 1/2)\leq x \quad \forall x \in [0,1],\\
&&\sup_{\theta \in H^+_i} \mP_\theta(2q_i \leq x \mid p_i> 1/2)\leq x \quad \forall x \in [0,1].
\end{eqnarray*}
\item[$(A2)$] $p_1,\ldots,p_n$ are mutually independent.
\end{description}

We are interested in simultaneous inference on the selected family of $n$ directional hypotheses $\{H_i^-: p_i\leq 1/2 \}\cup \{H_i^+: p_i> 1/2 \}.$  
We propose the following procedure. 

\begin{algo}[Directional closed testing]\label{algo-conditional closed testing} $ $
 \begin{enumerate}
        \item[Step 1] Select the $n$ one-sided hypotheses  for testing: $\{H_i^-: p_i\leq 1/2 \}\cup \{H_i^+: p_i> 1/2 \}$. Let $\mathcal S^- = \{i: p_i\leq 1/2  \}$ and  $\mathcal S^+ = \{i: p_i > 1/2  \}$ be, respectively, the indices for which we test  $H_i^-$ and $H_i^+$. 
        \item[Step 2] Apply a level $\alpha$ closed testing procedure on the family of $n$ one-sided selected hypotheses
        ,  using the conditional $p$-values $\{2p_i, i\in \mathcal S^- \}\cup \{2q_j, j\in \mathcal S^+ \}.$  
        The intersection hypothesis $$\displaystyle H_{\mathcal{I}} : \big\{\bigcap_{i \in \mathcal{I} \cap \mathcal{S}^-} H_i^- \big\}\cap  \big\{\bigcap_{i \in \mathcal{I} \cap \mathcal{S}^+} H_i^+ \big\}$$    is true if and only if all selected hypotheses with $i \in \mathcal{I}$ are true. For all $\mathcal I \subseteq [n]$, testing $H_{\mathcal{I}}$ at level $\alpha$ is done by the local test
\begin{eqnarray}\label{eq-combining_f}
\phi_{\mathcal{I}} = \mathds{1}\{ f(\{2p_i, i\in \mathcal I \cap \mathcal S^- \}\cup \{2q_j, j\in \mathcal I \cap \mathcal S^+ \}) \leq \alpha \}
\end{eqnarray}
using a combining function $f(\cdot)$. 
 \end{enumerate}
 \end{algo}

Note that $f(\{2p_i, i\in \mathcal I \cap \mathcal S^- \}\cup \{2q_j, j\in \mathcal I \cap \mathcal S^+ \}  = f(2\min(p_i,q_i), i \in \mathcal I)$,  so the test of $H_{\mathcal I}$ coincides with the test of the intersection of two-sided null hypotheses in $\mathcal I$ using  two-sided $p$-values. 
We denote by $\mathcal S$ the vector of signs that we condition on in Step 1 of the DCT Procedure \ref{algo-conditional closed testing}: $$\mathcal S = (sign(p_1-1/2), \ldots, sign(p_n-1/2)), \quad
 sign(p_i-1/2)  =
\begin{cases}
-1 & \text{if } p_i \leq 1/2, \\
1 & \text{if } p_i > 1/2.
\end{cases}
$$
Conditional on the vector of signs  $\mathcal S \in \{-1,1\}^n, $ the probability of rejecting the intersection hypotheses of the true nulls among the selected is at most $\alpha$, since in the local test in \eqref{eq-combining_f}  we combine (conditional) $p$-values that have each a distribution that is  at least as large as uniform given $\mathcal S$. We formalize this in the next proposition.

\begin{proposition}\label{prop-conditionalCT}
Let $p_1,\ldots,p_n$ fulfil conditions $(A1)-(A2)$, and let $\theta$ be the true unknown vector of parameters. So among the $n$ selected for testing in Step 1 of Procedure \ref{algo-conditional closed testing}, $\mathcal T = \lbrace \{i: \theta_i \in H_i^- \}\cap \mathcal S^-\rbrace\cup \lbrace \{i: \theta_i \in H_i^+ \}\cap \mathcal S^+\rbrace$ are the indices of true one sided null hypotheses. 
Then the test of the intersection of the hypotheses in $\mathcal T$ satisfies 
conditional type I error control:
\begin{equation}\label{eq-conditionalguarantee}
\sup_{\theta' \in H_{\mathcal{T}}}  \mP_{\theta'}\Big(  \phi_{\mathcal{T}} = 1 \mid \mathcal S \Big) \leq \alpha.
\end{equation}

\end{proposition}

Proposition \ref{prop-conditionalCT} implies that with probability at least $1-\alpha$ all the rejections of intersection hypotheses in Procedure \ref{algo-conditional closed testing}  are correct, conditionally on $\mathcal S$ and therefore also unconditionally:  
$\mP_{\theta}\left(\bar{\phi}_{\mathcal{I}} =0 \ \forall \ \mathcal I\subseteq \mathcal T \mid \mathcal S\right)\geq 1-\alpha$, $\mP_{\theta}\left(\bar{\phi}_{\mathcal{I}} =0 \ \forall \ \mathcal I\subseteq \mathcal T \right)\geq 1-\alpha$,
where 
$\bar{\phi}_{\mathcal{I}} = \min\{\phi_{\mathcal{J}}: \mathcal{J} \supseteq \mathcal{I}\}$ for $\phi_{\mathcal{I}}$ in \eqref{eq-combining_f}.

We now show explicitly how to obtain discoveries and confidence bounds on $n^+(\mathcal I)$ and $n^-(\mathcal I)$ from Procedure \ref{algo-conditional closed testing}. 
Let 
$\mathcal{D}_{\alpha}^{+}  = \{i \in \mathcal{S}^{-}: \bar{\phi}_{i}=1\}$ and $ \mathcal{D}_{\alpha}^{-}  = \{i \in \mathcal{S}^{+}: \bar{\phi}_{i}=1\}
$
   be the index sets of the positive and negative discoveries, respectively, 
 where $\bar{\phi}_{i} = \min\{\phi_{\mathcal{J}}:  \mathcal{J}\ni i\}$. 
The procedure concludes $\theta_i >0$ for all 
$i \in \mathcal{D}_{\alpha}^{+}$ and $\theta_i <0$ for all $i \in \mathcal{D}_{\alpha}^{-}$ while guaranteeing that the probability of at least one false conclusion  on the sign of the parameter, given the vector of signs $\mathcal S$, is at most $\alpha$.

More generally, for any $\mathcal I\subseteq [n]$, we can use  \eqref{eq-undirectional-lb} as follows.  For $\mathcal I \cap \mathcal S^-$, $\ell_{\alpha}(\mathcal{I}\cap \mathcal S^-)$ is the lower bound on the number of positive parameters: $n^+(\mathcal{I}\cap \mathcal S^-)\geq \ell_{\alpha}(\mathcal{I}\cap \mathcal S^-) $. Similarly,  for $\mathcal I \cap \mathcal S^+$: $
\ n^-(\mathcal{I}\cap \mathcal S^+)\geq \ell_{\alpha}(\mathcal{I}\cap \mathcal S^+).  $ From \eqref{eq-undirectional-lb} and \eqref{eq-conditionalguarantee} it follows that the  bounds 
\begin{eqnarray}\label{eq-bounds_I_plus_minus}
\ell_{\alpha}^{+}(\mathcal{I}) =  \ell_{\alpha}( \mathcal{I} \cap \mathcal{S}^{-}), \quad \ell_{\alpha}^{-}(\mathcal{I}) = \ell_{\alpha}( \mathcal{I} \cap \mathcal{S}^{+}) 
\end{eqnarray}
are such that
\begin{eqnarray}\label{eq-coverage_plus_minus}
\mP_{\theta}\Big(
n^{+}(\mathcal{I})\geq \ell_{\alpha}^{+}(\mathcal{I}), 
\,\,
n^{-}(\mathcal{I})\geq \ell_{\alpha}^{-}(\mathcal{I})
 \mathrm{\,\,for\,\,all\,\,}\mathcal{I} \mid \mathcal S \Big) \geq 1- \alpha.
\end{eqnarray}
In particular for $\mathcal I = [n]$, the $1-\alpha$ confidence bounds for $n^+$ and $n^-$ are $ \ell_{\alpha}^{+} = \ell_{\alpha}(\mathcal{S}^{-})$ and $ \ell_{\alpha}^{-} = \ell_{\alpha}(\mathcal{S}^{+})$, respectively. 
Note that the positive and negative discoveries are given by 
$
\mathcal{D}^{+}_{\alpha} = \{i: \ell_{\alpha}^{+}(i)=1\},
\mathcal{D}^{-}_{\alpha} = \{i: \ell_{\alpha}^{-}(i)=1\}, 
$
and $\ell_\alpha^+(\mathcal{I})\geq |\mathcal{D}^{+}_{\alpha}|$, $\ell_\alpha^-(\mathcal{I})\geq |\mathcal{D}^{-}_{\alpha}|$ always holds. In \S~\ref{sec - Zhao} - \S~\ref{sec - sim} we show that the gap between the lower bound and the number of discoveries can be large for $\mathcal I = [n]$, and it is larger for Fisher's combining function than for mSimes and Simes. 
Finally, note that 
$\ell_\alpha(\mathcal{I})$ serves as a lower bound for the number of non-zero parameters $n^+(\mathcal{I}) + n^-(\mathcal{I})$, and  $\ell_\alpha(\mathcal{I}) \geq \ell_\alpha^+(\mathcal{I}) + \ell_\alpha^-(\mathcal{I})$ always holds. 

Computing (\ref{eq-bounds_I_plus_minus}) involves at most the evaluation of exponentially many local tests \eqref{eq-combining_f}, but 
for specific  combining functions  there exist shortcuts to
derive the closed testing results in polynomial time, as detailed in \S~\ref{sec-background}. Therefore, the computational time is reasonable even for moderately large $n$.

\subsection{Dependent $p$-values}\label{subsec-dep}
DCT is valid if the local test satisfies \eqref{eq-conditionalguarantee}. Sufficient conditions are (A0)-(A2). 
However, for specific dependencies among the  $p$-values or test statistics, it may be that \eqref{eq-conditionalguarantee} is satisfied even if (A2) is not satisfied.

    A specific setting with dependent $p$-values for which we have a local test that satisfies \eqref{eq-conditionalguarantee} is the setting of testing $n$ population normal means, where the population variance $\sigma^2$ is unknown, but identical across all populations. This is the common one-way analysis of variance setting.  We shall describe this setting using the notation $N(\mu,\tau)$, $\chi^2_{\nu}$, $t_{\nu}$, and $F_{k,\nu}$ to denote the normal, chi-squared, central t, and central F distributions, respectively 
    (with their respective parameters or degrees of freedom), and $t_{\nu,x}$, $F_{k,\nu,x}$ denote the $x$-quantile of $t_{\nu}$ and $F_{k,\nu}$, respectively. 
    
        The test statistic for the mean $\theta_i$ is $y_i/(k_is)$, where $y_i\sim N(\theta_i, k_i\sigma)$ and $\nu s^2/\sigma^2\sim \chi^2_{\nu}$ for the appropriate $k_i$ and $\nu$ that depend on the sample sizes, $i=1,\ldots,n$. So $y_i/(k_is)\sim t_{\nu}$ and  $\frac 1{|\mathcal I|}\sum_{i\in \mathcal I}\frac{y_i^2}{k_i^2s^2}\sim F_{|\mathcal I|, \nu}$  for the zero vector $\theta=0$, for all $\mathcal I \subseteq [n]$. Although $y_1,\ldots,y_n$ are independent, the test statistics $y_i/(k_i s), i=1,\ldots,n$ are dependent since $\sigma$ is estimated using all observations.  The $p$-values are $p_i = \mP(t_{\nu}\geq y_i/(k_i s))$ and $q_i =\mP(t_{\nu}\leq y_i/(k_i s))$. 

A natural combining function is  the one for the F-test  $$f(\{2p_i, i\in \mathcal I \cap \mathcal S^- \}\cup \{2q_j, j\in \mathcal I \cap \mathcal S^+ \}  = f(2\min(p_i,q_i), i \in \mathcal I) = 
\mP\left(F_{|\mathcal I|, \nu}\geq \frac 1{|\mathcal I|}\sum_{i\in \mathcal I}\frac{y_i^2}{k_i^2s^2}\right). 
$$
To see that this is a valid combining function, note that for any $\theta \in H_{\mathcal I}$ (defined in Step 2 of the DCT procedure \ref{algo-conditional closed testing}):  
\begin{eqnarray}
    && \mP_{\theta }\left(\frac 1{|\mathcal I|}\sum_{i\in \mathcal I}\frac{y_i^2}{k_i^2s^2}\geq F_{|\mathcal I|, \nu, 1-x}\geq \mid \mathcal S\right) = 
    \mE\left(\mP_{\theta }\left(\frac 1{|\mathcal I|}\sum_{i\in \mathcal I}\frac{y_i^2}{k_i^2s^2}\geq F_{|\mathcal I|, \nu, 1-x} \mid \mathcal S, s \right) \right) \nonumber \\ 
    && \leq \mE\left( \mP_{0}\left(\frac 1{|\mathcal I|}\sum_{i\in \mathcal I}\frac{y_i^2}{k_i^2s^2}\geq F_{|\mathcal I|, \nu, 1-x} \mid \mathcal S,s\right) \right)
    = \mP_{0}\left(\frac 1{|\mathcal I|}\sum_{i\in \mathcal I}\frac{y_i^2}{k_i^2s^2}\geq F_{|\mathcal I|, \nu, 1-x}\geq \mid \mathcal S\right)\leq x, \nonumber 
\end{eqnarray}
where the expectation is over the distribution of $s$ (which is independent of $\mathcal S$), the first inequality in the second row follows from Lemma 1  of \cite{Mohamad2020}, and the last inequality follows  from the validity of the F test when the vector of means $\theta = 0$. The DCT procedures coincides with the closed F-test procedure on two-sided $p$-values suggested by  
\cite{finner1999stepwise}  for dFWER control. It complements the analysis for dFWER control suggested by \cite{finner1999stepwise} by providing also simultaneous post-hoc bounds for $n^+(\mathcal I)$ and $n^-(\mathcal I)$ for all $\mathcal I \subseteq [n].$

\section{Simultaneous inference on positive and non-positive parameters}\label{sec-partitioning}
 
\subsection{The general recipe for the simultaneous confidence lower bounds}\label{subsec-partitioning-cb-general}

We derive simultaneous $(1-\alpha)$-confidence lower bounds $\pl^+_\alpha(\mathcal{I})$ and $\pl_\alpha(\mathcal{I})$ for $n^+(\mathcal{I})$ and $n^-(\mathcal{I})+n^0(\mathcal I)$  by using the partitioning principle, reviewed in \S~\ref{sec-background}.  
The main idea 
is to partition the parameter space $\mathbb{R}^n$ into disjoint subspaces: consider the $2^n$ orthants
\begin{eqnarray*}\label{orthant_V}
\Theta_\mathcal{K} &=& \{ \theta \in \mathbb{R}^n: \theta_i > 0 \mathrm{\,\,for\,\,all\,\,} i \in \mathcal{K}, \theta_j \leq 0 \mathrm{\,\,for\,\,all\,\,} j \notin  \mathcal{K} \}
\end{eqnarray*}
for all $\mathcal{K}$, so that exactly one $\Theta_{\mathcal{K}}$ contains the true parameter $\theta$.

Each orthant $\Theta_\mathcal{K}$ has a corresponding null hypothesis $J_\mathcal{K}: \theta \in \Theta_\mathcal{K}$ that the orthant includes the true parameter. This hypothesis is true if and only if all  $K_i$, $i \in \mathcal{K}$ are true and all $H^-_j$, $j \notin \mathcal{K}$ are true:
\begin{eqnarray*}\label{eq-J_K}
J_{\mathcal{K}} : \{\bigcap_{i \in \mathcal{K}} K_i\} \cap \{\bigcap_{i \notin \mathcal{K}}H^-_i\}.
\end{eqnarray*}
Let $\psi_\mathcal{K}\in\{0,1\}$ be the level $\alpha$   
local test for hypotheses $J_{\mathcal{K}}$, for all $\mathcal{K}$.  The hypotheses $J_{\mathcal{K}}$ are all disjoint, and only one of them is true. 
For the true $J_{\mathcal{K}}$, $\displaystyle \sup_{\theta \in \Theta_\mathcal{K}} \mP_\theta (\psi_\mathcal{K} = 1)\leq \alpha$.

Consider the hypothesis $J_v(\mathcal{I}): n^+(\mathcal{I})=v$, which can be equivalently expressed as $J_v(\mathcal{I}):\bigcup_{\mathcal{K}:  |\mathcal{K} \cap \mathcal{I}|=v} J_{\mathcal{K}}$. Then $J_v(\mathcal{I})$ is rejected if and only if all $J_{\mathcal{K}}$ with $|\mathcal{K} \cap \mathcal{I}|=v$ are rejected, and  denote by
 \begin{eqnarray}\label{eq-psi_v_I}
 \psi_v(\mathcal{I}) = \min\{\psi_{\mathcal{K}}: \mathcal{K} \subseteq [n] \textrm{ and } |\mathcal{K} \cap \mathcal{I}|=v\}
\end{eqnarray} 
the statistical test for $J_v(\mathcal{I})$. The collection of values $v$ for which we failed to reject
$J_v(\mathcal{I})$ at level $\alpha$ constitutes a $(1-\alpha)$ confidence set for $n^+(\mathcal{I})$.

\begin{proposition}\label{prop-partitioning}
Let 
$\mathcal{N}^+(\mathcal{I})_\alpha = \{v \in \{0,\ldots, |\mathcal{I}|\}: \psi_{v}(\mathcal{I})=0 \}.
$
Then
$\mP_\theta ( n^+(\mathcal{I}) \in \mathcal{N}^+(\mathcal{I})_\alpha \mathrm{\,\,for\,\,all\,\,} \mathcal{I} ) \geq 1-\alpha$ for any $\theta \in \mathbb{R}^n$.
Furthermore, the lower  bounds given by
$\pl^+_{\alpha}(\mathcal{I}) = \min(\mathcal{N}^+(\mathcal{I})_\alpha), \quad 
\pl^-_{\alpha}(\mathcal{I}) =|\mathcal I|- \max(\mathcal{N}^+(\mathcal{I})_\alpha)
$
satisfy \eqref{eq-positiveonly-simultguarantee}.
\end{proposition}

The construction described above  is very general, for any valid local test. In particular, this means that if the $p$-values are dependent,  the confidence set can be computed as long as the local test used is valid for the dependency of the $n$ $p$-values combined for each orthant hypothesis. See  \S~\ref{sec-unconditionalParitioning} for an example with dependent $p$-values. 

Next, we show that by mapping the local tests used in DCT to local tests for testing orthant hypotheses, we can get a uniform improvement over DCT. 

\subsection{A uniform improvement over DCT for the bounds on $n^{+}(\mathcal{I}) $ and $n^-(\mathcal{I})+n^0(\mathcal I)$ }\label{subsec-simultaneous bounds}

We can use the DCT local tests, henceforth referred to as the {\it adaptive local tests}, for testing orthant hypothesis $J_{\mathcal{K}}$, as follows. Instead of testing this intersection hypothesis,  test the larger intersection hypothesis that only combines $p$-values from $\mathcal S^-$ and $\mathcal S^+$, i.e., test only the intersection of one-sided hypotheses selected in the first step of the DCT procedure. The larger intersection hypothesis has  indices in 
$\{\mathcal{K}^c \cap \mathcal{S}^{-}\} \cup \{\mathcal{K} \cap \mathcal{S}^{+}\}$. The adaptive local test is
\begin{eqnarray}\label{eq-psi_K}
\psi_\mathcal{K} &=& \phi_{\{\mathcal{K}^c \cap \mathcal{S}^-\} \cup \{\mathcal{K} \cap \mathcal{S}^+\}}\\
\nonumber
&=&\mathds{1}\{  f(\{2p_i, i\in \mathcal K^c \cap \mathcal S^- \}\cup \{2q_j, j\in \mathcal K \cap \mathcal S^+ \}) \leq \alpha \}.
\end{eqnarray}
This test is a valid local test for $J_{\mathcal{K}}$ since $$\big\{ \bigcap_{i\in \mathcal K^c}H_i^- \big\}  \cap \big\{\bigcap_{j\in \mathcal K}K_j \big\}\ \ \subseteq \ \ \big\{\bigcap_{i\in \mathcal K^c\cap \mathcal S^-}H_i^- \big\} \cap \big\{\bigcap_{j\in \mathcal K\cap \mathcal S^+}H^+_j \big\}.$$

Is this a good local test for an orthant? 
\cite{Zhao2019} argued that the power of a test of an intersection hypothesis that combines only the conditional $p$-values of $p$-values below a threshold has excellent power properties. Since for testing each orthant, we combine only conditional $p$-values for which the data favours the one-sided alternative, we expect  the adaptive local test to have the power advantage observed in \cite{Zhao2019}. Specifically, we expect  to have superior power over the non-adaptive local test that combines all $p$-values in the intersection when the parameter vector has mixed signs, and  not lose much power if the parameter vector does not have mixed signs. 

However, the computational problem remains: direct application of (\ref{eq-psi_v_I}) takes exponential time since there are $\binom{|\mathcal I|}{v}$ orthants with exactly $v$ positive parameters. For a combining function $f(\cdot)$ satisfying condition $(A0)$, the reduction to polynomial computation time is possible by first sorting the conditional $p$-values in each of the following sets: $\mathcal S^-\cap \mathcal I$, $\mathcal S^+\cap \mathcal I$, $\mathcal S^-\cap \mathcal I^c$, $\mathcal S^+\cap \mathcal I^c$. 
Consider the problem of  testing that $n^+(\mathcal I) = v$ and exactly  $k$ positive parameters are in $\mathcal S^+\cap \mathcal I$. Since $f(\cdot)$ is monotone, the orthant with the largest $p$-value is the one that combines: the $|\mathcal I\cap \mathcal S^-|-(v-k)$ largest $2p_i$'s in $\mathcal I\cap \mathcal S^-$; 
the $k$ largest $2q_i$'s in $\mathcal I\cap \mathcal S^+$; 
  between $0$ and $|\mathcal I^c\cap \mathcal S^-|$ largest $2p_i$'s in  $\mathcal I^c\cap \mathcal S^-$; 
between $0$ and $|\mathcal I^c\cap \mathcal S^+|$ largest  $2q_i$'s in  $\mathcal I^c\cap \mathcal S^+$.
So we can find the largest combination $p$-value by considering $|\mathcal S^-\cap \mathcal I^c|\times|\mathcal S^+\cap \mathcal I^c|$ possible combinations. Considering all possible values of $k$ and $v$, we see that the complexity is $\mathcal O(|\mathcal{I}|^2| \mathcal{I}^c|^2)$. See \S~\ref{sec-sm-shortcut} for the detailed algorithm. 

\begin{proposition}\label{prop-partitioning-properties}
Let $p_1,\ldots,p_n$ fulfil conditions $(A1)-(A2)$, and the combining function $f(\cdot)$ fulfils condition $(A0)$. 
Then using the adaptive local tests \eqref{eq-psi_K},  the partitioning procedure (AP)  returns 
the lower bounds $\pl^+_\alpha(\mathcal{I})$ and $\pl^-_\alpha(\mathcal{I})$ for $n^+(\mathcal{I})$ and $n^-(\mathcal{I})+n^0(\mathcal{I})$, respectively, with at most $\mathcal{O}(|\mathcal{I}|^2 \cdot \max(1,|\mathcal{I}^c|^2))$ computation. The bounds satisfy the following conditional (on the vector of signs, $\mathcal S$) and unconditional coverage guarantees, respectively: 
\begin{equation}\label{eq-positiveonly-simultguarantee-conditional}
  \mP_{\theta}\Big(
n^{+}(\mathcal{I}) 
\geq 
\pl_{\alpha}^{+}(\mathcal{I}),\,\,
n^{-}(\mathcal{I})+n^{0}(\mathcal{I}) 
\geq 
\pl_{\alpha}^{-}(\mathcal{I}),
 \mathrm{\,\,for\,\,all\,\,}\mathcal{I}  \mid \mathcal S \Big) \geq 1- \alpha,
\end{equation}
\begin{eqnarray}\label{eq-unconditionalcoverage}
    \mP_{\theta}\Big(
n^{+}(\mathcal{I}) 
\geq 
\pl_{\alpha}^{+}(\mathcal{I}),\,\,
n^{-}(\mathcal{I})+n^{0}(\mathcal{I}) 
\geq 
\pl_{\alpha}^{-}(\mathcal{I}),
 \mathrm{\,\,for\,\,all\,\,}\mathcal{I} \Big) \geq 1- (1-2^{-n}) \alpha > 1-\alpha.
\end{eqnarray}
\end{proposition}

The conditional $(1-\alpha)$  confidence bounds, i.e., the bounds  that satisfy \eqref{eq-positiveonly-simultguarantee-conditional},  have an unconditional $1-\alpha(1-2^{-n})$ confidence guarantee, i.e., they satisfy \eqref{eq-unconditionalcoverage}. Therefore, adaptive local tests with partitioning can be carried out at level $\tilde{\alpha}= \alpha/(1-2^{-n})$ if the conditional guarantee is not necessary. Using $\tilde \alpha$  falls within the framework of the holistic approach of \cite{goeman-arxiv}, that view selection and conditioning as means for useful inferences with unconditional error guarantees. 

With $(1-\alpha)$ confidence,  $\pl_{\tilde{\alpha}}^{+}(\mathcal{I})\leq n^+(\mathcal I)\leq |\mathcal I| - \pl_{\tilde{\alpha}}^{-}(\mathcal{I})$.  These are tighter bounds than with $(1-\tilde{\alpha})$ confidence. Specifically, the overall lower bounds for $n^+$ or $n^-+n^0$ are positive  if  
$f(\{2p_i, i\in  \mathcal S^- \}) \leq \tilde{\alpha}$ or $f( \{2q_j, j\in  \mathcal S^+ \}) \leq \tilde{\alpha}$, respectively.  In contrast, with DCT more testing takes place to correct for multiplicity, and moreover the local tests are carried out at level $\alpha<\tilde{\alpha}$. Therefore, there is less power to detect positive lower bounds for $n^+$ and $n^-+n^0$ with DCT. 
The number of discoveries for FWER control at level   $\alpha$ with partitioning may also be greater than with DCT, but only because it is carried out at level $\tilde \alpha>\alpha$. The multiplicity correction is essentially the same, since a discovery is made for $i\in [n]$ using partitioning only if all adaptive local tests that include it are rejected, and these are precisely  the local tests for all $\mathcal I\ni i$. By this same reasoning, the lower bounds  are the same when using partitioning  and DCT, with the same level $\alpha$ local tests, for $n^+(\mathcal I)$ if $\mathcal I\subseteq \mathcal S^-$,  and for $n^-(\mathcal I)+n^0(\mathcal I)$ if $\mathcal I\subseteq \mathcal S^+$.

\begin{proposition}\label{prop-AP-vs-DCT}
For a given combining function satisfying $(A0)$, 
let $\ell_{\alpha}^+(\cdot), \ell_{\alpha}^-(\cdot)$
be the DCT lower bounds that satisfy \eqref{eq-bounds_I_plus_minus}, 
and $\pl_{\alpha}^+(\cdot), \pl_{\alpha}^-(\cdot)$  be the bounds extracted from applying the level $\alpha$ adaptive local tests in \eqref{eq-psi_K}. 
Then 
$\pl_{\alpha}^+(\cdot)\geq \ell_{\alpha}^+(\cdot)$,  and  $\pl_{\alpha}^-(\cdot)\geq \ell_{\alpha}^-(\cdot)$. Moreover, if $\mathcal{I} \subseteq \mathcal{S}^{-}$ then $\pl_{\alpha}^+(\mathcal{I}) =  \ell_{\alpha}^+(\mathcal{I})$ and $\pl_{\alpha}^-(\mathcal{I})= \ell_{\alpha}^-(\mathcal{I})=0$;  if $\mathcal{I} \subseteq \mathcal{S}^{+}$ then $\pl_{\alpha}^-(\mathcal{I})= \ell_{\alpha}^-(\mathcal{I})$ and $\pl_{\alpha}^+(\mathcal{I}) =  \ell_{\alpha}^+(\mathcal{I})=0$.
\end{proposition}

\subsection{Inference for $n=2$}\label{sec-partitioning-example}

To enhance our understanding of the partitioning procedure with adaptive local tests, we examine  the basic case of $n=2$.  
Table \ref{tab:ppn2} presents the adaptive local tests used as a function of the selection event, while the top-left plot in Figure \ref{Figure:simesregions} displays the regions that provide inference about $n^+$ using Simes' combining function $f_{\mathrm{Simes}}$.

\begin{table}
\caption{\label{tab:ppn2} For each orthant hypothesis (rows), the adaptive local tests utilized by the partitioning procedure for the case of $n=2$ if the realized vector  of signs is $\mathcal S = (sign(p_1-1/2), sign(p_2-1/2))$ (columns).}
\centering
\begin{tabular}{c|c|c|c|c}
 & & & & \\
& $\mathcal S  = (-1,-1)$ & $\mathcal S  = (-1,1)$ & $\mathcal S  = (1,-1)$ & $\mathcal S  = (1,1)$\\
\hline
 & & & & \\
$H^-_1 \cap H^-_2$  & $f(2p_1,2p_2)$ & $2p_1$ & $2p_2$ & 1\\
$H^-_1 \cap K_2$  & $2p_1$ & $f(2p_1,2q_2)$ & 1 & $2q_2$\\
$K_1 \cap H^-_2$  & $2p_2$ & 1 & $f(2q_1,2p_2)$ & $2q_1 $\\
$K_1 \cap K_2$  & 1 & $2q_2$ & $2q_1$ & $f(2q_1,2q_2)$ \\
\end{tabular}
\end{table}

By examining the top row of Figure \ref{Figure:simesregions}, we can clearly observe the power gain 
for the partitioning procedure with adaptive local tests (top-left plot) compared to DCT (top-right plot).
One notable aspect is that the partitioning region expands significantly at the top-left and bottom-right quadrants, thereby enhancing the partitioning procedure's capability to detect both positive and non-positive effects. This improved performance is accompanied by the advantage of operating at a higher significance level, as determined by $\tilde\alpha = (1-(1-2^{-n}))\alpha = 4\alpha/3$.
This increase in power comes with the restriction that the partitioning procedure can only draw inferences about non-positive effects. For instance, if $p_1 > 1-\frac{\alpha}{4}$, the partitioning procedure would infer that $\theta_1 \leq 0$. In contrast, the DCT procedure can provide a more conclusive statement by inferring that $\theta_1 < 0$.

\begin{figure}
\centering
\begin{tikzpicture}[scale=5.5]


\def\alpgr{0.06666667}
\def\gam{0.1333333}
\def\alp{0.1040607}

\begin{scope}[yshift = -1.4 cm]

\draw (0.5,1) node[above]{Partitioning with adaptive tests};

\fill[fill=gray] (0,0) rectangle (1,\alpgr);
\fill[fill=gray] (0,0) rectangle (\alpgr,1);
\fill[fill=gray!30] (0,1) rectangle (1,1-\alpgr);
\fill[fill=gray!30] (1,1) rectangle (1-\alpgr,0);

\fill[fill=black] (0,0) rectangle (\gam,\gam);
\fill[fill=gray!10] (1,1) rectangle (1-\gam,1-\gam);
\fill[fill=gray!60] (0,1) rectangle (\gam,1-\gam);
\fill[fill=gray!60] (0,1) rectangle (\gam,1-\gam);
\fill[fill=gray!60] (1-\gam,0) rectangle (1,\gam);
\fill[fill=gray!60] (1-\gam,0) rectangle (1,\gam);

\fill[fill=gray] (\alpgr,1-\gam) rectangle (\gam,0.5);
\fill[fill=gray] (0.5,\alpgr) rectangle (1-\gam,\gam);
\fill[fill=gray!30] (\gam,1-\alpgr) rectangle (0.5,1-\gam);
\fill[fill=gray!30] (1-\gam,\gam) rectangle (1-\alpgr,0.5);

\draw (1,\gam) -- (1.01,\gam) node[right]{$2\alpha/3$};
\draw (1,\alpgr) -- (1.01,\alpgr) node[right]{$\alpha/3$};

\draw (0,0) -- (1,0) node[midway, below]{$p_1$};
\draw (0,0) -- (0,1) node[midway, left]{$p_2$};
\draw (0,0) rectangle (1,1);

\end{scope}


\def\alpgr{0.05}
\def\gam{0.1}
\def\alp{0.07804555}

\begin{scope}[yshift = -1.4 cm, xshift=1.4 cm]

\draw (0.5,1) node[above]{Directional closed testing};

\fill[fill=gray] (0,0) rectangle (1,\alpgr);
\fill[fill=gray] (0,0) rectangle (\alpgr,1);
\fill[fill=gray!30] (0,1) rectangle (1,1-\alpgr);
\fill[fill=gray!30] (1,1) rectangle (1-\alpgr,0);

\fill[fill=black] (0,0) rectangle (\gam,\gam);
\fill[fill=gray!10] (1,1) rectangle (1-\gam,1-\gam);

\fill[fill=gray!60] (0,1) rectangle (\gam,1-\gam);

\fill[fill=gray!60] (1-\gam,0) rectangle (1,\gam);

\draw (1,\gam) -- (1.01,\gam) node[right]{$\alpha/2$};
\draw (1,\alpgr) -- (1.01,\alpgr) node[right]{$\alpha/4$};

\draw (0,0) -- (1,0) node[midway, below]{$p_1$};
\draw (0,0) -- (0,1) node[midway, left]{$p_2$};

\draw (0,0) rectangle (1,1);


\end{scope}

\begin{scope}[yshift = -1.65 cm, xshift = 0cm]

\draw[fill=black] (0,0.05) rectangle (0.1,0);
\draw (0.25, 0.025) node{$n^+=2$};

\draw[fill=gray] (0.5,0.05) rectangle (0.6,0);
\draw (0.75, 0.025) node{$n^+\geq 1$};

\draw[fill=gray!60] (1,0.05) rectangle (1.1,0);
\draw (1.25, 0.025) node{$n^+= 1$};

\draw[fill=gray!30] (1.5,0.05) rectangle (1.6,0);
\draw (1.75, 0.025) node{$n^+\leq 1$};

\draw[fill=gray!10] (2,0.05) rectangle (2.1,0);
\draw (2.25, 0.025) node{$n^+ = 0$};
\end{scope}

\end{tikzpicture}
\def\alp{0.07804555}
\def\alpgr{0.05}
\def\gam{0.1}

\begin{tikzpicture}[scale=5.5]

\begin{scope}[yshift = -1.4 cm, xshift=1.4 cm]

\def\alpgr{0.05}
\def\gam{0.1}

\draw (0.5,1) node[above]{Directional closed testing};

\draw[pattern=dots] (\gam,0) rectangle (1-\gam,\alpgr);
\draw[pattern=dots] (0,\gam) rectangle (\alpgr,1-\gam);
\draw[pattern=dots] (\gam,1) rectangle (1-\gam,1-\alpgr);
\draw[pattern=dots] (1,1-\gam) rectangle (1-\alpgr,\gam);

 \draw[pattern=north east lines] (0,0) rectangle (\gam,\gam);
\draw[pattern=north east lines] (1,1) rectangle (1-\gam,1-\gam);
\draw[pattern=north east lines] (0,1) rectangle (\gam,1-\gam);
\draw[pattern=north east lines] (1-\gam,0) rectangle (1,\gam);

\draw (1,\gam) -- (1.01,\gam) node[right]{$\alpha/2$};
\draw (1,\alpgr) -- (1.01,\alpgr) node[right]{$\alpha/4$};

\draw (0,0) -- (1,0) node[midway, below]{$p_1$};
\draw (0,0) -- (0,1) node[midway, left]{$p_2$};


 \draw (\gam/2, \gam/2) node[pin=45:{$H_1^-,H_2^-$}]{};
 \draw (\gam/2, 1-\gam/2) node[pin=315:{$H_1^-,H_2^+$}]{};
 \draw (1-\gam/2, 1-\gam/2) node[pin=225:{$H_1^+,H_2^+$}]{};
 \draw (1-\gam/2, \gam/2) node[pin=135:{$H_1^+,H_2^-$}]{};
 \draw (0, 0.5) node[pin=right:{$H_1^-$}]{};
 \draw (1, 0.5) node[pin=left:{$H_1^+$}]{};
 \draw (0.5, 0) node[pin=above:{$H_2^-$}]{};
 \draw (0.5, 1) node[pin=below:{$H_2^+$}]{};

\end{scope}


\begin{scope}[yshift = -1.4 cm]

\def\alpgr{0.06666667}
\def\gam{0.1333333}

\draw (0.5,1) node[above]{Partitioning with adaptive tests};

\draw[pattern=dots] (\gam,0) rectangle (1-\gam,\alpgr);
\draw[pattern=dots] (0,\gam) rectangle (\alpgr,1-\gam);
\draw[pattern=dots] (\gam,1) rectangle (1-\gam,1-\alpgr);
\draw[pattern=dots] (1,1-\gam) rectangle (1-\alpgr,\gam);

 \draw[pattern=north east lines] (0,0) rectangle (\gam,\gam);
\draw[pattern=north east lines] (1,1) rectangle (1-\gam,1-\gam);
\draw[pattern=north east lines] (0,1) rectangle (\gam,1-\gam);
\draw[pattern=north east lines] (1-\gam,0) rectangle (1,\gam);

\draw (1,\gam) -- (1.01,\gam) node[right]{$2\alpha/3$};
\draw (1,\alpgr) -- (1.01,\alpgr) node[right]{$\alpha/3$};

\draw (0,0) -- (1,0) node[midway, below]{$p_1$};
\draw (0,0) -- (0,1) node[midway, left]{$p_2$};


 \draw (\gam/2, \gam/2) node[pin=45:{$H_1^-,H_2^-$}]{};
 \draw (\gam/2, 1-\gam/2) node[pin=315:{$H_1^-,K_2$}]{};
 \draw (1-\gam/2, 1-\gam/2) node[pin=225:{$K_1,K_2$}]{};
 \draw (1-\gam/2, \gam/2) node[pin=135:{$K_1,H_2^-$}]{};
 \draw (0, 0.5) node[pin=right:{$H_1^-$}]{};
 \draw (1, 0.5) node[pin=left:{$K_1$}]{};
 \draw (0.5, 0) node[pin=above:{$H_2^-$}]{};
 \draw (0.5, 1) node[pin=below:{$K_2$}]{};

\end{scope}

\end{tikzpicture}
\caption{
The top row illustrates regions leading to inference about $n^+$ with $(1-\alpha)$ confidence for the partitioning procedure with adaptive local tests (top-left plot) and the directional closed testing procedure  (top-right plot) with Simes' combining function. The gray scale indicates the type of inference. 
The bottom row illustrates regions leading to the rejection of base hypotheses for partitioning with adaptive local tests (bottom left plot) and directional closed testing (bottom right plot). 
The pattern indicates the number of hypotheses rejected at level $\alpha$: one (dots) or two (diagonal lines).
Each plot is based on $\alpha=0.2$.}
\label{Figure:simesregions}
\end{figure}
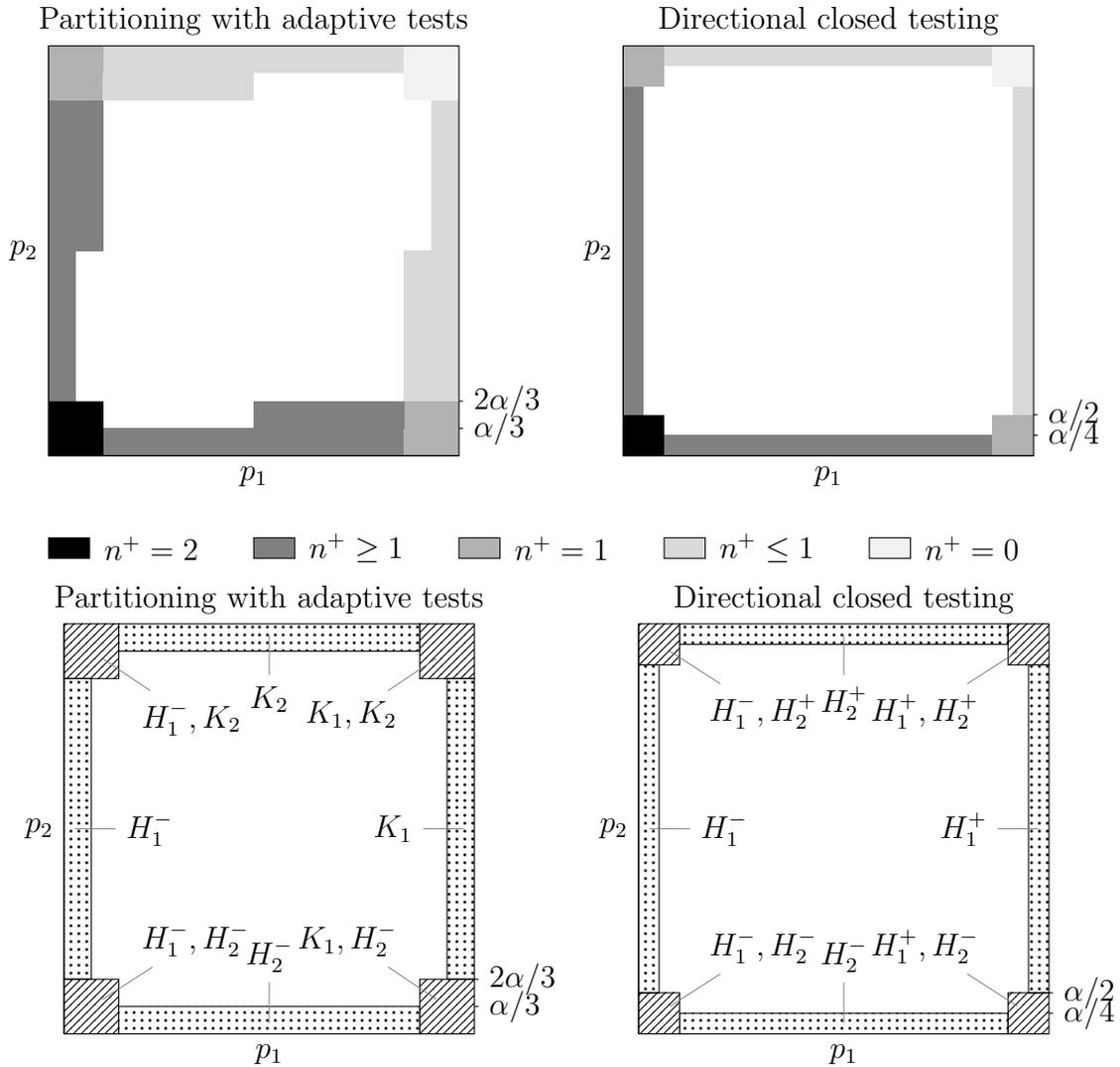

The bottom-left plot of Figure \ref{Figure:simesregions} displays the rejection region for the base hypotheses: $H_1^-$, $H_2^-$, $K_1$, and $K_2$. Notably, if we observe $(p_1,p_2)$ such that $1/2<p_1<1-2\alpha/3$ and $\alpha/3<p_2<2\alpha/3$, we refrain from rejecting any base hypothesis, even if we can infer that $n^+ \geq 1$.
In contrast, in DCT using Simes' combining function (bottom-right plot), whenever there is an informative inference about $n^+$, it is accompanied by the rejection of base hypotheses $H_1^-$, $H_2^-$, $K_1$, and $K_2$. This is not the case for $n>2$ or by using Fisher's combining function for $n=2$ (see \S~\ref{Appendix-moren2} for this and the comparison with other 
FWER controlling procedures). 
The shape of the region in the left plot of Figure \ref{Figure:simesregions} is non-monotone, in the sense that a larger pair of $p$-values may result in a larger lower bound on $n^+$. 
For example, if we observe $(p_1,p_2)$ such that $1/2<p_1<1-2\alpha/3$ and $\alpha/3<p_2<2\alpha/3$, we would conclude that $n^+\geq 1$ by rejecting $H_1^- \cap H_2^-$ because $2p_2 \leq \tilde\alpha$ (see column 3 in Table \ref{tab:ppn2}). However, if $p_1$ decreases to $\alpha/3< p_1< 1/2$, we would not draw any inference. This outcome is rather counter-intuitive since a smaller value of $p_1$ corresponds to stronger evidence supporting $\theta_1 > 0$. 
This is a result of employing adaptive tests rather than non-adaptive ones, i.e. of performing selection. The procedure is nevertheless attractive, since  conditional on the selection (of which $p$-values of one-sided hypotheses to combine), the rejection region is monotone.

Figure \ref{Figure:power2} shows the power advantage of  partitioning over DCT, 
for a range of parameters $(\theta_1,\theta_2)$ and $\alpha=5\%$. 
The $p$-values are derived from independent test statistics $\hat{\theta}_i$ following a Gaussian distribution $N(\theta_i, 1)$, $i = 1, 2$.
The non-trivial bounds probability (left panel) refers  to the probability that $(p_1, p_2)$ lies in the non-white area of the top plots in Figure \ref{Figure:simesregions}. The expected number of discoveries (right panel) is the weighted average of the number of rejections displayed in the bottom plots of Figure \ref{Figure:simesregions}, with weights determined by their respective probabilities. .
The gain is substantial, particularly in terms of the non-trivial bounds probability when the parameters have opposite signs. For example, with $(\theta_1,\theta_2) = (2,-1.5)$, the DCT has a probability of 54.1\%, while partitioning has 71.6\%, resulting in a difference of 17.5\%. This difference is nearly triple the value of $61.7\% - 55.1\% = 6.6\%$ observed when both effects are positive, with $(\theta_1,\theta_2) = (2,1.5)$.

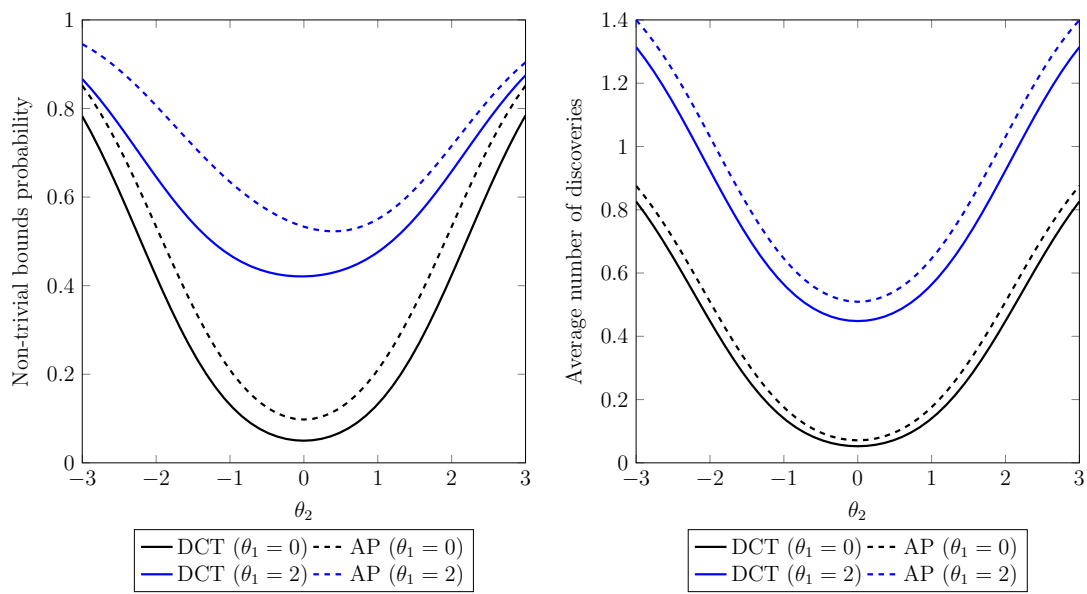
\begin{figure}
\centering
\begin{tikzpicture}[scale=0.7]
	\begin{axis}[
	xmin = -3,
	xmax = 3,
	ymin = 0,
	ymax = 1,
	ylabel=Non-trivial bounds probability,
	xlabel=$\theta_2$,
        height=10cm,
	width=10cm,
		legend style={at={(0.5,-0.15)},
		anchor=north,legend columns=2}
	]
	\addplot[color=black, very thick] coordinates {
( -3 , 0.7825 ) ( -2.88 , 0.746 ) ( -2.76 , 0.7067 ) ( -2.64 , 0.6647 ) ( -2.52 , 0.6208 ) ( -2.4 , 0.5753 ) ( -2.28 , 0.5289 ) ( -2.16 , 0.4823 ) ( -2.04 , 0.4361 ) ( -1.92 , 0.391 ) ( -1.8 , 0.3476 ) ( -1.68 , 0.3064 ) ( -1.56 , 0.2679 ) ( -1.44 , 0.2324 ) ( -1.32 , 0.2002 ) ( -1.2 , 0.1714 ) ( -1.08 , 0.146 ) ( -0.96 , 0.124 ) ( -0.84 , 0.1052 ) ( -0.72 , 0.0896 ) ( -0.6 , 0.0769 ) ( -0.48 , 0.0668 ) ( -0.36 , 0.0593 ) ( -0.24 , 0.054 ) ( -0.12 , 0.051 ) ( 0 , 0.05 ) ( 0.12 , 0.0511 ) ( 0.24 , 0.0543 ) ( 0.36 , 0.0597 ) ( 0.48 , 0.0675 ) ( 0.6 , 0.0777 ) ( 0.72 , 0.0906 ) ( 0.84 , 0.1065 ) ( 0.96 , 0.1254 ) ( 1.08 , 0.1476 ) ( 1.2 , 0.1732 ) ( 1.32 , 0.2022 ) ( 1.44 , 0.2347 ) ( 1.56 , 0.2703 ) ( 1.68 , 0.309 ) ( 1.8 , 0.3503 ) ( 1.92 , 0.3937 ) ( 2.04 , 0.4389 ) ( 2.16 , 0.4851 ) ( 2.28 , 0.5316 ) ( 2.4 , 0.5779 ) ( 2.52 , 0.6233 ) ( 2.64 , 0.6672 ) ( 2.76 , 0.7089 ) ( 2.88 , 0.7481 ) ( 3 , 0.7844 )
};
 \addlegendentry{DCT ($\theta_1=0$)}
\addplot[color=black, dashed, very thick] coordinates {
( -3 , 0.8517 ) ( -2.88 , 0.8231 ) ( -2.76 , 0.7913 ) ( -2.64 , 0.7564 ) ( -2.52 , 0.7188 ) ( -2.4 , 0.6786 ) ( -2.28 , 0.6364 ) ( -2.16 , 0.5927 ) ( -2.04 , 0.5481 ) ( -1.92 , 0.5032 ) ( -1.8 , 0.4586 ) ( -1.68 , 0.4149 ) ( -1.56 , 0.3727 ) ( -1.44 , 0.3326 ) ( -1.32 , 0.2951 ) ( -1.2 , 0.2604 ) ( -1.08 , 0.2289 ) ( -0.96 , 0.2007 ) ( -0.84 , 0.176 ) ( -0.72 , 0.1548 ) ( -0.6 , 0.1371 ) ( -0.48 , 0.1228 ) ( -0.36 , 0.1118 ) ( -0.24 , 0.104 ) ( -0.12 , 0.0993 ) ( 0 , 0.0978 ) ( 0.12 , 0.0993 ) ( 0.24 , 0.104 ) ( 0.36 , 0.1118 ) ( 0.48 , 0.1228 ) ( 0.6 , 0.1371 ) ( 0.72 , 0.1548 ) ( 0.84 , 0.176 ) ( 0.96 , 0.2007 ) ( 1.08 , 0.2289 ) ( 1.2 , 0.2604 ) ( 1.32 , 0.2951 ) ( 1.44 , 0.3326 ) ( 1.56 , 0.3727 ) ( 1.68 , 0.4149 ) ( 1.8 , 0.4586 ) ( 1.92 , 0.5032 ) ( 2.04 , 0.5481 ) ( 2.16 , 0.5927 ) ( 2.28 , 0.6364 ) ( 2.4 , 0.6786 ) ( 2.52 , 0.7188 ) ( 2.64 , 0.7564 ) ( 2.76 , 0.7913 ) ( 2.88 , 0.8231 ) ( 3 , 0.8517 )
};
 \addlegendentry{AP ($\theta_1=0$)}
	\addplot[color=blue, very thick] coordinates {
( -3 , 0.8666 ) ( -2.88 , 0.8443 ) ( -2.76 , 0.8202 ) ( -2.64 , 0.7945 ) ( -2.52 , 0.7676 ) ( -2.4 , 0.7398 ) ( -2.28 , 0.7115 ) ( -2.16 , 0.683 ) ( -2.04 , 0.6548 ) ( -1.92 , 0.6273 ) ( -1.8 , 0.6008 ) ( -1.68 , 0.5757 ) ( -1.56 , 0.5522 ) ( -1.44 , 0.5306 ) ( -1.32 , 0.511 ) ( -1.2 , 0.4935 ) ( -1.08 , 0.4781 ) ( -0.96 , 0.4647 ) ( -0.84 , 0.4534 ) ( -0.72 , 0.4439 ) ( -0.6 , 0.4363 ) ( -0.48 , 0.4303 ) ( -0.36 , 0.4258 ) ( -0.24 , 0.4228 ) ( -0.12 , 0.4212 ) ( 0 , 0.4209 ) ( 0.12 , 0.4219 ) ( 0.24 , 0.4243 ) ( 0.36 , 0.428 ) ( 0.48 , 0.4332 ) ( 0.6 , 0.44 ) ( 0.72 , 0.4485 ) ( 0.84 , 0.4589 ) ( 0.96 , 0.4711 ) ( 1.08 , 0.4854 ) ( 1.2 , 0.5018 ) ( 1.32 , 0.5202 ) ( 1.44 , 0.5406 ) ( 1.56 , 0.563 ) ( 1.68 , 0.5871 ) ( 1.8 , 0.6127 ) ( 1.92 , 0.6395 ) ( 2.04 , 0.6673 ) ( 2.16 , 0.6955 ) ( 2.28 , 0.7238 ) ( 2.4 , 0.7518 ) ( 2.52 , 0.7791 ) ( 2.64 , 0.8053 ) ( 2.76 , 0.8302 ) ( 2.88 , 0.8535 ) ( 3 , 0.875 )};
\addlegendentry{DCT ($\theta_1=2$)}
 \addplot[color=blue, dashed, very thick] coordinates {
( -3 , 0.9454 ) ( -2.88 , 0.9339 ) ( -2.76 , 0.9208 ) ( -2.64 , 0.9062 ) ( -2.52 , 0.8901 ) ( -2.4 , 0.8726 ) ( -2.28 , 0.8539 ) ( -2.16 , 0.8341 ) ( -2.04 , 0.8134 ) ( -1.92 , 0.792 ) ( -1.8 , 0.7703 ) ( -1.68 , 0.7485 ) ( -1.56 , 0.7269 ) ( -1.44 , 0.7056 ) ( -1.32 , 0.6849 ) ( -1.2 , 0.6651 ) ( -1.08 , 0.6462 ) ( -0.96 , 0.6284 ) ( -0.84 , 0.6118 ) ( -0.72 , 0.5964 ) ( -0.6 , 0.5824 ) ( -0.48 , 0.5697 ) ( -0.36 , 0.5584 ) ( -0.24 , 0.5485 ) ( -0.12 , 0.54 ) ( 0 , 0.5331 ) ( 0.12 , 0.5279 ) ( 0.24 , 0.5244 ) ( 0.36 , 0.5227 ) ( 0.48 , 0.5231 ) ( 0.6 , 0.5255 ) ( 0.72 , 0.5301 ) ( 0.84 , 0.5371 ) ( 0.96 , 0.5465 ) ( 1.08 , 0.5583 ) ( 1.2 , 0.5725 ) ( 1.32 , 0.5889 ) ( 1.44 , 0.6076 ) ( 1.56 , 0.6282 ) ( 1.68 , 0.6504 ) ( 1.8 , 0.6741 ) ( 1.92 , 0.6987 ) ( 2.04 , 0.724 ) ( 2.16 , 0.7495 ) ( 2.28 , 0.7748 ) ( 2.4 , 0.7995 ) ( 2.52 , 0.8234 ) ( 2.64 , 0.846 ) ( 2.76 , 0.8671 ) ( 2.88 , 0.8866 ) ( 3 , 0.9043 )
};
 \addlegendentry{AP ($\theta_1=2$)}

	\end{axis}
\end{tikzpicture}
\begin{tikzpicture}[scale=0.7]
	\begin{axis}[
	xmin = -3,
	xmax = 3,
	ymin = 0,
	ymax = 1.4,
	ylabel=Average number of discoveries,
	xlabel=$\theta_2$,
        height=10cm,
	width=10cm,
		legend style={at={(0.5,-0.15)},
		anchor=north,legend columns=2}
	]
	\addplot[color=black, very thick] coordinates {
( -3 , 0.826 ) ( -2.88 , 0.7881 ) ( -2.76 , 0.7472 ) ( -2.64 , 0.7035 ) ( -2.52 , 0.6577 ) ( -2.4 , 0.6101 ) ( -2.28 , 0.5615 ) ( -2.16 , 0.5126 ) ( -2.04 , 0.4641 ) ( -1.92 , 0.4166 ) ( -1.8 , 0.3708 ) ( -1.68 , 0.3272 ) ( -1.56 , 0.2863 ) ( -1.44 , 0.2486 ) ( -1.32 , 0.2143 ) ( -1.2 , 0.1835 ) ( -1.08 , 0.1563 ) ( -0.96 , 0.1327 ) ( -0.84 , 0.1125 ) ( -0.72 , 0.0957 ) ( -0.6 , 0.0819 ) ( -0.48 , 0.071 ) ( -0.36 , 0.0627 ) ( -0.24 , 0.057 ) ( -0.12 , 0.0536 ) ( 0 , 0.0525 ) ( 0.12 , 0.0536 ) ( 0.24 , 0.057 ) ( 0.36 , 0.0627 ) ( 0.48 , 0.071 ) ( 0.6 , 0.0819 ) ( 0.72 , 0.0957 ) ( 0.84 , 0.1125 ) ( 0.96 , 0.1327 ) ( 1.08 , 0.1563 ) ( 1.2 , 0.1835 ) ( 1.32 , 0.2143 ) ( 1.44 , 0.2486 ) ( 1.56 , 0.2863 ) ( 1.68 , 0.3272 ) ( 1.8 , 0.3708 ) ( 1.92 , 0.4166 ) ( 2.04 , 0.4641 ) ( 2.16 , 0.5126 ) ( 2.28 , 0.5615 ) ( 2.4 , 0.6101 ) ( 2.52 , 0.6577 ) ( 2.64 , 0.7035 ) ( 2.76 , 0.7472 ) ( 2.88 , 0.7881 ) ( 3 , 0.826 )
};
 \addlegendentry{DCT ($\theta_1=0$)}
\addplot[color=black, dashed, very thick] coordinates {
( -3 , 0.8756 ) ( -2.88 , 0.8409 ) ( -2.76 , 0.8028 ) ( -2.64 , 0.7616 ) ( -2.52 , 0.7177 ) ( -2.4 , 0.6715 ) ( -2.28 , 0.6236 ) ( -2.16 , 0.5747 ) ( -2.04 , 0.5254 ) ( -1.92 , 0.4766 ) ( -1.8 , 0.4287 ) ( -1.68 , 0.3826 ) ( -1.56 , 0.3387 ) ( -1.44 , 0.2976 ) ( -1.32 , 0.2596 ) ( -1.2 , 0.2251 ) ( -1.08 , 0.1942 ) ( -0.96 , 0.1669 ) ( -0.84 , 0.1434 ) ( -0.72 , 0.1234 ) ( -0.6 , 0.1069 ) ( -0.48 , 0.0938 ) ( -0.36 , 0.0837 ) ( -0.24 , 0.0767 ) ( -0.12 , 0.0725 ) ( 0 , 0.0711 ) ( 0.12 , 0.0725 ) ( 0.24 , 0.0767 ) ( 0.36 , 0.0837 ) ( 0.48 , 0.0938 ) ( 0.6 , 0.1069 ) ( 0.72 , 0.1234 ) ( 0.84 , 0.1434 ) ( 0.96 , 0.1669 ) ( 1.08 , 0.1942 ) ( 1.2 , 0.2251 ) ( 1.32 , 0.2596 ) ( 1.44 , 0.2976 ) ( 1.56 , 0.3387 ) ( 1.68 , 0.3826 ) ( 1.8 , 0.4287 ) ( 1.92 , 0.4766 ) ( 2.04 , 0.5254 ) ( 2.16 , 0.5747 ) ( 2.28 , 0.6236 ) ( 2.4 , 0.6715 ) ( 2.52 , 0.7177 ) ( 2.64 , 0.7616 ) ( 2.76 , 0.8028 ) ( 2.88 , 0.8409 ) ( 3 , 0.8756 )
};
 \addlegendentry{AP ($\theta_1=0$)}
	\addplot[color=blue, very thick] coordinates {
( -3 , 1.314 ) ( -2.88 , 1.2773 ) ( -2.76 , 1.2369 ) ( -2.64 , 1.1932 ) ( -2.52 , 1.1466 ) ( -2.4 , 1.0975 ) ( -2.28 , 1.0465 ) ( -2.16 , 0.9944 ) ( -2.04 , 0.9417 ) ( -1.92 , 0.8894 ) ( -1.8 , 0.838 ) ( -1.68 , 0.7883 ) ( -1.56 , 0.7409 ) ( -1.44 , 0.6964 ) ( -1.32 , 0.6552 ) ( -1.2 , 0.6176 ) ( -1.08 , 0.5838 ) ( -0.96 , 0.554 ) ( -0.84 , 0.5281 ) ( -0.72 , 0.5061 ) ( -0.6 , 0.4879 ) ( -0.48 , 0.4733 ) ( -0.36 , 0.4621 ) ( -0.24 , 0.4543 ) ( -0.12 , 0.4496 ) ( 0 , 0.4481 ) ( 0.12 , 0.4496 ) ( 0.24 , 0.4543 ) ( 0.36 , 0.4621 ) ( 0.48 , 0.4733 ) ( 0.6 , 0.4879 ) ( 0.72 , 0.5061 ) ( 0.84 , 0.5281 ) ( 0.96 , 0.554 ) ( 1.08 , 0.5838 ) ( 1.2 , 0.6176 ) ( 1.32 , 0.6552 ) ( 1.44 , 0.6964 ) ( 1.56 , 0.7409 ) ( 1.68 , 0.7883 ) ( 1.8 , 0.838 ) ( 1.92 , 0.8894 ) ( 2.04 , 0.9417 ) ( 2.16 , 0.9944 ) ( 2.28 , 1.0465 ) ( 2.4 , 1.0975 ) ( 2.52 , 1.1466 ) ( 2.64 , 1.1932 ) ( 2.76 , 1.2369 ) ( 2.88 , 1.2773 ) ( 3 , 1.314 )
};
\addlegendentry{DCT ($\theta_1=2$)}
 \addplot[color=blue, dashed, very thick] coordinates {
( -3 , 1.3997 ) ( -2.88 , 1.367 ) ( -2.76 , 1.3306 ) ( -2.64 , 1.2905 ) ( -2.52 , 1.247 ) ( -2.4 , 1.2004 ) ( -2.28 , 1.1514 ) ( -2.16 , 1.1004 ) ( -2.04 , 1.0481 ) ( -1.92 , 0.9953 ) ( -1.8 , 0.9426 ) ( -1.68 , 0.8909 ) ( -1.56 , 0.8408 ) ( -1.44 , 0.793 ) ( -1.32 , 0.7481 ) ( -1.2 , 0.7064 ) ( -1.08 , 0.6685 ) ( -0.96 , 0.6345 ) ( -0.84 , 0.6046 ) ( -0.72 , 0.5788 ) ( -0.6 , 0.5572 ) ( -0.48 , 0.5397 ) ( -0.36 , 0.5262 ) ( -0.24 , 0.5167 ) ( -0.12 , 0.511 ) ( 0 , 0.5091 ) ( 0.12 , 0.511 ) ( 0.24 , 0.5167 ) ( 0.36 , 0.5262 ) ( 0.48 , 0.5397 ) ( 0.6 , 0.5572 ) ( 0.72 , 0.5788 ) ( 0.84 , 0.6046 ) ( 0.96 , 0.6345 ) ( 1.08 , 0.6685 ) ( 1.2 , 0.7064 ) ( 1.32 , 0.7481 ) ( 1.44 , 0.793 ) ( 1.56 , 0.8408 ) ( 1.68 , 0.8909 ) ( 1.8 , 0.9426 ) ( 1.92 , 0.9953 ) ( 2.04 , 1.0481 ) ( 2.16 , 1.1004 ) ( 2.28 , 1.1514 ) ( 2.4 , 1.2004 ) ( 2.52 , 1.247 ) ( 2.64 , 1.2905 ) ( 2.76 , 1.3306 ) ( 2.88 , 1.367 ) ( 3 , 1.3997 )
};
 \addlegendentry{AP ($\theta_1=2$)}

	\end{axis}
\end{tikzpicture}
\caption{Probability of observing non-trivial bounds (left plot) and 
the expected number of discoveries (right plot)  versus $\theta_2$ with directional closed testing (DCT, solid line) and  adaptive local tests with partitioning (AP, dashed line) at the confidence level of 95\%. Black curves correspond to $\theta_1 = 0$, and blue curves to $\theta_1 = 2$. 
}
\label{Figure:power2}
\end{figure}

\section{Applications}\label{sec - Zhao}

\subsection{Subgroup Analysis}\label{subsec-smallexample}
In subgroup analysis, it is common to apply two-sided tests in order to conclude about the  intervention effect in each subgroup. An easier, yet informative goal, is to provide lower bounds on $n^+$ and either $n^-$ or $n^-+n^0$ . 
The data of \cite{Fisher83}, re-analysed in \cite{gail1985testing}, consists of the difference in disease-free survival probabilities at 3 years for breast cancer patients who underwent a PFT treatment (L-Phenylalanine mustard, 5-Fluorouracil, and Tamoxifen) compared to those who received a PF treatment (L-Phenylalanine mustard and 5-Fluorouracil). The study includes a total of 1260 patients divided in $n=4$ subgroups, classified based on age and progesterone receptor levels. Let $m_i^j$ be the sample size for treatment $j \in \{\mathrm{PFT},\mathrm{PF}\}$ in the $i$th subgroup. 
Table \ref{tab:Gail} provides $\hat{\pi}_i^{j}$, the Kaplan-Meier estimate  of $\pi_i^{j}$, the disease-free survival probability at 3 years for patients who underwent treatment $j$ and belong to subgroup $i$. 
The corresponding standard error $\mathrm{SE}(\hat{\pi}^j_{i})$
is calculated using Greenwood's formula \citep{greenwood1926report}. 

However, the variable $\sqrt{m^{j}_i}(\hat{\pi}^j_{i} - \pi^j_{i})$ converges in distribution to a Normal variable with a mean of 0 and a variance that depends on the (unknown) survival probability $\pi^j_{i}$, which simplifies to $\pi_i^{j}(1-\pi_i^{j})$ if there is no censoring.  For censored data, \cite{cutler1958maximum} suggested replacing $m^j_i$ with the effective sample size $\tilde{m}^j_i = \hat{\pi}^j_i (1-\hat{\pi}^j_i)/ \mathrm{SE}^2(\hat{\pi}^j_{i})$. \cite{anderson1982approximate} showed that $\sqrt{\tilde{m}^j_i}(\hat{\pi}^j_i - \pi^j_i)$ converges in distribution to $N(0, \pi^j_i(1-\pi^j_i))$. 
This motivates the arcsine-square root transformation $g(\hat{\pi}^j_i) = \arcsin{ \sqrt{ \hat{\pi}^j_i } }$: by the delta method, the asymptotic variance of $\sqrt{\tilde{m}^{j}_i} \{g(\hat{\pi}^j_i) - g(\pi^j_i)\}$ is  
$1/4$, making it an effective variance-stabilizing transformation.  
Then, the large-sample normal approximation for the standardized difference $\hat{\theta}_i=  \{ g(\hat{\pi}^{\mathrm{PFT}}_{i}) - g(\hat{\pi}^{\mathrm{PF}}_{i}) \} /\sqrt{ 1/(4 \tilde{m}^{\mathrm{PFT}}_i) + 1/(4 \tilde{m}^{\mathrm{PF}}_i) }$  can be used to calculate the $p$-value $p_i = 1-\Phi(\hat{\theta}_i)$ for $H_i^-$, where $\Phi(\cdot)$ is the $N(0,1)$ CDF \citep{klein2007analyzing}.

 \begin{table}
     \caption{\label{tab:Gail} Analysis of  disease-free survival probabilities at 3 years $\hat{\pi}_{i}^j$ for breast cancer patients who underwent treatment $j$th treatment ($j=$PFT, PF) and belong to subgroup $i$ defined by age and progesterone receptor (PR) levels ($i=1$: Age $< 50$, PR $<10$; $i=2$: Age $\geq 50$, PR $<10$; $i=3$: Age $< 50$, PR $\geq 10$; $i=4$: Age $\geq  50$, PR $\geq 10$). 
 Reproduced from Table 2 in \cite{gail1985testing}.
 The last two rows are the (standardized) differences between  PFT and PF disease-free survival probabilities at 3 years (arcsin-square root transformed) and the  $p$-value for $H_i^-$ computed using the large sample normal approximation.}
 \begin{tabular}{c cc|cc|cc|cc}
 \hline
 \\
Subgroup $i$ & \multicolumn{2}{c}{1} &  \multicolumn{2}{c}{2} &  \multicolumn{2}{c}{3} &  \multicolumn{2}{c}{4} \\
 \\
Treatment $j$ & PFT & PF & PFT & PF & PFT & PF & PFT & PF \\
$\hat{\pi}_{i}^j$ & .599 & .436 & .526 & .639 & .651 & .698 & .639 & .790 \\
$\mathrm{SE}(\hat{\pi}^j_{i})$ &.0542 & .0572 &  .0510 & .0463 & .0431 & .0438 & .0386 & .0387 \\
\\
$\hat{\theta}_i$ & \multicolumn{2}{c}{2.051} & \multicolumn{2}{c}{-1.635} & \multicolumn{2}{c}{-0.764} & \multicolumn{2}{c}{-2.708} \\
$p_i$ & \multicolumn{2}{c}{.0202} & \multicolumn{2}{c}{.9490} & \multicolumn{2}{c}{.7774} & \multicolumn{2}{c}{.9966} \\
  \\
  
  \hline
 \end{tabular}
\end{table}

 Applying the first step of hypothesis selection based on the direction favored by the data,  the DCT procedure selects $H_1^-, H_2^+, H_3^+, H_4^+$. The $p$-values used for testing in the second step, are therefore: $0.0202, 1-0.9490, 1-0.7774$, and $1-0.9966$, for groups 1, 2, 3, and 4, respectively.  The procedure is 
  illustrated in  the top panel of Figure \ref{Figure:Gail_CT}. 
 
 Applying DCT at level $\alpha = 0.05$, we have  $0 \leq n^+$ and $1 \leq n^- $ with 95\% confidence. Furthermore, we conclude that $\theta_4 < 0$, since $H^+_4$ is rejected by closed testing. We obtain a trivial lower bound for $n^+$ because $H^-_1$ is not rejected by closed testing, and a non-trivial lower bound for $n^-$ because $H^+_2 \cap H^+_3 \cap H_4^+$ is rejected by closed testing.

 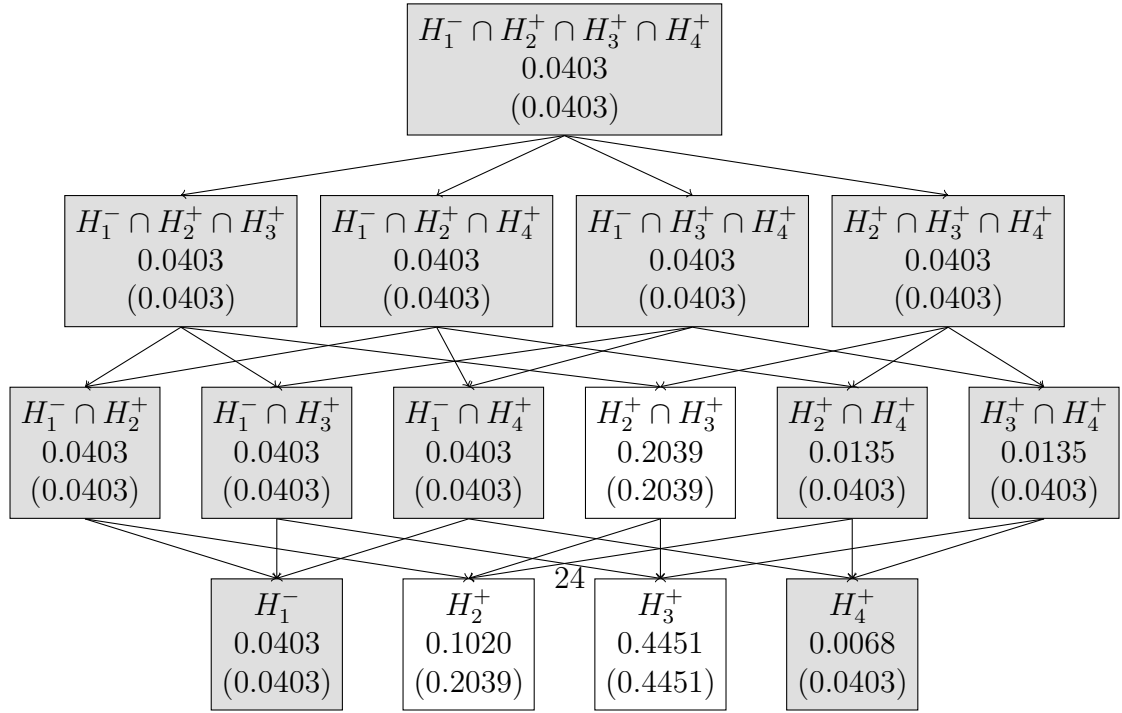
\begin{figure}
 \centering
 \begin{tikzpicture}[scale=.85]
\path (7,25.5) node {Procedure 3.1 (Directional closed testing)};
  \path (7,24) node[draw, fill= gray!25, align=center] (1234) {  $H_{1}^-\cap H_{2}^+ \cap H_{3}^+ \cap H_{4}^+$ \\ $0.0271$ \\  $(0.0271)$};
  \path (1,21) node[draw, align=center] (123) {  $H_{1}^-\cap H_{2}^+ \cap H_{3}^+$ \\ $0.1209$ \\ $(0.1158)$};
  \path (5,21) node[draw, align=center, fill= gray!25] (124) {   $H_{1}^-\cap H_{2}^+  \cap H_{4}^+$ \\ $0.0203$ \\ $(0.0271)$};
  \path (9,21) node[draw, align=center, fill= gray!25] (134) {   $H_{1}^-\cap  H_{3}^+ \cap H_{4}^+$ \\ $0.0203$ \\ $(0.0271 )$};
  \path (13,21) node[draw, align=center, fill= gray!25] (234) { $H_{2}^+ \cap H_{3}^+ \cap H_{4}^+$ \\ $0.0203$\\ $(0.0271)$};
  \path (-0.5,18) node[draw, align=center] (12) { $H_{1}^-\cap H_{2}^+$ \\ $0.0806$\\ $(0.1209)$};
  \path (2.5,18) node[draw, align=center] (13) {  $H_{1}^-\cap H_{3}^+$ \\ $0.0806$ \\ $(0.1209)$};
  \path (5.5,18) node[draw, align=center, fill= gray!25] (14) {$H_{1}^-\cap H_{4}^+$ \\ $0.0135$ \\ $(0.0271)$};
  \path (8.5,18) node[draw, align=center] (23) { $H_{2}^+\cap H_{3}^+$ \\ $0.2039$ \\ $(0.2039)$};
  \path (11.5,18) node[draw, align=center, fill= gray!25] (24) { $H_{2}^+\cap H_{4}^+$ \\ $0.0135$ \\ $(0.0271)$};
  \path (14.5,18) node[draw, align=center, fill= gray!25] (34) {  $H_{3}^+\cap H_{4}^+$ \\ $0.0135$ \\ $(0.0271)$};
  \path (2.5,15) node[draw, align=center] (1) {$H_{1}^-$ \\ $0.0403$ \\ $(0.1209)$};
  \path (5.5,15) node[draw, align=center] (2) { $H_{2}^+$ \\ $0.1020$ \\ $(0.2039)$};
  \path (8.5,15) node[draw, align=center] (3) { $H_{3}^+$ \\ $0.4451$ \\ $(0.4451)$};
  \path (11.5,15) node[draw, align=center, fill= gray!25] (4) {  $H_{4}^+$ \\ $0.0068$ \\ $(0.0271)$};
  
  \begin{scope}
  \draw[->]  (1234.south) -- (123.north);
  \draw[->]  (1234.south) -- (124.north);
  \draw[->]  (1234.south) -- (134.north);
  \draw[->]  (1234.south) -- (234.north);
  \draw[->]  (123.south) -- (12.north);
  \draw[->]  (123.south) -- (13.north);
  \draw[->]  (123.south) -- (23.north);
  \draw[->]  (124.south) -- (12.north);
  \draw[->]  (124.south) -- (14.north);
  \draw[->]  (124.south) -- (24.north);
  \draw[->]  (134.south) -- (13.north);
  \draw[->]  (134.south) -- (14.north);
  \draw[->]  (134.south) -- (34.north);
  \draw[->]  (234.south) -- (23.north);
  \draw[->]  (234.south) -- (24.north);
  \draw[->]  (234.south) -- (34.north);
  \draw[->]  (12.south) -- (1.north);
  \draw[->]  (12.south) -- (2.north);
  \draw[->]  (13.south) -- (1.north);
  \draw[->]  (13.south) -- (3.north);
  \draw[->]  (14.south) -- (1.north);
  \draw[->]  (14.south) -- (4.north);
  \draw[->]  (23.south) -- (2.north);
  \draw[->]  (23.south) -- (3.north);
  \draw[->]  (24.south) -- (2.north);
  \draw[->]  (24.south) -- (4.north);
  \draw[->]  (34.south) -- (3.north);
  \draw[->]  (34.south) -- (4.north);
  \end{scope}

\path (7,12.5) node {Procedure 5.1 (Closed testing for qualitative interactions)};
  \path (7,11) node[draw, fill= gray!25, align=center] (qi1234) {  $H_{1}^-\cap H_{2}^+ \cap H_{3}^+ \cap H_{4}^+$ \\ $0.0403$ \\  $(0.0403)$};
  \path (1,8) node[draw, align=center, fill= gray!25] (qi123) {  $H_{1}^-\cap H_{2}^+ \cap H_{3}^+$ \\ $0.0403$ \\ $(0.0403)$};
  \path (5,8) node[draw, align=center, fill= gray!25] (qi124) {   $H_{1}^-\cap H_{2}^+  \cap H_{4}^+$ \\ $0.0403$ \\ $(0.0403)$};
  \path (9,8) node[draw, align=center, fill= gray!25] (qi134) {   $H_{1}^-\cap  H_{3}^+ \cap H_{4}^+$ \\ $0.0403$ \\ $(0.0403)$};
  \path (13,8) node[draw, align=center, fill= gray!25] (qi234) { $H_{2}^+ \cap H_{3}^+ \cap H_{4}^+$ \\ $0.0403$\\ $(0.0403)$};
  \path (-0.5,5) node[draw, align=center, fill= gray!25] (qi12) { $H_{1}^-\cap H_{2}^+$ \\ $0.0403$\\ $(0.0403)$};
  \path (2.5,5) node[draw, align=center, fill= gray!25] (qi13) {  $H_{1}^-\cap H_{3}^+$ \\ $0.0403$ \\ $(0.0403)$};
  \path (5.5,5) node[draw, align=center, fill= gray!25] (qi14) {$H_{1}^-\cap H_{4}^+$ \\ $0.0403$ \\ $(0.0403)$};
  \path (8.5,5) node[draw, align=center] (qi23) { $H_{2}^+\cap H_{3}^+$ \\ $0.2039$ \\ $(0.2039)$};
  \path (11.5,5) node[draw, align=center, fill= gray!25] (qi24) { $H_{2}^+\cap H_{4}^+$ \\ $0.0135$ \\ $(0.0403)$};
  \path (14.5,5) node[draw, align=center, fill= gray!25] (qi34) {  $H_{3}^+\cap H_{4}^+$ \\ $0.0135$ \\ $(0.0403)$};
  \path (2.5,2) node[draw, align=center, fill= gray!25] (qi1) {$H_{1}^-$ \\ $0.0403$ \\ $(0.0403)$};
  \path (5.5,2) node[draw, align=center] (qi2) { $H_{2}^+$ \\ $0.1020$ \\ $(0.2039)$};
  \path (8.5,2) node[draw, align=center] (qi3) { $H_{3}^+$ \\ $0.4451$ \\ $(0.4451)$};
  \path (11.5,2) node[draw, align=center, fill= gray!25] (qi4) {  $H_{4}^+$ \\ $0.0068$ \\ $(0.0403)$};
  
  \begin{scope}
  \draw[->]  (qi1234.south) -- (qi123.north);
  \draw[->]  (qi1234.south) -- (qi124.north);
  \draw[->]  (qi1234.south) -- (qi134.north);
  \draw[->]  (qi1234.south) -- (qi234.north);
  \draw[->]  (qi123.south) -- (qi12.north);
  \draw[->]  (qi123.south) -- (qi13.north);
  \draw[->]  (qi123.south) -- (qi23.north);
  \draw[->]  (qi124.south) -- (qi12.north);
  \draw[->]  (qi124.south) -- (qi14.north);
  \draw[->]  (qi124.south) -- (qi24.north);
  \draw[->]  (qi134.south) -- (qi13.north);
  \draw[->]  (qi134.south) -- (qi14.north);
  \draw[->]  (qi134.south) -- (qi34.north);
  \draw[->]  (qi234.south) -- (qi23.north);
  \draw[->]  (qi234.south) -- (qi24.north);
  \draw[->]  (qi234.south) -- (qi34.north);
  \draw[->]  (qi12.south) -- (qi1.north);
  \draw[->]  (qi12.south) -- (qi2.north);
  \draw[->]  (qi13.south) -- (qi1.north);
  \draw[->]  (qi13.south) -- (qi3.north);
  \draw[->]  (qi14.south) -- (qi1.north);
  \draw[->]  (qi14.south) -- (qi4.north);
  \draw[->]  (qi23.south) -- (qi2.north);
  \draw[->]  (qi23.south) -- (qi3.north);
  \draw[->]  (qi24.south) -- (qi2.north);
  \draw[->]  (qi24.south) -- (qi4.north);
  \draw[->]  (qi34.south) -- (qi3.north);
  \draw[->]  (qi34.south) -- (qi4.north);
  \end{scope}

\end{tikzpicture}
 \caption{The example from Table \ref{tab:Gail}. Top Plot: Procedure \ref{algo-conditional closed testing} (Directional closed testing) with Simes' local tests.  
 Bottom Plot: Procedure \ref{algo-CTQI} (Closed testing for qualitative interactions) with \cite{Zhao2019} test for 
$\{H_1^-\} \cup \{H_2^+ \cap H_3^+ \cap H_4^+\}$.
Each node represents a closed testing intersection hypothesis, and the arrows indicate subset relationships. 
 Below each intersection hypothesis, the corresponding $p$-value and the closed testing adjusted $p$-value (in brackets).
 Gray nodes are hypotheses rejected by the $5\%$ level closed testing procedure.}
 \label{Figure:Gail_CT}
 \end{figure}

The procedure 
 using the same adaptive local tests with partitioning  is   illustrated in  Figure \ref{Figure:Gail_P}.
 We can extract the confidence set for $n^+$ from these non-rejected orthant hypotheses as follows. Since the number of positive parameters in each non-rejected orthant  hypothesis is at least one and at most three, the  bounds are $1 \leq n^+ \leq 3$ with 95\% confidence. This is an  improvement over the bounds for $n^+$ with DCT. Alternatively, we can see that the lower bound is one because $H_1^- \cap H_2^- \cap H_3^- \cap H_4^-$  is rejected  (by a local test of   $H_1^-$, since for the remaining coordinates $K_2, K_3,K_4$ were selected for testing). Moreover, since the fourth individual hypothesis in each non-rejected partition null  hypothesis is non-positive, we further  conclude that $\theta_4 \leq 0$. This is a weaker conclusion compared to $\theta_4 < 0$ obtained by DCT.

 \begin{figure}
 \centering
 \begin{tikzpicture}[scale=1]

  \path (7,12) node[draw, fill= gray!25, align=center] (1234) {  { $K_{1} \cap K_{2} \cap K_{3} \cap K_{4}$ } \\ $0.0203$};
  \path (0,9) node[draw, align=center] (123) { $K_{1} \cap K_{2} \cap K_{3} \cap H_{4}^-$ \\ $0.2039$ };
  \path (4.75,9) node[draw, align=center, fill= gray!25] (124) {   $K_{1} \cap K_{2} \cap H_{3}^- \cap K_{4}$ \\ $0.0135$ };
  \path (9.25,9) node[draw, align=center, fill= gray!25] (134) { $K_{1} \cap H_{2}^- \cap K_{3} \cap K_{4}$ \\ $0.0135$ };
  \path (14,9) node[draw, align=center, fill= gray!25] (234) {$H_{1}^- \cap K_{2} \cap K_{3} \cap K_{4}$  \\ $0.0271$};
  \path (-0.5,6) node[draw, align=center] (12) { \tiny $K_{1} \cap K_{2} \cap H_{3}^- \cap H_{4}^-$ \\ $0.1020$};
  \path (2.5,6) node[draw, align=center] (13) { \tiny $K_{1} \cap H^-_{2} \cap K_{3} \cap H_{4}^-$ \\ $0.4451$ };
  \path (5.5,6) node[draw, align=center, fill= gray!25] (14) { \tiny $K_{1} \cap H^-_{2} \cap H_{3}^- \cap K_{4}$ \\ $0.0068$ };
  \path (8.5,6) node[draw, align=center] (23) { \tiny $H_{1}^-\cap K_{2} \cap K_{3} \cap H_{4}^-$ \\ $0.1209$ };
  \path (11.5,6) node[draw, align=center, fill= gray!25] (24) { \tiny $H_{1}^-\cap K_{2} \cap H_{3}^- \cap K_{4}$ \\ $0.0203$  };
  \path (14.5,6) node[draw, align=center, fill= gray!25] (34) { \tiny $H_{1}^-\cap H_{2}^- \cap K_{3} \cap K_{4}$ \\ $0.0203$ };
  \path (0,3) node[draw, align=center] (1) {$K_{1} \cap H_{2}^- \cap H_{3}^- \cap H_{4}^-$ \\ $1$ };
  \path (4.75,3) node[draw, align=center] (2) { $H_{1}^- \cap K_{2} \cap H_{3}^- \cap H_{4}^-$ \\ $0.0806$ };
  \path (9.25,3) node[draw, align=center] (3) { $H_{1}^- \cap H_{2}^- \cap K_{3} \cap H_{4}^-$ \\ $0.0806$ };
  \path (14,3) node[draw, align=center, fill= gray!25] (4) {  $H_{1}^- \cap H_{2}^- \cap H_{3}^- \cap K_{4}$ \\ $0.0135$ };
  \path (7,0) node[draw, fill= gray!25, align=center] (0) {  $H_{1}^- \cap H_{2}^- \cap H_{3}^- \cap H_{4}^-$ \\ $0.0403$ };
  
\end{tikzpicture}
 \caption{The example from Table \ref{tab:Gail}. Each box represents an orthant hypothesis. 
 Below each orthant hypothesis, the corresponding $p$-value obtained by the Simes local test.
 Gray boxes are orthant hypotheses  rejected by the partitioning procedure at the $5\%$ level. }
 \label{Figure:Gail_P}
 \end{figure}
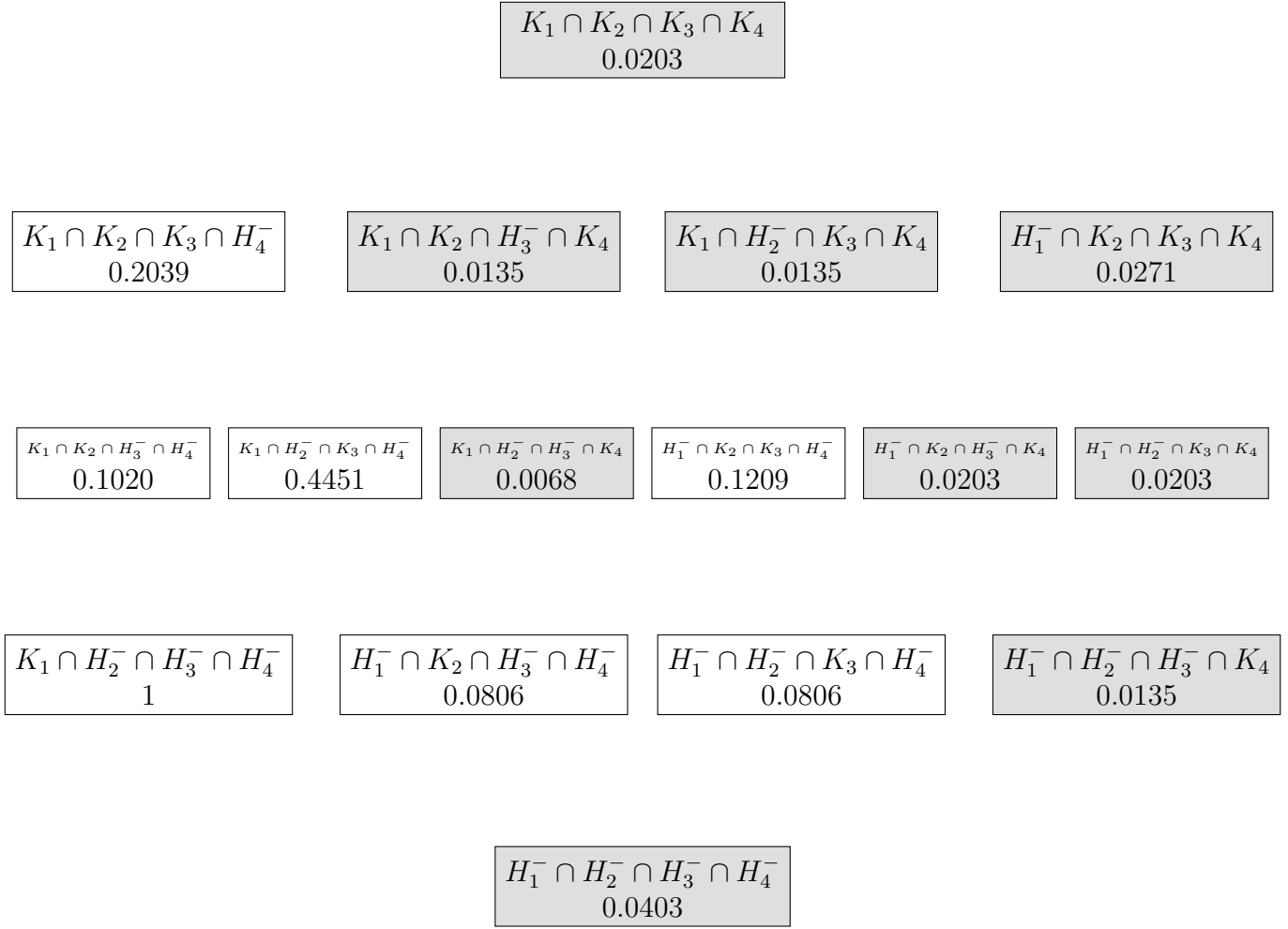

\subsection{Testing for qualitative interactions with follow-up inference}\label{subsec - QI}
If testing for mixed signs of the parameter vector $\theta$ is of primary concern, the approach described in \S~\ref{subsec - meta-analysis} can be adjusted so that DCT is applied only if the hypothesis of no qualitative interaction, $H_0: \{n^+=0 \}\cup \{n^-=0\},$  is rejected.   
 \cite{gail1985testing} analyzed the data of Table \ref{tab:Gail} 
to see whether they could demonstrate a qualitative interaction at the $\alpha=5\%$ level. However, the likelihood ratio test proposed by \cite{gail1985testing} results in a $p$-value of $0.0877$. 

The test of \cite{Zhao2019}, with  selection threshold $\tau = 1/2$, rejects $H_0: \{n^+=0 \}\cup \{n^-=0\}$ if $H_{\mathcal{S}^{-}}$ is rejected with $f(\{2p_i, i\in  \mathcal S^- \})$ and $H_{\mathcal{S}^{+}}$ is rejected with $f(\{2q_i, i\in  \mathcal S^+ \})$. This is a valid test   since $\{n^+=0 \}\subseteq H_{\mathcal{S}^{-}}$ and $\{n^-=0 \}\subseteq H_{\mathcal{S}^{+}}$. 
 The $p$-value is $\max(0.0403, 0.0203)=0.0403$ with  Simes' combination test.

We propose the following procedure tailored for qualitative interactions: start with \cite{Zhao2019} global test for $H_0$ and, if rejected, continue with the DCT procedure. 

\begin{algo}[Closed testing for qualitative interactions]\label{algo-CTQI} $ $

\begin{enumerate}

\item[Step 1] Apply Step 1 of procedure 2.1 

\item[Step 2] Apply a level $\alpha$ closed testing procedure on the family of hypotheses $\{H_i^-: p_i\leq 1/2 \}\cup \{H_i^+: p_i> 1/2 \}$. Testing $H_{\mathcal{I}}$ at level $\alpha$ is done by the local test
\begin{eqnarray*}\label{eq-testQI}
\phi_{\mathcal{I}} = \left\{ \begin{array}{cc} 
\mathds{1}\{ \max\{f(\{2p_i, i\in  \mathcal S^- \}), f(\{2q_i, i\in  \mathcal S^+ \})\} \leq \alpha \} & \mathrm{if\,\,} \mathcal{I}\supseteq \mathcal{S}^{-} \mathrm{\,\,or\,\,} \mathcal{I}\supseteq \mathcal{S}^{+} \\
\mathds{1}\{ f(\{2p_i, i\in \mathcal I \cap \mathcal S^- \}\cup \{2q_j, j\in \mathcal I \cap \mathcal S^+ \}) \leq \alpha \} & \mathrm{otherwise}\\
\end{array}\right. 
\end{eqnarray*}

\end{enumerate}
 \end{algo}

If $H_0$ is rejected, then the lower bounds for $n^+$ and $n^-$ are both informative, i.e. $n^+>0,n^->0$.
From the bottom plot of Figure \ref{Figure:Gail_CT} we see that at level $5\%$ procedure \ref{algo-CTQI} rejects all hypotheses but $H_2^+ \cap H_3^+$, $H_2^+$ and $H_3^+$, giving $1\leq n^+ \leq 3$, $1\leq n^- \leq 3$ and $\theta_1 >0$, $\theta_4 < 0$. This result improves  the one obtained by the DCT procedure \ref{algo-conditional closed testing} ($1\leq n^- \leq 4$ and $\theta_4 < 0$). However, with procedure \ref{algo-CTQI}, rejecting $H_0$ is crucial: if $H_0$ is not rejected, the closed testing procedure \ref{algo-CTQI} does not reject any intersection hypothesis. For instance, at $\alpha=3\%$, procedure \ref{algo-conditional closed testing} gives the same conclusions as at $\alpha=5\%$, while procedure \ref{algo-CTQI} is uninformative since no rejections are made. 

Note that the partitioning procedure with adaptive tests may employ the local tests of procedure \ref{algo-CTQI}, giving
\begin{eqnarray*}\label{eq-partitioningQI}
\psi_{\mathcal{K}} = \left\{ \begin{array}{cc} 
\mathds{1}\{ \max\{f(\{2p_i, i\in  \mathcal S^- \}), f(\{2q_i, i\in  \mathcal S^+ \})\} \leq \alpha \} & \mathrm{if\,\,} \mathcal{K}\subseteq \mathcal{S}^{-} \mathrm{\,\,or\,\,} \mathcal{K}\subseteq \mathcal{S}^{+} \\
\mathds{1}\{ f(\{2p_i, i\in \mathcal{K}^c \cap \mathcal S^- \}\cup \{2q_j, j\in \mathcal{K} \cap \mathcal S^+ \}) \leq \alpha \} & \mathrm{otherwise}\\
\end{array}\right. 
\end{eqnarray*}
Applied to the data in Table \ref{tab:Gail}, this procedure at level $\alpha=5\%$ gives the same conclusions as  Procedure \ref{algo-CTQI} at the same level.

\subsection{Enhancing meta-analysis}\label{subsec - meta-analysis}

The first step in a meta-analysis is often the computation of a global null $p$-value, which communicates the strength of the evidence against the null hypothesis of no association in any of the studies. Rejecting the global null does not rule out the possibility that this is due to a non-zero association only in a single study (e.g., due to the particulars of the cohort in that study, or due to bias). Therefore, 
it is of interest provide tight lower and upper bounds on the number of studies with, e.g.,  a positive effect. 
Moreover, for non-trivial bounds it is natural also to try to identify the studies where there is an association. DCT should be used for this purpose if inference on positive and negative associations  are desired.  Adaptive local tests with partitioning should be used if it is enough to infer the lower and upper bounds on $n^+$, and to identify positive and non-positive associations. Since in a typical meta-analysis  the studies are independent and the  test statistic in each study is approximately normal,  the 
required conditions $(A1)-(A2)$ are satisfied.  

Fore a concrete example, we consider the  data of \cite{cooper2003effects}, made public by \cite{konstantopoulos2011fixed}, where the  effect of a modified school calendar (with more frequent but shorter breaks) on student achievement is examined. The meta-analysis used studies of $n=56$ schools in 11 districts.  The experimental design is a two-group comparison (modified calendar vs traditional
calendar) which involves computing the standardized mean difference
$y_i \sim N(\theta_i,1)$, where $\theta_i>0$ indicates a positive effects, i.e. a higher level of achievement in the
group following the modified school calendar.
\cite{Zhao2019} showed that there is evidence of  a qualitative interaction at the $5\%$ level, i.e., that the effect of the modified school calendar  on student achievement is positive in at least one study,  and negative in at least one study.

Table \ref{table:konstantopoulos} shows the results of our analysis. 
With DCT, the upper bound on $n^+$ is the trivial bound of $n=56$. Adaptive local tests with partitioning gives much more informative bounds: using Fisher's combining method, the 95\% confidence bound for $n^+$ is 10 to 53, and we have four individual discoveries; using  mSimes, the  95\% confidence bound for $n^+$ is 8 to 54, and we have eight individual discoveries. The bounds are tighter than the bounds obtained by counting discoveries in each direction, since  local tests used for the bounds can reject the intersection hypotheses even when base hypotheses are not rejected.
Using the procedure of \cite{guo2015stepwise} which is  targeted for FWER control,   the 95\% confidence bound for $n^+$ is 8 to 55, with nine individual discoveries. Using the procedure of \cite{guo2015stepwise} which is  targeted for FDR control we get the narrowest bounds for $n^+$, 12 to 52. However, these are not 95\% confidence bounds since this procedure does not provide a simultaneous confidence guarantee. 

Post-hoc bounds are obtained for $n^+(\mathcal{I}_k)$ for $k=1,\ldots,n$ using Fisher's combining function, where $\mathcal{I}_k$ represents the indices of the top $k=|\mathcal{I}_k|$ schools with the largest (in absolute value) estimated effects $|y_i|$.

\begin{table}
\caption{\label{table:konstantopoulos} The effect of modified school calendars on student achievement:  intervals for $n^+$ and $n^-$ and corresponding number of discoveries $|\mathcal{D}^{+}|$ and $|\mathcal{D}^{-}|$, for different combining methods and procedures, with $95\%$ confidence. Results for Procedure \ref{algo-conditional closed testing} (DCT), the adaptive local tests with partitioning (AP), and the procedures in \cite{guo2015stepwise} for FWER and FDR control (last two rows).    $|\mathcal{D}^{-}|$ is the number of discoveries of negative parameters for DCT  and non-positive parameters otherwise.
}
\begin{tabular}{cc|ccc|ccc|}
 \multicolumn{2}{c}{} & \multicolumn{3}{c}{$n^+$} & \multicolumn{3}{c}{$n^-\ \backslash \ n^-+n^0$} \\
Local test & Procedure  &  lower bound & upper bound & $|\mathcal{D}^{+}|$  &  lower bound & upper bound & $|\mathcal{D}^{-}|$ \\
 \hline
\multirow{2}{*}{Fisher}& DCT   &  9 & 56 & \multirow{2}{*}{4} & 0 & 47 & \multirow{2}{*}{0}  \\
& AP  &    10 & 53 &   & 3 & 46 & \\
 \hline
\multirow{2}{*}{Simes}& DCT   &  8 & 56 & \multirow{2}{*}{8}  & 0 & 48 & \multirow{2}{*}{0} \\
& AP   &    8 & 55 &   & 1 & 48 & \\
 \hline
modified  & DCT  & 8 & 56 & \multirow{2}{*}{8}  & 0 & 48 & \multirow{2}{*}{0} \\
Simes       & AP  & 8 & 54 &  & 2 & 48 &  \\
 \hline
adaptive & DCT  & 9 & 56 & \multirow{2}{*}{5}   & 0 & 47 & \multirow{2}{*}{0}  \\
LRT &  AP  & 11 & 53 &    & 3 & 45 & \\
 \hline
\multirow{2}{*}{Guo-Romano} & FWER & 8 & 55 & 8 &   1 & 48 & 1  \\
& FDR & 12 & 52 & 12 &   4 & 44 & 4  \\
\hline
\end{tabular}

\end{table}

\section{Simulations}\label{sec - sim}

We carry out simulation studies in order to evaluate the performance of our suggested methods. Our main focus is on the evaluation of the DCT procedure  \ref{algo-conditional closed testing} and the  adaptive local tests with partitioning (AP) procedure, with different combining methods of potential interest. Proposition \ref{prop-AP-vs-DCT} states that the bounds of AP are at least as tight as with DCT.  We aim  to quantify the potential gain in terms of tighter bounds. 
We  use two representative combining methods for this purpose: Fisher, which is a sum-based method, and modified Simes \eqref{eq-modified-adaptive-Simes}, which is quantile based. Sum based tests, such as Fisher, are better for the bounds but worse for discoveries.  
In the SM, we further consider 
the adaptive likelihood ratio test  \citep{Mohamad2020}, which is a sum based test for normal test statistics,  and Simes -- it is inferior to mSimes for the bounds as well as for (directional) decisions in our simulations. 

We compare our procedures to the FWER and FDR controlling procedures in \cite{guo2015stepwise}, which were introduced in \S~\ref{sec-background-gr}. These procedures aim to identify the positive and non-positive parameters, so they control the directional error rates. We extract the  lower  bounds for $n^+(\mathcal I)$ and $n^-(\mathcal I) + n^0(\mathcal I)$ from these procedures by counting the number of discoveries that are claimed to positive, and non-positive, respectively. The FWER controlling procedure is a good competitor for discoveries, but does poorly for the bounds. The FDR controlling procedure is less relevant, since 
unlike all the other procedures considered, the probability that the number of positive or non-positive discoveries are greater than the true number of positive and non-positive parameters in the set of discoveries can be much larger than $\alpha$, see Figure \ref{fig-SM-GRFDR-coverage} in the SM.   
However, we included it since it is interesting that for the bounds it is not  necessarily better, even though there is no confidence guarantee.

Our data generating mechanism is as follows.  For a total of $n=50$ parameters, there are $n^+ \in \{ 0, \ldots, n\}$ positive parameters and $n^-\in \{ 0, \ldots, n-n^+\}$ negative parameters.  The remaining parameters are zero.  
The test statistics are generated independently from a Gaussian distribution with standard deviation one and mean centered at the parameter value, $\hat \theta_i \sim N(\theta_i,1)$. The one sided $p$-value for $H_i^-$ is $p_i = 1-\Phi(\hat \theta_i), i=1, \ldots, n$.  Our analysis is based on $B=2000$ data generating mechanisms, and $\alpha = 0.05$. 

We display the average  sum of the lower bounds on the number of positive and non-positive parameters (the average length of the bounds is $n$ minus this average) and the number of parameters discovered for each procedure, in Figure \ref{fig-tikz-bounds-main} and supplementary Figure \ref{fig-SM-sim}.   
Our key findings are the following. First,  AP provides much tighter bounds than DCT.
Second, the sum-based tests are better for the bounds than the quantile based tests and than the  FWER controlling procedure in  \cite{guo2015stepwise}. They can also be much better for the bounds than the FDR controlling procedure in  \cite{guo2015stepwise} in  data generating mechanisms with fairly weak signal. 
Third, for making discoveries, mSimes  is competitive with the FWER controlling procedure in  \cite{guo2015stepwise}, and they are both  much better than the sum-based tests. 
From the simulations, we can support the following general guidance:  mSimes is a good choice if making discoveries is important in addition to the bounds, or if it is expected  that the fraction of  parameters far enough from zero  is small. Otherwise it is best to use a sum based combining method, such as Fisher  or the ALRT.

\begin{figure}
\centering
\begin{tikzpicture}[scale=0.7]
	\begin{axis}[
	xmin = 0,
	xmax = 3,
	xtick ={0,...,3},
xticklabels ={0,1,2,3},
	ymin = -1,
	ymax = 27,
	xlabel=Signal-to-noise ratio,
		ylabel= $\ell^+ + \ell^-$,
		height=10cm,
		width=10cm,
		legend style={at={(0.5,-0.15)},
		anchor=north,legend columns=3}
	]
	\addplot[only marks, mark size=1.5pt, color=black,  mark=o] coordinates {
( 0,0.1105 ) ( 0.1,0.136 ) ( 0.2,0.175 ) ( 0.3,0.1745 ) ( 0.4,0.2715 ) ( 0.5,0.402 ) ( 0.6,0.5955 ) ( 0.7,0.884 ) ( 0.8,1.2345 ) ( 0.9,1.727 ) ( 1,2.495 ) ( 1.1,3.27 ) ( 1.2,4.0865 ) ( 1.3,5.2445 ) ( 1.4,6.281 ) ( 1.5,7.587 ) ( 1.6,8.7985 ) ( 1.7,10.135 ) ( 1.8,11.4605 ) ( 1.9,12.61 ) ( 2,13.8555 ) ( 2.1,15.2775 ) ( 2.2,16.3795 ) ( 2.3,17.4775 ) ( 2.4,18.5945 ) ( 2.5,19.6695 ) ( 2.6,20.4835 ) ( 2.7,21.2385 ) ( 2.8,22.063 ) ( 2.9,22.7415 ) ( 3,23.3695 )
};

	\addplot[only marks, mark size=2pt, color=black,  mark=diamond] coordinates {
( 0,0.1045 ) ( 0.1,0.11 ) ( 0.2,0.1255 ) ( 0.3,0.132 ) ( 0.4,0.174 ) ( 0.5,0.2145 ) ( 0.6,0.2765 ) ( 0.7,0.358 ) ( 0.8,0.464 ) ( 0.9,0.5875 ) ( 1,0.7955 ) ( 1.1,1.014 ) ( 1.2,1.262 ) ( 1.3,1.683 ) ( 1.4,2.13 ) ( 1.5,2.668 ) ( 1.6,3.2625 ) ( 1.7,4.059 ) ( 1.8,4.901 ) ( 1.9,5.795 ) ( 2,6.759 ) ( 2.1,8.018 ) ( 2.2,9.132 ) ( 2.3,10.298 ) ( 2.4,11.5825 ) ( 2.5,12.961 ) ( 2.6,14.1965 ) ( 2.7,15.3865 ) ( 2.8,16.4925 ) ( 2.9,17.6745 ) ( 3,18.6815 )
};

	\addplot[only marks, mark size=1.5pt, color=blue,  mark=o] coordinates {
( 0,0.001 ) ( 0.1,0.001 ) ( 0.2,0.002 ) ( 0.3,0.0065 ) ( 0.4,0.0085 ) ( 0.5,0.024 ) ( 0.6,0.0415 ) ( 0.7,0.0845 ) ( 0.8,0.1815 ) ( 0.9,0.322 ) ( 1,0.6005 ) ( 1.1,0.9715 ) ( 1.2,1.407 ) ( 1.3,2.1895 ) ( 1.4,2.929 ) ( 1.5,3.9535 ) ( 1.6,5.0015 ) ( 1.7,6.251 ) ( 1.8,7.502 ) ( 1.9,8.6195 ) ( 2,9.9 ) ( 2.1,11.3555 ) ( 2.2,12.537 ) ( 2.3,13.6935 ) ( 2.4,14.888 ) ( 2.5,16.0815 ) ( 2.6,16.9845 ) ( 2.7,17.8545 ) ( 2.8,18.762 ) ( 2.9,19.581 ) ( 3,20.259 )
};

	\addplot[only marks, mark size=2pt, color=blue,  mark=diamond] coordinates {
( 0,0.058 ) ( 0.1,0.057 ) ( 0.2,0.0635 ) ( 0.3,0.068 ) ( 0.4,0.085 ) ( 0.5,0.1125 ) ( 0.6,0.15 ) ( 0.7,0.2035 ) ( 0.8,0.2575 ) ( 0.9,0.3335 ) ( 1,0.4635 ) ( 1.1,0.5905 ) ( 1.2,0.7645 ) ( 1.3,1.0175 ) ( 1.4,1.312 ) ( 1.5,1.683 ) ( 1.6,2.117 ) ( 1.7,2.6585 ) ( 1.8,3.3 ) ( 1.9,3.9935 ) ( 2,4.743 ) ( 2.1,5.851 ) ( 2.2,6.7975 ) ( 2.3,7.811 ) ( 2.4,9.036 ) ( 2.5,10.2995 ) ( 2.6,11.5595 ) ( 2.7,12.7155 ) ( 2.8,13.817 ) ( 2.9,15.1165 ) ( 3,16.1595 )
};

	\addplot[ mark size=2pt, color=red,  mark=x] coordinates {
( 0,0.108 ) ( 0.1,0.1065 ) ( 0.2,0.12 ) ( 0.3,0.134 ) ( 0.4,0.1575 ) ( 0.5,0.1995 ) ( 0.6,0.2495 ) ( 0.7,0.304 ) ( 0.8,0.3925 ) ( 0.9,0.476 ) ( 1,0.605 ) ( 1.1,0.7645 ) ( 1.2,0.939 ) ( 1.3,1.1545 ) ( 1.4,1.455 ) ( 1.5,1.7555 ) ( 1.6,2.1485 ) ( 1.7,2.59 ) ( 1.8,3.0725 ) ( 1.9,3.6555 ) ( 2,4.255 ) ( 2.1,5.148 ) ( 2.2,5.913 ) ( 2.3,6.809 ) ( 2.4,7.9315 ) ( 2.5,8.997 ) ( 2.6,10.286 ) ( 2.7,11.451 ) ( 2.8,12.587 ) ( 2.9,13.993 ) ( 3,15.1895 )
};

	\addplot[ mark size=2pt, color=red,  mark=+] coordinates {
( 0,0.1245 ) ( 0.1,0.1395 ) ( 0.2,0.149 ) ( 0.3,0.1615 ) ( 0.4,0.21 ) ( 0.5,0.2595 ) ( 0.6,0.3715 ) ( 0.7,0.458 ) ( 0.8,0.6415 ) ( 0.9,0.84 ) ( 1,1.218 ) ( 1.1,1.595 ) ( 1.2,2.012 ) ( 1.3,2.8545 ) ( 1.4,3.767 ) ( 1.5,4.8835 ) ( 1.6,6.1525 ) ( 1.7,7.7925 ) ( 1.8,9.418 ) ( 1.9,11.1655 ) ( 2,12.87 ) ( 2.1,14.93 ) ( 2.2,16.599 ) ( 2.3,18.2775 ) ( 2.4,20.085 ) ( 2.5,21.491 ) ( 2.6,22.745 ) ( 2.7,23.9975 ) ( 2.8,25.154 ) ( 2.9,26.136 ) ( 3,26.8915 )
};
	\end{axis}
\end{tikzpicture} ~ \begin{tikzpicture}[scale=0.7]
	\begin{axis}[
	xmin = 0,
	xmax = 3,
	xtick ={0,...,3},
xticklabels ={0,1,2,3},
	ymin = -1,
	ymax = 27,
	xlabel=Signal-to-noise ratio,
		ylabel= $|\mathcal{D}^+| + |\mathcal{D}^-|$,
		height=10cm,
		width=10cm,
		legend style={at={(0.5,-0.15)},
		anchor=north,legend columns=2}
	]
	\addplot[only marks, mark size=1.5pt, color=black,  mark=o] coordinates {
( 0,0 ) ( 0.1,0 ) ( 0.2,0 ) ( 0.3,0 ) ( 0.4,0 ) ( 0.5,0 ) ( 0.6,0 ) ( 0.7,0 ) ( 0.8,0 ) ( 0.9,0 ) ( 1,0 ) ( 1.1,0 ) ( 1.2,0 ) ( 1.3,0 ) ( 1.4,0.001 ) ( 1.5,5e-04 ) ( 1.6,0 ) ( 1.7,0.002 ) ( 1.8,0.003 ) ( 1.9,0.004 ) ( 2,0.017 ) ( 2.1,0.022 ) ( 2.2,0.035 ) ( 2.3,0.0585 ) ( 2.4,0.096 ) ( 2.5,0.116 ) ( 2.6,0.194 ) ( 2.7,0.2665 ) ( 2.8,0.3685 ) ( 2.9,0.4695 ) ( 3,0.5995 )
};

	\addplot[only marks, mark size=2pt, color=black,  mark=diamond] coordinates {
( 0,0.0575 ) ( 0.1,0.0555 ) ( 0.2,0.0615 ) ( 0.3,0.0675 ) ( 0.4,0.0825 ) ( 0.5,0.1085 ) ( 0.6,0.146 ) ( 0.7,0.1945 ) ( 0.8,0.245 ) ( 0.9,0.3115 ) ( 1,0.4405 ) ( 1.1,0.5465 ) ( 1.2,0.702 ) ( 1.3,0.925 ) ( 1.4,1.1675 ) ( 1.5,1.4845 ) ( 1.6,1.8645 ) ( 1.7,2.2855 ) ( 1.8,2.7845 ) ( 1.9,3.355 ) ( 2,3.9605 ) ( 2.1,4.8655 ) ( 2.2,5.647 ) ( 2.3,6.528 ) ( 2.4,7.5635 ) ( 2.5,8.66 ) ( 2.6,9.8175 ) ( 2.7,10.8925 ) ( 2.8,11.9615 ) ( 2.9,13.2755 ) ( 3,14.387 )
};

	\addplot[ mark size=2pt, color=red,  mark=x] coordinates {
( 0,0.108 ) ( 0.1,0.1065 ) ( 0.2,0.12 ) ( 0.3,0.134 ) ( 0.4,0.1575 ) ( 0.5,0.1995 ) ( 0.6,0.2495 ) ( 0.7,0.304 ) ( 0.8,0.3925 ) ( 0.9,0.476 ) ( 1,0.605 ) ( 1.1,0.7645 ) ( 1.2,0.939 ) ( 1.3,1.1545 ) ( 1.4,1.455 ) ( 1.5,1.7555 ) ( 1.6,2.1485 ) ( 1.7,2.59 ) ( 1.8,3.0725 ) ( 1.9,3.6555 ) ( 2,4.255 ) ( 2.1,5.148 ) ( 2.2,5.913 ) ( 2.3,6.809 ) ( 2.4,7.9315 ) ( 2.5,8.997 ) ( 2.6,10.286 ) ( 2.7,11.451 ) ( 2.8,12.587 ) ( 2.9,13.993 ) ( 3,15.1895 )
};

	\addplot[ mark size=2pt, color=red,  mark=+] coordinates {
( 0,0.1245 ) ( 0.1,0.1395 ) ( 0.2,0.149 ) ( 0.3,0.1615 ) ( 0.4,0.21 ) ( 0.5,0.2595 ) ( 0.6,0.3715 ) ( 0.7,0.458 ) ( 0.8,0.6415 ) ( 0.9,0.84 ) ( 1,1.218 ) ( 1.1,1.595 ) ( 1.2,2.012 ) ( 1.3,2.8545 ) ( 1.4,3.767 ) ( 1.5,4.8835 ) ( 1.6,6.1525 ) ( 1.7,7.7925 ) ( 1.8,9.418 ) ( 1.9,11.1655 ) ( 2,12.87 ) ( 2.1,14.93 ) ( 2.2,16.599 ) ( 2.3,18.2775 ) ( 2.4,20.085 ) ( 2.5,21.491 ) ( 2.6,22.745 ) ( 2.7,23.9975 ) ( 2.8,25.154 ) ( 2.9,26.136 ) ( 3,26.8915 )
};

	\end{axis}

\end{tikzpicture} 
\begin{tikzpicture}[scale=0.7]
	\begin{axis}[
	xmin = 0,
	xmax = 3,
	xtick ={0,...,3},
xticklabels ={0,1,2,3},
	ymin = -1,
	ymax = 27,
	xlabel=Signal-to-noise ratio,
		ylabel= $\ell^+ + \ell^-$,
		height=10cm,
		width=10cm,
		legend style={at={(0.5,-0.15)},
		anchor=north,legend columns=3}
	]
	\addplot[only marks, mark size=1.5pt, color=black,  mark=o] coordinates {
( 0,0.1055 ) ( 0.1,0.16 ) ( 0.2,0.1845 ) ( 0.3,0.2655 ) ( 0.4,0.425 ) ( 0.5,0.665 ) ( 0.6,0.9725 ) ( 0.7,1.4795 ) ( 0.8,2.1325 ) ( 0.9,2.7925 ) ( 1,3.7435 ) ( 1.1,4.8585 ) ( 1.2,6.0135 ) ( 1.3,7.0975 ) ( 1.4,8.4445 ) ( 1.5,9.765 ) ( 1.6,11.0345 ) ( 1.7,12.5775 ) ( 1.8,13.846 ) ( 1.9,15.0035 ) ( 2,16.288 ) ( 2.1,17.53 ) ( 2.2,18.5255 ) ( 2.3,19.7095 ) ( 2.4,20.6365 ) ( 2.5,21.542 ) ( 2.6,22.3205 ) ( 2.7,23.1185 ) ( 2.8,23.816 ) ( 2.9,24.372 ) ( 3,24.919 )
};
 \addlegendentry{Fisher AP}

	\addplot[only marks, mark size=2pt, color=black,  mark=diamond] coordinates {
( 0,0.102 ) ( 0.1,0.11 ) ( 0.2,0.1155 ) ( 0.3,0.132 ) ( 0.4,0.171 ) ( 0.5,0.1935 ) ( 0.6,0.2595 ) ( 0.7,0.3215 ) ( 0.8,0.444 ) ( 0.9,0.532 ) ( 1,0.7215 ) ( 1.1,1.0045 ) ( 1.2,1.3075 ) ( 1.3,1.677 ) ( 1.4,2.1995 ) ( 1.5,2.8465 ) ( 1.6,3.5245 ) ( 1.7,4.4615 ) ( 1.8,5.4425 ) ( 1.9,6.493 ) ( 2,7.6075 ) ( 2.1,8.941 ) ( 2.2,10.1065 ) ( 2.3,11.507 ) ( 2.4,12.768 ) ( 2.5,14.1575 ) ( 2.6,15.1695 ) ( 2.7,16.522 ) ( 2.8,17.6385 ) ( 2.9,18.817 ) ( 3,19.823 )
};
 \addlegendentry{mSimes AP}

	\addplot[only marks, mark size=1.5pt, color=blue,  mark=o] coordinates {
( 0,0.002 ) ( 0.1,0.0035 ) ( 0.2,0.0075 ) ( 0.3,0.0175 ) ( 0.4,0.0355 ) ( 0.5,0.0805 ) ( 0.6,0.163 ) ( 0.7,0.352 ) ( 0.8,0.6515 ) ( 0.9,1.067 ) ( 1,1.717 ) ( 1.1,2.6185 ) ( 1.2,3.5785 ) ( 1.3,4.6155 ) ( 1.4,5.964 ) ( 1.5,7.1955 ) ( 1.6,8.49 ) ( 1.7,9.9855 ) ( 1.8,11.3085 ) ( 1.9,12.46 ) ( 2,13.804 ) ( 2.1,15.0715 ) ( 2.2,16.0915 ) ( 2.3,17.3025 ) ( 2.4,18.2685 ) ( 2.5,19.2445 ) ( 2.6,20.072 ) ( 2.7,20.9145 ) ( 2.8,21.6885 ) ( 2.9,22.342 ) ( 3,22.929 )
};
 \addlegendentry{Fisher DCT}

	\addplot[only marks, mark size=2pt, color=blue,  mark=diamond] coordinates {
( 0,0.0535 ) ( 0.1,0.05 ) ( 0.2,0.061 ) ( 0.3,0.0695 ) ( 0.4,0.09 ) ( 0.5,0.107 ) ( 0.6,0.146 ) ( 0.7,0.199 ) ( 0.8,0.258 ) ( 0.9,0.322 ) ( 1,0.4425 ) ( 1.1,0.6455 ) ( 1.2,0.842 ) ( 1.3,1.064 ) ( 1.4,1.4485 ) ( 1.5,1.889 ) ( 1.6,2.4055 ) ( 1.7,3.0605 ) ( 1.8,3.8415 ) ( 1.9,4.678 ) ( 2,5.6625 ) ( 2.1,6.7735 ) ( 2.2,7.8415 ) ( 2.3,9.084 ) ( 2.4,10.273 ) ( 2.5,11.5765 ) ( 2.6,12.644 ) ( 2.7,13.984 ) ( 2.8,15.186 ) ( 2.9,16.4555 ) ( 3,17.512 )
};
 \addlegendentry{mSimes DCT}

	\addplot[ mark size=2pt, color=red,  mark=x] coordinates {
( 0,0.097 ) ( 0.1,0.1 ) ( 0.2,0.114 ) ( 0.3,0.126 ) ( 0.4,0.17 ) ( 0.5,0.1925 ) ( 0.6,0.231 ) ( 0.7,0.2935 ) ( 0.8,0.3895 ) ( 0.9,0.4355 ) ( 1,0.5645 ) ( 1.1,0.7715 ) ( 1.2,0.943 ) ( 1.3,1.141 ) ( 1.4,1.4395 ) ( 1.5,1.779 ) ( 1.6,2.1835 ) ( 1.7,2.559 ) ( 1.8,3.0885 ) ( 1.9,3.712 ) ( 2,4.364 ) ( 2.1,5.0975 ) ( 2.2,5.9655 ) ( 2.3,6.8815 ) ( 2.4,7.879 ) ( 2.5,8.9795 ) ( 2.6,10.075 ) ( 2.7,11.2745 ) ( 2.8,12.588 ) ( 2.9,13.8965 ) ( 3,15.249 )
};
 \addlegendentry{GR FWER}
 
	\addplot[ mark size=2pt, color=red,  mark=+] coordinates {
( 0,0.1255 ) ( 0.1,0.1225 ) ( 0.2,0.1395 ) ( 0.3,0.167 ) ( 0.4,0.233 ) ( 0.5,0.258 ) ( 0.6,0.356 ) ( 0.7,0.4445 ) ( 0.8,0.629 ) ( 0.9,0.812 ) ( 1,1.0715 ) ( 1.1,1.571 ) ( 1.2,2.146 ) ( 1.3,2.7945 ) ( 1.4,3.7395 ) ( 1.5,4.9325 ) ( 1.6,6.2355 ) ( 1.7,7.7635 ) ( 1.8,9.5445 ) ( 1.9,11.209 ) ( 2,13.0515 ) ( 2.1,14.9335 ) ( 2.2,16.431 ) ( 2.3,18.36 ) ( 2.4,19.8515 ) ( 2.5,21.5025 ) ( 2.6,22.6755 ) ( 2.7,23.9535 ) ( 2.8,25.0505 ) ( 2.9,26.0315 ) ( 3,26.807 )
};
 \addlegendentry{GR FDR}
	\end{axis}
\end{tikzpicture} ~ \begin{tikzpicture}[scale=0.7]
	\begin{axis}[
	xmin = 0,
	xmax = 3,
	xtick ={0,...,3},
xticklabels ={0,1,2,3},
	ymin = -1,
	ymax = 27,
	xlabel=Signal-to-noise ratio,
		ylabel= $|\mathcal{D}^+| + |\mathcal{D}^-|$,
		height=10cm,
		width=10cm,
		legend style={at={(0.5,-0.15)},
		anchor=north,legend columns=2}
	]
	\addplot[only marks, mark size=1.5pt, color=black,  mark=o] coordinates {
( 0,0 ) ( 0.1,0 ) ( 0.2,0 ) ( 0.3,0 ) ( 0.4,0 ) ( 0.5,0 ) ( 0.6,0 ) ( 0.7,0 ) ( 0.8,0 ) ( 0.9,0 ) ( 1,0 ) ( 1.1,0 ) ( 1.2,0 ) ( 1.3,0 ) ( 1.4,5e-04 ) ( 1.5,0 ) ( 1.6,0.001 ) ( 1.7,0.003 ) ( 1.8,0.0045 ) ( 1.9,0.007 ) ( 2,0.0145 ) ( 2.1,0.019 ) ( 2.2,0.037 ) ( 2.3,0.0615 ) ( 2.4,0.0815 ) ( 2.5,0.1265 ) ( 2.6,0.1775 ) ( 2.7,0.2685 ) ( 2.8,0.356 ) ( 2.9,0.46 ) ( 3,0.6275 )
};
 \addlegendentry{Fisher}

	\addplot[only marks, mark size=2pt, color=black,  mark=diamond] coordinates {
( 0,0.052 ) ( 0.1,0.05 ) ( 0.2,0.0595 ) ( 0.3,0.067 ) ( 0.4,0.0885 ) ( 0.5,0.103 ) ( 0.6,0.14 ) ( 0.7,0.19 ) ( 0.8,0.2395 ) ( 0.9,0.2975 ) ( 1,0.394 ) ( 1.1,0.575 ) ( 1.2,0.7285 ) ( 1.3,0.8915 ) ( 1.4,1.1915 ) ( 1.5,1.4905 ) ( 1.6,1.895 ) ( 1.7,2.2835 ) ( 1.8,2.83 ) ( 1.9,3.4105 ) ( 2,4.095 ) ( 2.1,4.8275 ) ( 2.2,5.6625 ) ( 2.3,6.5615 ) ( 2.4,7.553 ) ( 2.5,8.6195 ) ( 2.6,9.625 ) ( 2.7,10.7655 ) ( 2.8,11.993 ) ( 2.9,13.257 ) ( 3,14.431 )
};
 \addlegendentry{mSimes}

	\addplot[ mark size=2pt, color=red,  mark=x] coordinates {
( 0,0.097 ) ( 0.1,0.1 ) ( 0.2,0.114 ) ( 0.3,0.126 ) ( 0.4,0.17 ) ( 0.5,0.1925 ) ( 0.6,0.231 ) ( 0.7,0.2935 ) ( 0.8,0.3895 ) ( 0.9,0.4355 ) ( 1,0.5645 ) ( 1.1,0.7715 ) ( 1.2,0.943 ) ( 1.3,1.141 ) ( 1.4,1.4395 ) ( 1.5,1.779 ) ( 1.6,2.1835 ) ( 1.7,2.559 ) ( 1.8,3.0885 ) ( 1.9,3.712 ) ( 2,4.364 ) ( 2.1,5.0975 ) ( 2.2,5.9655 ) ( 2.3,6.8815 ) ( 2.4,7.879 ) ( 2.5,8.9795 ) ( 2.6,10.075 ) ( 2.7,11.2745 ) ( 2.8,12.588 ) ( 2.9,13.8965 ) ( 3,15.249 )
};
 \addlegendentry{GR FWER}
 
	\addplot[ mark size=2pt, color=red,  mark=+] coordinates {
( 0,0.1255 ) ( 0.1,0.1225 ) ( 0.2,0.1395 ) ( 0.3,0.167 ) ( 0.4,0.233 ) ( 0.5,0.258 ) ( 0.6,0.356 ) ( 0.7,0.4445 ) ( 0.8,0.629 ) ( 0.9,0.812 ) ( 1,1.0715 ) ( 1.1,1.571 ) ( 1.2,2.146 ) ( 1.3,2.7945 ) ( 1.4,3.7395 ) ( 1.5,4.9325 ) ( 1.6,6.2355 ) ( 1.7,7.7635 ) ( 1.8,9.5445 ) ( 1.9,11.209 ) ( 2,13.0515 ) ( 2.1,14.9335 ) ( 2.2,16.431 ) ( 2.3,18.36 ) ( 2.4,19.8515 ) ( 2.5,21.5025 ) ( 2.6,22.6755 ) ( 2.7,23.9535 ) ( 2.8,25.0505 ) ( 2.9,26.0315 ) ( 3,26.807 )
};
 \addlegendentry{GR FDR}

	\end{axis}

\end{tikzpicture} 
\caption{The average sum of the lower bounds on positive and non-positive parameters (left column, higher is better) and  number of base hypotheses rejected (right column, higher is better), versus the signal to noise ratio (i.e., the absolute value of the non-null $\theta_i$ value),  
for the following methods: AP and DCT with combining functions Fisher (AP FISHER and DCT Fisher), and  mSimes (AP mSIMES and DCT mSimes); 
the Guo-Romano procedure for FWER control (GR FWER) and for FDR control (GR FDR). $n^+=n^-=15$ in the first row; $n^+=30$ and $n^-=0$ in the second row.  In the right column, the absence of DCT procedures  is due to the fact that they coincide with the AP procedures (Proposition \ref{prop-AP-vs-DCT}). 
Note that for GR-FWER and GR-FDR, the sum of lower bounds on the number of positive and non-positive parameters is their number of discoveries, so their curves are the same across panels in each row.
}
\label{fig-tikz-bounds-main}
\end{figure}
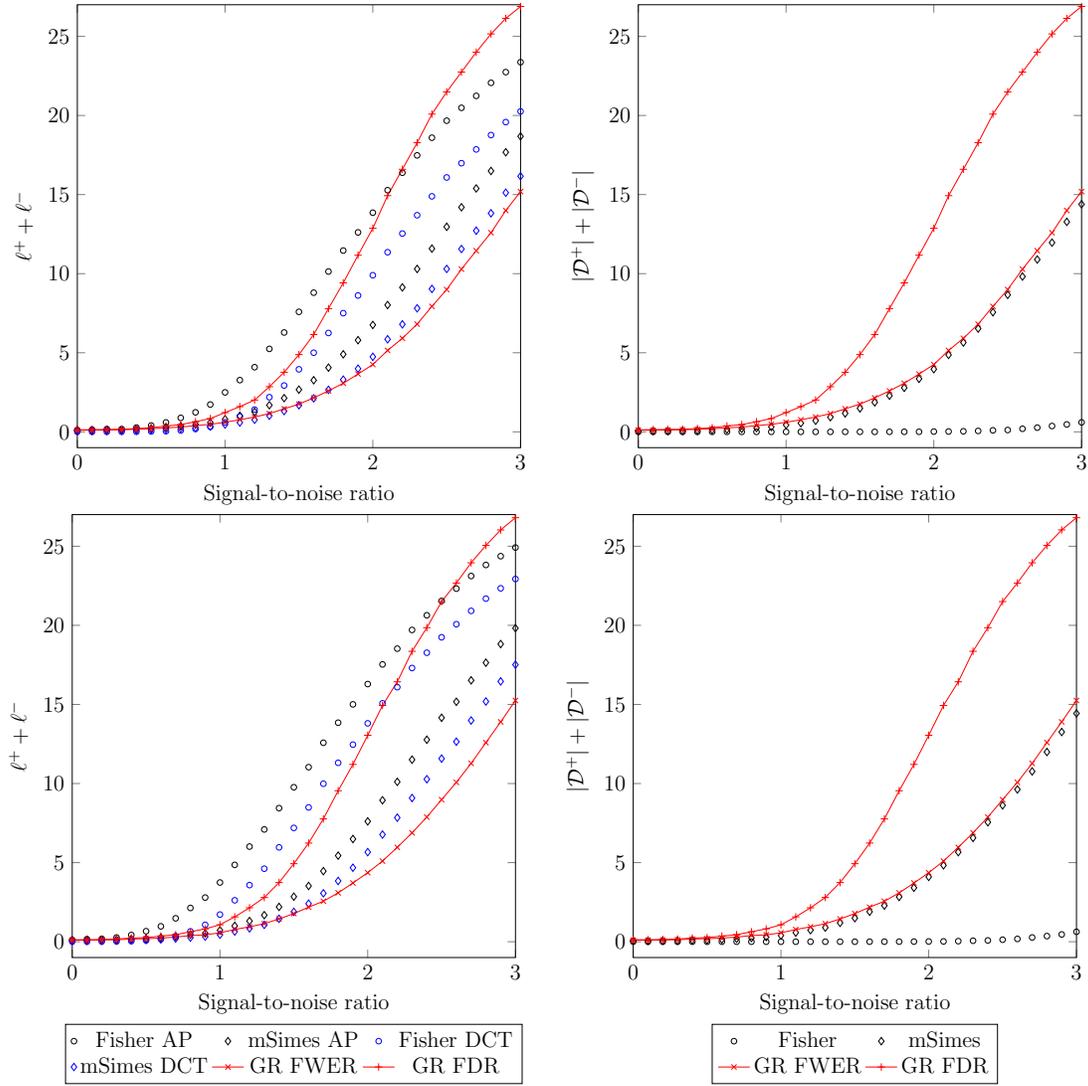

\section{Concluding remarks}\label{sec - disc}

In Procedure \ref{algo-conditional closed testing}, the selection Step 1 has a cost of inflating the $p$-values by a factor of 2. It is a reasonable cost, since it is no worse than the cost of 2-sided testing. 
It is possible to consider more generally the selection of $\{H_i^+: p_i\leq \tau \}\cup \{H_i^-: p_i> \tau \}$ for a $\tau \in (0,1)$. Such a selection step inflates the $p_i$s by $1/\tau$, and the $q_i$s by $1/(1-\tau)$. It may be useful when it is more important to detect one direction over another.   
More generally, each parameter may be accompanied by its own predefined selection threshold $\tau_i, i=1, \ldots, n$. In particular, it is possible to set $\tau_i \in \{0,1\}$ 
if the direction for inference is  a-priori known for  parameter $\theta_i$. Selection  is thus more general, and the directional closed testing or partitioning procedures are carried out as described, using the conditional $p$-values $\{p_i/\tau_i: p_i\leq \tau_i \} \cup \{q_i/(1-\tau_i): p_i> \tau_i \}$. Importantly, the same computational shortcuts can be used  on these conditional $p$-values, and inference is valid as long as $p_1,\ldots,p_n$ are  independent and uniformly valid.   

An open problem is how to deal with dependent $p$-values. For very specific structured dependencies, as discussed in \S~\ref{subsec-dep}, the procedures that condition on the vector of signs may be applied with the appropriate local test.  
Without conditioning, it is possible to apply the partitioning procedure as detailed in \S~\ref{subsec-partitioning-cb-general}, as long as the local test is valid for the specific dependency between the $p$-values. We provide a small example  in \S~\ref{sec-unconditionalParitioning}, and  leave for future research the more comprehensive examination of unconditional methods for dependent test statistics.

This work concentrated on post-hoc bounds. For simplicity of exposition, the examples we presented were small scale. But in principle the methods can be applied also in moderate and large scale testing. As shown in \S~\ref{sec - sim}, the bounds can be tighter than bounds obtained from FDR controlling procedures when the signal is weak and sparse. Specifically, our procedures can be useful for studies that examine longitudinal data on independent individuals. For example, in neuroimaging, the $p$-value from an experiment on an individual in a region of interest in the brain may be reported for hundreds of individuals. The lower bounds on $n^+(\mathcal I)$ and $n^-(\mathcal I)$ or $n^-(\mathcal I)+n^0(\mathcal I)$ may be of great interest, for $\mathcal I = [n]$ as well as for various $\mathcal I \subset [n]$.  

For large scale testing, methods that use negative controls or masking  have been suggested for inference on the signs. For example,  \cite{Leung23} (building upon the work of \citealt{Lei18}) suggests reflecting the $p$-values of the one-sided hypotheses around 0.25 or 0.75. For FDR control, these methods build on the knockoff+ method of \cite{Barber15}, so the estimated false discovery proportion in $\mathcal I$ is the estimated number of  negative controls in  $\mathcal I$ plus one, divided by $|\mathcal I|$.  The added one in the numerator makes these procedures less relevant for small scale testing or when very few parameters are  non-zero. An open question is how to extend their inference to post-hoc simultaneous bounds on  $n^+(\mathcal I)$ and $n^-(\mathcal I)$. For this purpose, it may be possible to build upon the work of \cite{Li22}. It will then be interesting to compare the two approaches: conditioning on the signs and using adjusted $p$-values as in the current manuscript,  versus conditioning on the masked $p$-values and using a knockoff approach.

In a high dimensional meta-analysis, where the studies are independent but each study examines many (say $m$) features, can we first select features, and then provide simultaneous post-hoc bounds? For example, for the neuroimaging application, how can we  address multiple regions of interest, and provide tight bounds for each region while controlling a relevant error measure? 
If $|\mathcal R|$ features out of $m$ are selected, and the post-hoc bounds are constructed at level $|\mathcal R| \alpha/m$, then a relevant error is controlled. This error is related to the 
 average FWER over the selected, i.e.,  the expected proportion of features with at least one erroneous rejection out of all the selected features, suggested in \cite{benjamini14}. Briefly,  the error we consider is the expected proportion of at least one non-covering post-hoc bound, over the selected features.  We plan to carry out in the future a detailed analysis of this procedure, and its theoretical guarantees for specific and general dependencies.

\bibliography{bib}
\bibliographystyle{apalike}

\appendix

\section{Proofs}\label{Appendix-proofs}
\subsection{Proof of Proposition \ref{prop-conditionalCT}}
\begin{proof}
From step 2 of the procedure, it follows that  $\phi_{\mathcal{T}}$ is a function of $\{2p_i, i\in \mathcal S^-\cap \mathcal T \}\cup \{2q_j, j\in \mathcal S^+\cap \mathcal T \}$. The validity of the local test given the vector of signs follows from the fact that it combines valid $p$-values, since the validity of each conditional $p$-values is guaranteed by assumption $(A1)$, and they are independent by assumption $(A2)$. So $\forall \ \theta\in \mathbb{R}^n$, $$\mP_{\{\theta_i: i\in \mathcal T\}}(\phi_{\mathcal{T}} =1\mid  \{sign(p_i-1/2): i\in \mathcal T\})\leq \alpha.$$ 
Moreover, since the $p$-values are independent (assumption ($A2$)), it follows that 
\begin{eqnarray}
 \mP_{\theta}(\phi_{\mathcal{T}} =1\mid \mathcal S) =  \mP_{\{\theta_i: i\in \mathcal T\}}(\phi_{\mathcal{T}} =1\mid  \{sign(p_i-1/2): i\in \mathcal T\}). \nonumber     
\end{eqnarray}
Thus completing the proof. 

\end{proof}

\subsection{Proof of Proposition \ref{prop-partitioning}}
\begin{proof}
Suppose that the true $\theta$ lies in $\Theta_{\mathcal{K}}$ for some $\mathcal{K} \subseteq [n]$.  We have $\psi_{|\mathcal{K} \cap \mathcal{I}|}(\mathcal{I}) \leq \psi_{\mathcal{K}}$ for every $\mathcal{I}$. Therefore, if $\psi_{\mathcal{K}}=0$, then $|\mathcal{K} \cap \mathcal{I}| = n^+(\mathcal{I}) \in \mathcal{N}^+(\mathcal{I})_\alpha$ for all $\mathcal{I}$, thus
$$\mP_{\theta} ( n^+(\mathcal{I}) \in \mathcal{N}^+(\mathcal{I})_\alpha \mathrm{\,\,}\forall  \mathcal{I} ) \geq \mP_{\theta}(\psi_{\mathcal{K}}=0) \geq 1-\alpha.
$$
Furthermore, the event $\{n^+(\mathcal{I}) \in \mathcal{N}^+(\mathcal{I})_\alpha\}$ implies 
$\{\min(\mathcal{N}^+(\mathcal{I})_\alpha) \leq n^+(\mathcal{I})\leq  \max(\mathcal{N}^+(\mathcal{I})_\alpha)\}$ for all $\mathcal{I}$, thus $\pl^+_{\alpha}(\mathcal{I}) = \min(\mathcal{N}^+(\mathcal{I})_\alpha)$, $\pl^-_{\alpha}(\mathcal{I}) = |\mathcal{I}| - \max(\mathcal{N}^+(\mathcal{I})_\alpha)$ satisfy (\ref{eq-positiveonly-simultguarantee}). 
\end{proof}

\subsection{Proof of Proposition \ref{prop-partitioning-properties}}
\begin{proof}

Proposition \ref{prop-complexity}  shows that Algorithm \ref{algo-1} returns the bounds $\pl_\alpha^+(\mathcal{I})$ and $\pl_\alpha^-(\mathcal{I})$ derived from the partitioning procedure with adaptive local tests (\ref{eq-psi_K}) with at most $O(|\mathcal{I}|^2 \cdot \max(1,|\mathcal{I}^c|^2) )$ computation.

Suppose the true $\theta \in \Theta_{\mathcal{K}}$ for some $\mathcal{K} \subseteq [n]$, and let $\tilde{\mathcal{T}} =  \{\mathcal{K}^c \cap \mathcal{S}^-\} \cup \{\mathcal{K} \cap \mathcal{S}^+\}$. The partitioning procedure testing the orthant hypothesis $J_{\mathcal{K}}$ by using $\psi_\mathcal{K}=\phi_{\tilde{\mathcal{T}}}$, and
$$\displaystyle \sup_{\theta \in \Theta_\mathcal{K}} \mP_\theta (\phi_{\mathcal{K}} = 1 \mid \mathcal{S})\leq \sup_{\theta \in H_{\tilde{\mathcal{T}}}} \mP_\theta (\phi_{\tilde{\mathcal{T}}} = 1 \mid \mathcal{S})\leq \alpha$$
where the first inequality follows from $J_\mathcal{K} \subseteq H_{\tilde{\mathcal{T}}}$ and the second inequality follows from Proposition \ref{prop-conditionalCT}. This establishes control of type I error at level $\alpha$, conditional on the vector of signs $\mathcal{S}$, for the test $\phi_{\mathcal{K}}$ of the true orthant hypothesis $J_{\mathcal{K}}$.

Furthermore, 
\begin{eqnarray}
&& \mathbb{P}_{\theta}(\tilde{\mathcal{T}} = \emptyset )   = \mathbb{P}_{\theta}(\mathcal{S}^- = \mathcal{K} ) =\nonumber \\ &&  \prod_{i \in \mathcal{K}} \mathbb{P}_{\theta}(p_i \leq 0.5) \prod_{j \notin \mathcal{K}} \mathbb{P}_{\theta}(p_j > 0.5) 
\geq \mathbb{P}_{0}(\mathcal{S}^-=\mathcal{K}) \geq 2^{-n}.\label{eq-sm-unconditional-gain}
\end{eqnarray}
Thus
$$
 \mathbb{P}_{\theta}( \psi_{\mathcal{K}} =1 ) = \sum_{\mathcal{S} }  \mathbb{P}_{\theta}( \phi_{\tilde{\mathcal{T}}} =1 | \mathcal{S} ) \mathbb{P}_{\theta}(  \mathcal{S} ) \leq  \alpha (1-\mathbb{P}_{\theta}( \tilde{\mathcal{T}} = \emptyset )) \leq \alpha(1-2^{-n}).$$
This ensures unconditional control of type I error at level $\tilde{\alpha}=\alpha(1-2^{-n})$ for the test $\phi_{\mathcal{K}}$ of the true orthant hypothesis $J_{\mathcal{K}}$. See also § 1.2 of the Supplementary Material to \cite{Mohamad2020} for the unconditional type I error control with the adaptive likelihood ratio test. 
 
The previous results together with Proposition \ref{prop-partitioning} imply the conditional and unconditional coverage (\ref{eq-positiveonly-simultguarantee-conditional}) and  (\ref{eq-unconditionalcoverage}).

\end{proof}

\subsection{Proof of Proposition \ref{prop-AP-vs-DCT}}
\begin{proof}

In order to prove that $\ell^{+}_\alpha(\mathcal{I}) \leq \pl^{+}_{\alpha}(\mathcal{I})$ for every $\mathcal{I}$, we proceed by contradiction.  Assume that  $\exists \, \mathcal{I}$ such that $ \ell^{+}_\alpha(\mathcal{I}) > \pl^{+}_{\alpha}(\mathcal{I})$. We have $\ell^{+}_\alpha(\mathcal{I}) = \ell^{+}_\alpha(\mathcal{I} \cap \mathcal{S}^-)$ and by Lemma 1 in \cite{goeman2021}:
\begin{eqnarray}\label{contra}
    \ell^{+}_\alpha(\mathcal{I\cap \mathcal{S}^-}) = \min_{\mathcal{V} \subseteq [n]}(|(\mathcal{I}\cap \mathcal{S}^-) \setminus \mathcal{V}|: \phi_\mathcal{V} = 0).
\end{eqnarray}

Because the null hypothesis that $n^{+}(\mathcal{I})$ is equal to $\pl^{+}_{\alpha}(\mathcal{I})$ is not rejected at level $\alpha$, then $\exists \, \mathcal{K}$ such that $|\mathcal{K} \cap \mathcal{I}|= \pl^{+}_{\alpha}(\mathcal{I})$ for which   $J_{\mathcal{K}}$ is not rejected at level $\alpha$. Equivalently, $\exists \, \mathcal{U}=(\mathcal{K}^c \cap \mathcal{S}^-) \cup (\mathcal{K} \cap \mathcal{S}^+)$ such that $\phi_{\mathcal{U}} = 0$ and 
$$|(\mathcal{I}\cap \mathcal{S}^-) \setminus \mathcal{U}| = |\mathcal{I}\cap \mathcal{S}^- \cap \mathcal{K}| \leq \pl^{+}_{\alpha}(\mathcal{I}) < \ell^{+}_{\alpha}(\mathcal{I}).$$
On the other hand, from equation (\ref{contra}) it follows that $\ell^{+}_{\alpha}(\mathcal{I}) \leq |(\mathcal{I}\cap \mathcal{S}^-) \setminus \mathcal{U}|$ contradicting the assumption that $\pl^{+}_{\alpha}(\mathcal{I}) < \ell^{+}_\alpha(\mathcal{I})$. 

$\ell^{-}_{\alpha}(\mathcal{I}) \leq \pl^{-}_{\alpha}(\mathcal{I})$ for every $\mathcal{I}$ can be proved in analogous way.

In order to prove that $\pl^{+}_\alpha(\mathcal{I}) = \ell^{+}_{\alpha}(\mathcal{I})$ for every $\mathcal{I} \subseteq \mathcal{S}^-$, write:
\begin{eqnarray*}
    \pl_\alpha^+(\mathcal{I}) &=& \min\{v \in \{0,\ldots,n\}: \psi_v(\mathcal{I}) = 0\}\\
    &=& \min\{v \in \{0,\ldots,n\}: \min_{\mathcal{K} \subseteq [n]: |\mathcal{K} \cap \mathcal{I}|=v} \phi_{ (\mathcal{K}^c \cap \mathcal{S}^{-}) \cup (\mathcal{K} \cap \mathcal{S}^{+} )  }  = 0\}\\
    &=& 
    \min_{\mathcal{K} \subseteq [n]}\{|\mathcal{K} \cap \mathcal{I}| :  \phi_{ (\mathcal{K}^c \cap \mathcal{S}^{-}) \cup (\mathcal{K} \cap \mathcal{S}^{+} )  }  = 0\}\\
        &=& 
    \min_{\mathcal{V} \subseteq [n]}\{|\mathcal{V}^c \cap \mathcal{I}| :  \phi_{ (\mathcal{V} \cap \mathcal{S}^{-}) \cup (\mathcal{V}^c \cap \mathcal{S}^{+} )  }  = 0\}\\
\end{eqnarray*}
If  $\mathcal{I} \subseteq \mathcal{S}^-$, then $\mathcal{I} \setminus \{(\mathcal{V} \cap \mathcal{S}^{-}) \cup (\mathcal{V}^c \cap \mathcal{S}^{+} )\} = |\mathcal{I} \setminus \mathcal{V}|$ for any $\mathcal{V}$, thus
$$\pl_\alpha^+(\mathcal{I}) = \min_{\mathcal{V} \subseteq [n]}\{| \mathcal{I} \setminus \mathcal{V}| :  \phi_{ (\mathcal{V} \cap \mathcal{S}^{-}) \cup (\mathcal{V}^c \cap \mathcal{S}^{+} )  }  = 0\} = \min_{\mathcal{V} \subseteq [n]}\{|\mathcal{I} \setminus \mathcal{V}| :  \phi_{ \mathcal{V}}  = 0\}=\ell_\alpha^+(\mathcal{I}).$$ 

$\pl^{-}_\alpha(\mathcal{I}) = \ell^{-}_{\alpha}(\mathcal{I})$ for every $\mathcal{I} \subseteq \mathcal{S}^+$  can be proved in analogous way.

\end{proof}

\section{Case $n=2$}\label{Appendix-moren2}

Figure  \ref{Figure:fisherelementary} displays the rejection regions for base hypotheses for the partitioning with adaptive local tests (top left plot) and directional closed testing (top right plot) using Fisher's combining functions.

We call a closed testing  procedure \emph{consonant} if the rejection of the intersection hypothesis $H_{\mathcal{I}}$ implies the rejection of at least one of its component hypotheses $H_i$, $i\in \mathcal{I}$ \citep{romano2011consonance}.
From the top-right and bottom-right plots of Figure \ref{Figure:fisherelementary}, we can observe that the directional closed testing procedure with Fisher's combination is also not consonant. For instance, when $\alpha=0.2$ as shown in the plots, if $(p_1,p_2) = (0.101,0.101)$, then $H_{1}^- \cap H_{2}^-$ is rejected by Fisher's combination test $f_{\mathrm{Fisher}}(2p_1,2p_2) = 0.1713 < \alpha$. However, neither $H_1^-$ nor $H_2^-$ is rejected because $2p_i > \alpha$, $i = 1,2$.

The directional closed testing procedure with Simes' local tests is also not consonant for $n>2$, although it is consonant for $n=2$, as can be seen from Figure \ref{Figure:simesregions}.

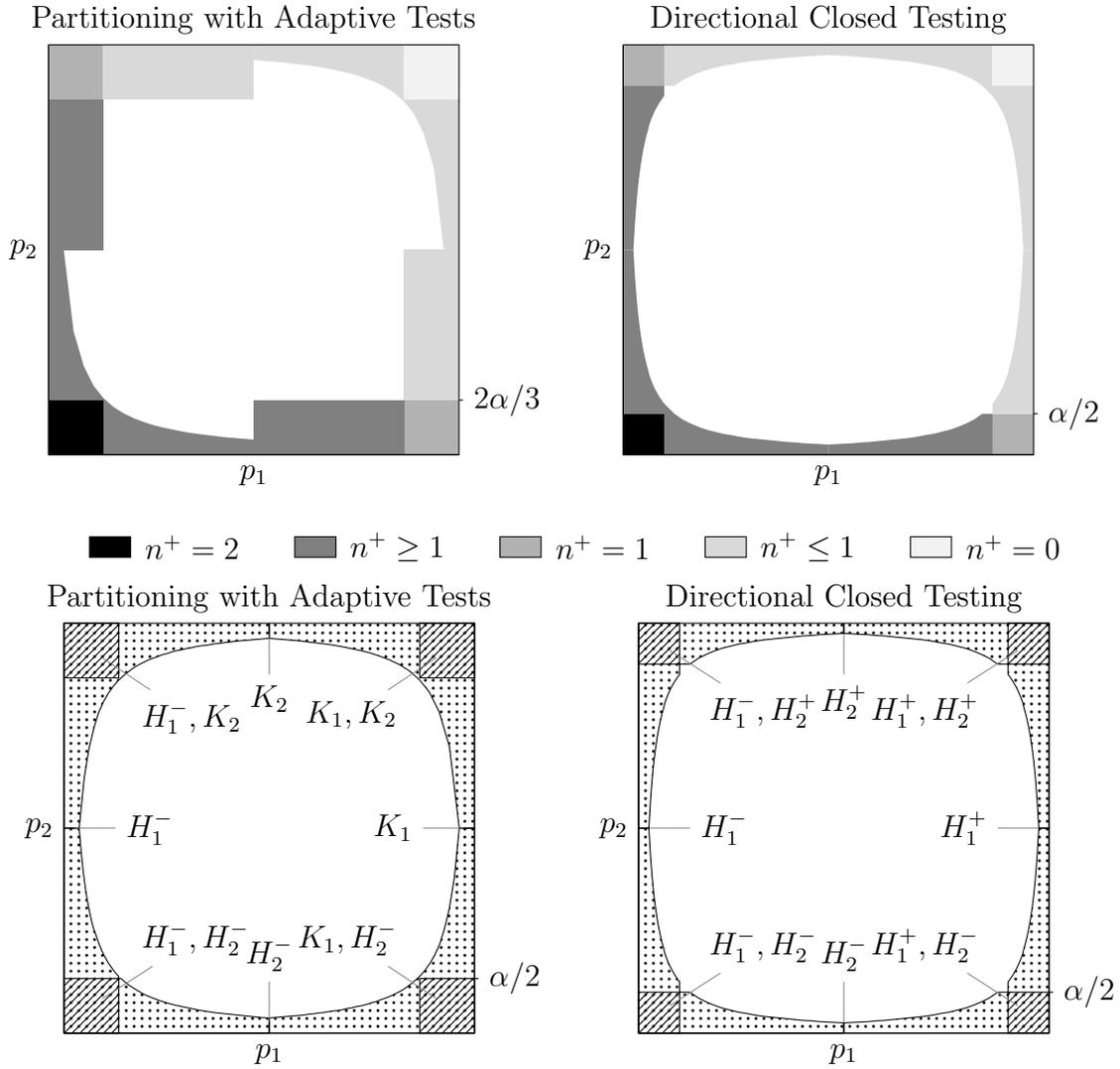
\begin{figure}
\centering
\def\alpgr{0.05}
\def\gam{0.1}
\def\alp{0.07804555}

\begin{tikzpicture}[scale=5.5]
\begin{scope}[yshift = -1.4 cm, xshift=1.4 cm]

\draw (0.5,1) node[above]{Directional Closed Testing};

\fill[fill=gray] plot coordinates { 
(0,0) (0,0.5) ( 0.025,0.5 ) ( 0.029,0.432 ) ( 0.033,0.38 ) ( 0.037,0.339 ) ( 0.041,0.307 ) ( 0.045,0.28 ) ( 0.049,0.257 ) ( 0.053,0.238 ) ( 0.057,0.221 ) ( 0.061,0.207 ) ( 0.064,0.194 ) ( 0.068,0.183 ) ( 0.072,0.173 ) ( 0.076,0.164 ) ( 0.08,0.156 ) ( 0.084,0.149 ) ( 0.088,0.142 ) ( 0.092,0.136 ) ( 0.096,0.13 ) ( 0.1,0.125 )
( 0.125,0.1 ) ( 0.145,0.086 ) ( 0.165,0.076 ) ( 0.184,0.068 ) ( 0.204,0.061 ) ( 0.224,0.056 ) ( 0.244,0.051 ) ( 0.263,0.048 ) ( 0.283,0.044 ) ( 0.303,0.041 ) ( 0.322,0.039 ) ( 0.342,0.037 ) ( 0.362,0.035 ) ( 0.382,0.033 ) ( 0.401,0.031 ) ( 0.421,0.03 ) ( 0.441,0.028 ) ( 0.461,0.027 ) ( 0.48,0.026 ) ( 0.5,0.025 )
(0.5,0) (0,0)
};
\fill[fill=gray!30] plot coordinates { 
(1,1) (1,0.5) 
( 0.975,0.5 ) ( 0.971,0.568 ) ( 0.967,0.62 ) ( 0.963,0.661 ) ( 0.959,0.693 ) ( 0.955,0.72 ) ( 0.951,0.743 ) ( 0.947,0.762 ) ( 0.943,0.779 ) ( 0.939,0.793 ) ( 0.936,0.806 ) ( 0.932,0.817 ) ( 0.928,0.827 ) ( 0.924,0.836 ) ( 0.92,0.844 ) ( 0.916,0.851 ) ( 0.912,0.858 ) ( 0.908,0.864 ) ( 0.904,0.87 ) 
( 0.9,0.875 )
( 0.875,0.9 ) 
( 0.855,0.914 ) ( 0.835,0.924 ) ( 0.816,0.932 ) ( 0.796,0.939 ) ( 0.776,0.944 ) ( 0.756,0.949 ) ( 0.737,0.952 ) ( 0.717,0.956 ) ( 0.697,0.959 ) ( 0.678,0.961 ) ( 0.658,0.963 ) ( 0.638,0.965 ) ( 0.618,0.967 ) ( 0.599,0.969 ) ( 0.579,0.97 ) ( 0.559,0.972 ) ( 0.539,0.973 ) ( 0.52,0.974 ) ( 0.5,0.975 ) 
(0.5, 1)
};
\fill[fill=black] (0,0) rectangle (\gam,\gam);
\fill[fill=gray!10] (1,1) rectangle (1-\gam,1-\gam);
\fill[fill=gray!60] (0,1) rectangle (\gam,1-\gam);
\fill[fill=gray!60] (0,1) rectangle (\gam,1-\gam);
\fill[fill=gray!60] (1-\gam,0) rectangle (1,\gam);
\fill[fill=gray!60] (1-\gam,0) rectangle (1,\gam);

\fill[fill=gray] plot coordinates { (0,0.9) (0,0.5) 
( 0.025,0.5 ) ( 0.029,0.568 ) ( 0.033,0.62 ) ( 0.037,0.661 ) ( 0.041,0.693 ) ( 0.045,0.72 ) ( 0.049,0.743 ) ( 0.053,0.762 ) ( 0.057,0.779 ) ( 0.061,0.793 ) ( 0.064,0.806 ) ( 0.068,0.817 ) ( 0.072,0.827 ) ( 0.076,0.836 ) ( 0.08,0.844 ) ( 0.084,0.851 ) ( 0.088,0.858 ) ( 0.092,0.864 ) ( 0.096,0.87 ) ( 0.1,0.875 )
(0.1,0.9)};

\fill[fill=gray!30] plot coordinates { (0.1,1) (0.1,0.9)
( 0.125,0.9 ) ( 0.145,0.914 ) ( 0.165,0.924 ) ( 0.184,0.932 ) ( 0.204,0.939 ) ( 0.224,0.944 ) ( 0.244,0.949 ) ( 0.263,0.952 ) ( 0.283,0.956 ) ( 0.303,0.959 ) ( 0.322,0.961 ) ( 0.342,0.963 ) ( 0.362,0.965 ) ( 0.382,0.967 ) ( 0.401,0.969 ) ( 0.421,0.97 ) ( 0.441,0.972 ) ( 0.461,0.973 ) ( 0.48,0.974 ) ( 0.5,0.975 )
(0.5,1)};

\fill[fill=gray!30] plot coordinates { (1,0.1) (1,0.5) 
( 0.975,0.5 ) ( 0.971,0.432 ) ( 0.967,0.38 ) ( 0.963,0.339 ) ( 0.959,0.307 ) ( 0.955,0.28 ) ( 0.951,0.257 ) ( 0.947,0.238 ) ( 0.943,0.221 ) ( 0.939,0.207 ) ( 0.936,0.194 ) ( 0.932,0.183 ) ( 0.928,0.173 ) ( 0.924,0.164 ) ( 0.92,0.156 ) ( 0.916,0.149 ) ( 0.912,0.142 ) ( 0.908,0.136 ) ( 0.904,0.13 ) ( 0.9,0.125 )
(0.9,0.1)};

\fill[fill=gray] plot coordinates { (0.9,0) (0.9,0.1)
( 0.875,0.1 ) ( 0.855,0.086 ) ( 0.835,0.076 ) ( 0.816,0.068 ) ( 0.796,0.061 ) ( 0.776,0.056 ) ( 0.756,0.051 ) ( 0.737,0.048 ) ( 0.717,0.044 ) ( 0.697,0.041 ) ( 0.678,0.039 ) ( 0.658,0.037 ) ( 0.638,0.035 ) ( 0.618,0.033 ) ( 0.599,0.031 ) ( 0.579,0.03 ) ( 0.559,0.028 ) ( 0.539,0.027 ) ( 0.52,0.026 ) ( 0.5,0.025 )
(0.5,0)};
\draw (1,\gam) -- (1.01,\gam) node[right]{$\alpha/2$};
\draw (0,0) -- (1,0) node[midway, below]{$p_1$};
\draw (0,0) -- (0,1) node[midway, left]{$p_2$};

\draw (0,0) rectangle (1,1);
\end{scope}

\def\alpgr{0.06666667}
\def\gam{0.1333333}
\def\alp{0.1040607}

\begin{scope}[yshift = -1.4 cm]

\draw (0.5,1) node[above]{Partitioning with Adaptive Tests};

\fill[fill=gray] plot coordinates { 
(0,0) (0,0.5) ( 0.037,0.5 ) ( 0.061,0.301 ) ( 0.086,0.216 ) ( 0.11,0.168 ) ( 0.134,0.138 ) ( 0.159,0.116 ) ( 0.183,0.101 ) ( 0.208,0.089 ) ( 0.232,0.08 ) ( 0.256,0.072 ) ( 0.281,0.066 ) ( 0.305,0.061 ) ( 0.329,0.056 ) ( 0.354,0.052 ) ( 0.378,0.049 ) ( 0.403,0.046 ) ( 0.427,0.043 ) ( 0.451,0.041 ) ( 0.476,0.039 ) ( 0.5,0.037 ) (0.5,0) (0,0)
};
\fill[fill=gray!30] plot coordinates { 
(1,1) (1,0.5) ( 0.963,0.5 ) ( 0.939,0.699 ) ( 0.914,0.784 ) ( 0.89,0.832 ) ( 0.866,0.862 ) ( 0.841,0.884 ) ( 0.817,0.899 ) ( 0.792,0.911 ) ( 0.768,0.92 ) ( 0.744,0.928 ) ( 0.719,0.934 ) ( 0.695,0.939 ) ( 0.671,0.944 ) ( 0.646,0.948 ) ( 0.622,0.951 ) ( 0.597,0.954 ) ( 0.573,0.957 ) ( 0.549,0.959 ) ( 0.524,0.961 ) ( 0.5,0.963 ) ( 0.5,0.975 ) 
(0.5, 1)
};
\fill[fill=black] (0,0) rectangle (\gam,\gam);
\fill[fill=gray!10] (1,1) rectangle (1-\gam,1-\gam);
\fill[fill=gray!60] (0,1) rectangle (\gam,1-\gam);
\fill[fill=gray!60] (0,1) rectangle (\gam,1-\gam);
\fill[fill=gray!60] (1-\gam,0) rectangle (1,\gam);
\fill[fill=gray!60] (1-\gam,0) rectangle (1,\gam);

\fill[fill=gray] (0,1-\gam) rectangle (\gam,0.5);
\fill[fill=gray] (0.5,0) rectangle (1-\gam,\gam);
\fill[fill=gray!30] (\gam,1) rectangle (0.5,1-\gam);
\fill[fill=gray!30] (1-\gam,\gam) rectangle (1,0.5);
\draw (1,\gam) -- (1.01,\gam) node[right]{$2\alpha/3$};
\draw (0,0) -- (1,0) node[midway, below]{$p_1$};
\draw (0,0) -- (0,1) node[midway, left]{$p_2$};
\draw (0,0) rectangle (1,1);

\end{scope}

\begin{scope}[yshift = -1.65 cm, xshift = 0.1cm]

\draw[fill=black] (0,0.05) rectangle (0.1,0);
\draw (0.25, 0.025) node{$n^+=2$};

\draw[fill=gray] (0.5,0.05) rectangle (0.6,0);
\draw (0.75, 0.025) node{$n^+\geq 1$};

\draw[fill=gray!60] (1,0.05) rectangle (1.1,0);
\draw (1.25, 0.025) node{$n^+= 1$};

\draw[fill=gray!30] (1.5,0.05) rectangle (1.6,0);
\draw (1.75, 0.025) node{$n^+\leq 1$};

\draw[fill=gray!10] (2,0.05) rectangle (2.1,0);
\draw (2.25, 0.025) node{$n^+ = 0$};
\end{scope}

\end{tikzpicture}
\def\alp{0.07804555}
\def\alpgr{0.05}
\def\gam{0.1}

\begin{tikzpicture}[scale=5.5]

\begin{scope}[yshift = -1.4 cm, xshift = 1.4cm]

\draw (0.5,1) node[above]{Directional Closed Testing};

\draw[pattern=dots] plot coordinates { 
(0,0) (0,0.5) 
( 0.025,0.5 ) ( 0.029,0.432 ) ( 0.033,0.38 ) ( 0.037,0.339 ) ( 0.041,0.307 ) ( 0.045,0.28 ) ( 0.049,0.257 ) ( 0.053,0.238 ) ( 0.057,0.221 ) ( 0.061,0.207 ) ( 0.064,0.194 ) ( 0.068,0.183 ) ( 0.072,0.173 ) ( 0.076,0.164 ) ( 0.08,0.156 ) ( 0.084,0.149 ) ( 0.088,0.142 ) ( 0.092,0.136 ) ( 0.096,0.13 ) ( 0.1,0.125 )
(0.1,0.1) 
( 0.125,0.1 ) ( 0.145,0.086 ) ( 0.165,0.076 ) ( 0.184,0.068 ) ( 0.204,0.061 ) ( 0.224,0.056 ) ( 0.244,0.051 ) ( 0.263,0.048 ) ( 0.283,0.044 ) ( 0.303,0.041 ) ( 0.322,0.039 ) ( 0.342,0.037 ) ( 0.362,0.035 ) ( 0.382,0.033 ) ( 0.401,0.031 ) ( 0.421,0.03 ) ( 0.441,0.028 ) ( 0.461,0.027 ) ( 0.48,0.026 ) ( 0.5,0.025 )
(0.5,0) (0,0) 
};

\draw[pattern=dots] plot coordinates { 
(1,0) (1,0.5) 
( 0.975,0.5 ) ( 0.971,0.432 ) ( 0.967,0.38 ) ( 0.963,0.339 ) ( 0.959,0.307 ) ( 0.955,0.28 ) ( 0.951,0.257 ) ( 0.947,0.238 ) ( 0.943,0.221 ) ( 0.939,0.207 ) ( 0.936,0.194 ) ( 0.932,0.183 ) ( 0.928,0.173 ) ( 0.924,0.164 ) ( 0.92,0.156 ) ( 0.916,0.149 ) ( 0.912,0.142 ) ( 0.908,0.136 ) ( 0.904,0.13 ) ( 0.9,0.125 )
(0.9,0.1)
( 0.875,0.1 ) ( 0.855,0.086 ) ( 0.835,0.076 ) ( 0.816,0.068 ) ( 0.796,0.061 ) ( 0.776,0.056 ) ( 0.756,0.051 ) ( 0.737,0.048 ) ( 0.717,0.044 ) ( 0.697,0.041 ) ( 0.678,0.039 ) ( 0.658,0.037 ) ( 0.638,0.035 ) ( 0.618,0.033 ) ( 0.599,0.031 ) ( 0.579,0.03 ) ( 0.559,0.028 ) ( 0.539,0.027 ) ( 0.52,0.026 ) ( 0.5,0.025 )
(0.5,0)};

\draw[pattern=dots] plot coordinates { 
(0,1) (0,0.5) 
( 0.025,0.5 ) ( 0.029,0.568 ) ( 0.033,0.62 ) ( 0.037,0.661 ) ( 0.041,0.693 ) ( 0.045,0.72 ) ( 0.049,0.743 ) ( 0.053,0.762 ) ( 0.057,0.779 ) ( 0.061,0.793 ) ( 0.064,0.806 ) ( 0.068,0.817 ) ( 0.072,0.827 ) ( 0.076,0.836 ) ( 0.08,0.844 ) ( 0.084,0.851 ) ( 0.088,0.858 ) ( 0.092,0.864 ) ( 0.096,0.87 ) ( 0.1,0.875 )
(0.1,0.9)
( 0.125,0.9 ) ( 0.145,0.914 ) ( 0.165,0.924 ) ( 0.184,0.932 ) ( 0.204,0.939 ) ( 0.224,0.944 ) ( 0.244,0.949 ) ( 0.263,0.952 ) ( 0.283,0.956 ) ( 0.303,0.959 ) ( 0.322,0.961 ) ( 0.342,0.963 ) ( 0.362,0.965 ) ( 0.382,0.967 ) ( 0.401,0.969 ) ( 0.421,0.97 ) ( 0.441,0.972 ) ( 0.461,0.973 ) ( 0.48,0.974 ) ( 0.5,0.975 )
(0.5,1)};

\draw[pattern=dots] plot coordinates { 
(1,1) (1,0.5) 
( 0.975,0.5 ) ( 0.971,0.568 ) ( 0.967,0.62 ) ( 0.963,0.661 ) ( 0.959,0.693 ) ( 0.955,0.72 ) ( 0.951,0.743 ) ( 0.947,0.762 ) ( 0.943,0.779 ) ( 0.939,0.793 ) ( 0.936,0.806 ) ( 0.932,0.817 ) ( 0.928,0.827 ) ( 0.924,0.836 ) ( 0.92,0.844 ) ( 0.916,0.851 ) ( 0.912,0.858 ) ( 0.908,0.864 ) ( 0.904,0.87 ) ( 0.9,0.875 )
( 0.9,0.9 )
( 0.875,0.9 ) ( 0.855,0.914 ) ( 0.835,0.924 ) ( 0.816,0.932 ) ( 0.796,0.939 ) ( 0.776,0.944 ) ( 0.756,0.949 ) ( 0.737,0.952 ) ( 0.717,0.956 ) ( 0.697,0.959 ) ( 0.678,0.961 ) ( 0.658,0.963 ) ( 0.638,0.965 ) ( 0.618,0.967 ) ( 0.599,0.969 ) ( 0.579,0.97 ) ( 0.559,0.972 ) ( 0.539,0.973 ) ( 0.52,0.974 ) ( 0.5,0.975 ) 
(0.5, 1)
};

\draw[pattern=north east lines] (0,0) rectangle (\gam,\gam);
\draw[pattern=north east lines] (1,1) rectangle (1-\gam,1-\gam);
\draw[pattern=north east lines] (0,1) rectangle (\gam,1-\gam);
\draw[pattern=north east lines] (1-\gam,0) rectangle (1,\gam);

\draw (\gam/2, \gam/2) node[pin=45:{$H_1^-,H_2^-$}]{};
\draw (\gam/2, 1-\gam/2) node[pin=315:{$H_1^-,H_2^+$}]{};
\draw (1-\gam/2, 1-\gam/2) node[pin=225:{$H_1^+,H_2^+$}]{};
\draw (1-\gam/2, \gam/2) node[pin=135:{$H_1^+,H_2^-$}]{};
\draw (0, 0.5) node[pin=right:{$H_1^-$}]{};
\draw (1, 0.5) node[pin=left:{$H_1^+$}]{};
\draw (0.5, 0) node[pin=above:{$H_2^-$}]{};
\draw (0.5, 1) node[pin=below:{$H_2^+$}]{};

\draw (0,0) -- (1,0) node[midway, below]{$p_1$};
\draw (0,0) -- (0,1) node[midway, left]{$p_2$};
\draw (1,\gam) -- (1.01,\gam) node[right]{$\alpha/2$};

\draw (0,0) rectangle (1,1);

\end{scope}


\def\alpgr{0.06666667}
\def\gam{0.1333333}

\begin{scope}[yshift = -1.4 cm]

\draw (0.5,1) node[above]{Partitioning with Adaptive Tests};

\draw[pattern=dots] plot coordinates { 
(0,0) ( 0, 0.5 ) ( 0.037,0.5 ) ( 0.042,0.44 ) ( 0.047,0.392 ) ( 0.052,0.354 ) ( 0.057,0.323 ) ( 0.062,0.297 ) ( 0.067,0.274 ) ( 0.072,0.255 ) ( 0.078,0.239 ) ( 0.083,0.224 ) ( 0.088,0.211 ) ( 0.093,0.199 ) ( 0.098,0.189 ) ( 0.103,0.18 ) ( 0.108,0.171 ) ( 0.113,0.164 ) ( 0.118,0.157 ) ( 0.123,0.15 ) ( 0.128,0.144 ) ( 0.133,0.139 )
(\gam,\gam)  
( 0.139,0.133 ) ( 0.158,0.117 ) ( 0.177,0.105 ) ( 0.196,0.094 ) ( 0.215,0.086 ) ( 0.234,0.079 ) ( 0.253,0.073 ) ( 0.272,0.068 ) ( 0.291,0.064 ) ( 0.31,0.06 ) ( 0.329,0.056 ) ( 0.348,0.053 ) ( 0.367,0.05 ) ( 0.386,0.048 ) ( 0.405,0.046 ) ( 0.424,0.044 ) ( 0.443,0.042 ) ( 0.462,0.04 ) ( 0.481,0.038 ) ( 0.5,0.037 )
(0.5,0) (0,0) 
};

\draw[pattern=dots] plot coordinates { 
(1,0) (1,0.5) 
( 0.963,0.5 ) ( 0.951,0.379 ) ( 0.939,0.305 ) ( 0.927,0.255 ) ( 0.916,0.219 ) ( 0.904,0.192 ) ( 0.892,0.171 ) ( 0.88,0.154 ) ( 0.868,0.14 ) ( 0.856,0.129 ) ( 0.844,0.119 ) ( 0.832,0.11 ) ( 0.821,0.103 ) ( 0.809,0.097 ) ( 0.797,0.091 ) ( 0.785,0.086 ) ( 0.773,0.082 ) ( 0.761,0.077 ) ( 0.749,0.074 ) ( 0.737,0.07 ) ( 0.726,0.067 ) ( 0.714,0.065 ) ( 0.702,0.062 ) ( 0.69,0.06 ) ( 0.678,0.057 ) ( 0.666,0.055 ) ( 0.654,0.054 ) ( 0.642,0.052 ) ( 0.631,0.05 ) ( 0.619,0.049 ) ( 0.607,0.047 ) ( 0.595,0.046 ) ( 0.583,0.044 ) ( 0.571,0.043 ) ( 0.559,0.042 ) ( 0.547,0.041 ) ( 0.536,0.04 ) ( 0.524,0.039 ) ( 0.512,0.038 ) ( 0.5,0.037 )
(0.5,0)};

\draw[pattern=dots] plot coordinates { 
(0,1) (0,0.5) 
( 0.037,0.5 ) ( 0.049,0.621 ) ( 0.061,0.695 ) ( 0.073,0.745 ) ( 0.084,0.781 ) ( 0.096,0.808 ) ( 0.108,0.829 ) ( 0.12,0.846 ) ( 0.132,0.86 ) ( 0.144,0.871 ) ( 0.156,0.881 ) ( 0.168,0.89 ) ( 0.179,0.897 ) ( 0.191,0.903 ) ( 0.203,0.909 ) ( 0.215,0.914 ) ( 0.227,0.918 ) ( 0.239,0.923 ) ( 0.251,0.926 ) ( 0.263,0.93 ) ( 0.274,0.933 ) ( 0.286,0.935 ) ( 0.298,0.938 ) ( 0.31,0.94 ) ( 0.322,0.943 ) ( 0.334,0.945 ) ( 0.346,0.946 ) ( 0.358,0.948 ) ( 0.369,0.95 ) ( 0.381,0.951 ) ( 0.393,0.953 ) ( 0.405,0.954 ) ( 0.417,0.956 ) ( 0.429,0.957 ) ( 0.441,0.958 ) ( 0.453,0.959 ) ( 0.464,0.96 ) ( 0.476,0.961 ) ( 0.488,0.962 ) ( 0.5,0.963 )
(0.5,1)};

\draw[pattern=dots] plot coordinates { 
(1,1) (1,0.5) ( 0.963,0.5 ) ( 0.939,0.699 ) ( 0.914,0.784 ) ( 0.89,0.832 ) ( 0.866,0.862 ) ( 0.841,0.884 ) ( 0.817,0.899 ) ( 0.792,0.911 ) ( 0.768,0.92 ) ( 0.744,0.928 ) ( 0.719,0.934 ) ( 0.695,0.939 ) ( 0.671,0.944 ) ( 0.646,0.948 ) ( 0.622,0.951 ) ( 0.597,0.954 ) ( 0.573,0.957 ) ( 0.549,0.959 ) ( 0.524,0.961 ) ( 0.5,0.963 ) ( 0.5,0.975 ) 
(0.5, 1)
};

\draw[pattern=north east lines] (0,0) rectangle (\gam,\gam);
\draw[pattern=north east lines] (1,1) rectangle (1-\gam,1-\gam);
\draw[pattern=north east lines] (0,1) rectangle (\gam,1-\gam);
\draw[pattern=north east lines] (1-\gam,0) rectangle (1,\gam);

\draw (\gam/2, \gam/2) node[pin=45:{$H_1^-,H_2^-$}]{};
\draw (\gam/2, 1-\gam/2) node[pin=315:{$H_1^-,K_2$}]{};
\draw (1-\gam/2, 1-\gam/2) node[pin=225:{$K_1,K_2$}]{};
\draw (1-\gam/2, \gam/2) node[pin=135:{$K_1,H_2^-$}]{};
\draw (0, 0.5) node[pin=right:{$H_1^-$}]{};
\draw (1, 0.5) node[pin=left:{$K_1$}]{};
\draw (0.5, 0) node[pin=above:{$H_2^-$}]{};
\draw (0.5, 1) node[pin=below:{$K_2$}]{};

\draw (0,0) -- (1,0) node[midway, below]{$p_1$};
\draw (0,0) -- (0,1) node[midway, left]{$p_2$};
\draw (1,\gam) -- (1.01,\gam) node[right]{$\alpha/2$};

\draw (0,0) rectangle (1,1);

\end{scope}

\end{tikzpicture}
\caption{The top row illustrates regions leading to inference about $n^+$ with $(1-\alpha)$ confidence for the partitioning procedure with adaptive tests (top-left plot) and the directional closed testing procedure  (top-right plot) with Fisher's combining function. The gray scale indicates the type of inference. 
The bottom row illustrates regions leading to the rejection of base hypotheses for partitioning with adaptive local tests (bottom left plot) and directional closed testing (bottom right plot). 
The pattern indicates the number of hypotheses rejected at level $\alpha$: one (dots) or two (diagonal lines).
Each plot is based on $\alpha=0.2$. }
\label{Figure:fisherelementary}
\end{figure}

From Figure \ref{Figure:fishercorner}, it can be observed that when $\alpha=0.3$ and $(p_1,p_2) = (0.19,0.19)$, we are able to reject both $H_1^-\cap K_1$ and $K_1\cap H_2^-$ because $2p_i < 2\alpha/3$, where $i = 1,2$. However, we cannot reject $H_1^- \cap H_2^-$ since $f_{\mathrm{Fisher}}(2p_1,2p_2) = 0.423 > \alpha$. Therefore, we conclude that $n^+$ cannot be equal to 1, i.e., $n^+ \in \mathcal{N}^+_\alpha = \{0,2\}$. In this case, however, if we calculate the lower and upper bounds for $n^+$ as $\min(\mathcal{N}^+_\alpha)$ and $\max(\mathcal{N}^+_\alpha)$, respectively, we obtain the uninformative result of $0$ and $2$.

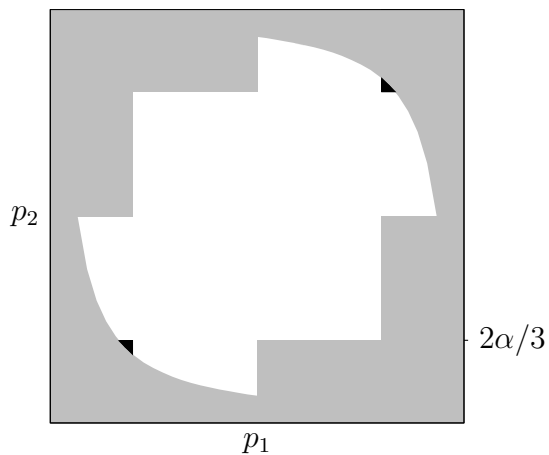
\begin{figure}
\centering
\def\alp{0.16334}
\def\alpgr{0.10}
\def\gam{0.2}

\begin{tikzpicture}[scale=5.5]
\begin{scope}[yshift = -1.4 cm, xshift=1.6 cm]

\draw (0.5,1) node[above]{Partitioning with Adaptive Tests};

\fill[fill=black] (0,0) rectangle (\gam,\gam);
\fill[fill=black] (1,1) rectangle (1-\gam,1-\gam);
\fill[fill=gray!50] plot coordinates { (0,0) (0,.5)  ( 0.066,0.5 ) ( 0.089,0.372 ) ( 0.112,0.296 ) ( 0.135,0.246 ) ( 0.158,0.21 ) ( 0.18,0.183 ) ( 0.203,0.163 ) ( 0.226,0.146 ) ( 0.249,0.133 ) ( 0.272,0.122 ) ( 0.295,0.112 ) ( 0.317,0.104 ) ( 0.34,0.097 ) ( 0.363,0.091 ) ( 0.386,0.086 ) ( 0.409,0.081 ) ( 0.432,0.077 ) ( 0.454,0.073 ) ( 0.477,0.069 ) ( 0.5,0.066 ) 
(0.5,0) (0,0) 
};
\fill[fill=gray!50] plot coordinates { (1,1) (1,0.5) ( 0.934,0.5 ) ( 0.911,0.628 ) ( 0.888,0.704 ) ( 0.865,0.754 ) ( 0.842,0.79 ) ( 0.82,0.817 ) ( 0.797,0.837 ) ( 0.774,0.854 ) ( 0.751,0.867 ) ( 0.728,0.878 ) ( 0.705,0.888 ) ( 0.683,0.896 ) ( 0.66,0.903 ) ( 0.637,0.909 ) ( 0.614,0.914 ) ( 0.591,0.919 ) ( 0.568,0.923 ) ( 0.546,0.927 ) ( 0.523,0.931 ) ( 0.5,0.934 ) (0.5,1)};
\fill[fill=gray!50] (0,1) rectangle (\gam,1-\gam);
\fill[fill=gray!50] (0,1) rectangle (\gam,1-\gam);
\fill[fill=gray!50] (1-\gam,0) rectangle (1,\gam);
\fill[fill=gray!50] (1-\gam,0) rectangle (1,\gam);
\fill[fill=black] plot coordinates { (\gam,\gam) (\gam,\alp) (\alp, \gam) };
\fill[fill=black] plot coordinates { (1-\gam,1-\gam) (1-\gam,1-\alp) (1-\alp, 1-\gam) };
\fill[fill=gray!50] (0,1-\gam) rectangle (\gam,0.5);
\fill[fill=gray!50] (0.5,0) rectangle (1-\gam,\gam);
\fill[fill=gray!50] (\gam,1) rectangle (0.5,1-\gam);
\fill[fill=gray!50] (1-\gam,\gam) rectangle (1,0.5);
\draw (1,\gam) -- (1.01,\gam) node[right]{$2\alpha/3$};
\draw (0,0) -- (1,0) node[midway, below]{$p_1$};
\draw (0,0) -- (0,1) node[midway, left]{$p_2$};
\draw (1, 0) -- (1, 1);
\draw (0,0) rectangle (1,1);
\end{scope}
\end{tikzpicture}
\caption{The partitioning procedure with adaptive tests and Fisher's combining function concludes $n^{+} \in \{0,2\}$ with $(1-\alpha)$ confidence when $(p_1,p_2)$ belongs to the regions in black. The plot is based on $\alpha=0.3$.}
\label{Figure:fishercorner}
\end{figure}

Finally, Figure \ref{Figure:GRn2} illustrates the rejection regions for $n=2$ of two FWER-controlling procedures for the family of hypotheses \eqref{eq-setup}. The first is the Holm-type procedure proposed by \cite{guo2015stepwise}, which requires conditions (A1)-(A2) in \S~\ref{sec-setup}. The second procedure is the $\mathrm{\check{S}id\acute{a}k}$-type procedure, also known as the ``Sp cross'' originally introduced by \cite{spjotvoll1972optimality} and modified by \cite{bohrer1980optimal}, which requires require only independence of the $n$ $p$-values. 
See \S~\ref{sec-unconditionalParitioning} for its derivation as a partitioning procedure.

\begin{figure}
\centering
\def\alp{0.1055728}

\begin{tikzpicture}[scale=5.5]

\begin{scope}[yshift = -1.4 cm, xshift=1.4 cm]

\draw (0.5,1) node[above]{ \cite{spjotvoll1972optimality}  };

\fill[fill=gray] (0,0) rectangle (1,\alp);
\fill[fill=gray] (0,0) rectangle (\alp,1);
\fill[fill=gray!30] (0,1) rectangle (1,1-\alp);
\fill[fill=gray!30] (1,1) rectangle (1-\alp,0);

\fill[fill=black] (0,0) rectangle (\alp,\alp);
\fill[fill=gray!10] (1,1) rectangle (1-\alp,1-\alp);
\fill[fill=gray!60] (0,1) rectangle (\alp,1-\alp);
\fill[fill=gray!60] (1-\alp,0) rectangle (1,\alp);

\draw (1,\alp) -- (1.01,\alp) node[right]{{$1-(1-\alpha)^{1/2}$}};

\draw (0,0) -- (1,0) node[midway, below]{$p_1$};
\draw (0,0) -- (0,1) node[midway, left]{$p_2$};

\draw (0,0) rectangle (1,1);

\end{scope}


\def\alp{0.1055728}
\def\alpgr{0.09090909}
\def\gam{0.1666667}

\begin{scope}[yshift = -1.4 cm]

\draw (0.5,1) node[above]{ \cite{guo2015stepwise}};

\fill[fill=gray] (0,0) rectangle (1,\alpgr);
\fill[fill=gray] (0,0) rectangle (\alpgr,1);
\fill[fill=gray!30] (0,1) rectangle (1,1-\alpgr);
\fill[fill=gray!30] (1,1) rectangle (1-\alpgr,0);

\fill[fill=black] (0,0) rectangle (\gam,\alpgr);
\fill[fill=black] (0,0) rectangle (\alpgr,\gam);

\fill[fill=gray!10] (1,1) rectangle (1-\alpgr,1-\gam);
\fill[fill=gray!10] (1,1) rectangle (1-\gam,1-\alpgr);
\fill[fill=gray!60] (0,1) rectangle (\alpgr,1-\gam);
\fill[fill=gray!60] (0,1) rectangle (\gam,1-\alpgr);
\fill[fill=gray!60] (1-\alpgr,0) rectangle (1,\gam);
\fill[fill=gray!60] (1-\gam,0) rectangle (1,\alpgr);

\draw (1,\gam) -- (1.01,\gam) node[right]{$\alpha/(1+\alpha)$};
\draw (1,\alpgr) -- (1.01,\alpgr) node[right]{$\alpha/(2+\alpha)$};

\draw (0,0) -- (1,0) node[midway, below]{$p_1$};
\draw (0,0) -- (0,1) node[midway, left]{$p_2$};

\draw (0,0) rectangle (1,1);

\end{scope}

\end{tikzpicture}
\caption{
 Holm-type procedure of \cite{guo2015stepwise} (left plot) and $\mathrm{\check{S}id\acute{a}k}$-type procedure of \cite{spjotvoll1972optimality}  (right plot). Each plot is based on $\alpha=0.2$.}
\label{Figure:GRn2}
\end{figure}

\section{The partitioning procedure based on unconditional $p$-values}\label{sec-unconditionalParitioning}

The partitioning procedure based on unconditional $p$-values $p_1,\ldots,p_n$ and $q_1,\ldots,q_n$ tests each partitioning hypothesis $J_{\mathcal{K}}$ by
\begin{eqnarray*}
\psi_\mathcal{K}
&=&\mathds{1}\{  f(\{p_i, i\in \mathcal K^c  \}\cup \{q_j, j\in \mathcal K  \}) \leq \alpha \}.
\end{eqnarray*}

If we assume that $p_1, \ldots,p_n$ are independent, we can use $\mathrm{\check{S}id\acute{a}k}$'s combining function $f_{\mathrm{\check{S}id\acute{a}k}}(x_1,\ldots,x_d) = 1- (1-\min(x_1,\ldots,x_d))^d$ and obtain the following FWER controlling procedure: if $p_i \leq 1-(1-\alpha)^{(1/n)}$, declare $\theta_i > 0$; if  $q_i \leq 1-(1-\alpha)^{(1/n)}$, declare $\theta_i \leq 0$. The right plot in Figure \ref{Figure:GRn2} illustrates the rejection region for $n=2$. 

Several combining function are valid for arbitrary dependence among $p$-values $p_1, \ldots,p_n$ \citep{vovk2022admissible}. The most well-known one is arguably the Bonferroni's combining function $f_{\mathrm{Bonferroni}}(x_1,\ldots,x_d) = d\min(x_1,\ldots,x_d) \wedge 1$. When $f_{\mathrm{Bonferroni}}$ is used in the partitioning procedure, it gives the FWER controlling procedure of \cite{bauer1986multiple}: if $p_i \leq \alpha/n$, declare $\theta_i > 0$; if  $q_i \leq \alpha/n$, declare $\theta_i \leq 0$.

Positively correlated $p$-values $p_1,\ldots,p_n$ arises in comparing $n$ groups with a common control group in a one-way layout. For simplicity, let's assume that there are $m$ observations per group.
The parameters of interest are $\theta_i = \mu_i - \mu_0$ with corresponding test statistics $\hat{\theta}_i = (\hat{\mu}_i - \hat{\mu}_0)/\sqrt{2\hat{\sigma}^2/m}$ following $t_{nm-1}$, a Student's $t$ distribution with $nm -1$ degrees of freedom. 
The $p$-values are $p_i = \mP(t_{nm-1} \geq \hat\theta_i)$, for $i=1,\ldots,n$.
It is worth noting that the $p$-values are correlated since $(\hat{\theta}_1,\ldots, \hat{\theta}_n)$ follows an $n$-variate $t$ distribution with $nm -1$ degrees of freedom and correlation matrix with off-diagonal elements equal to $1/2$. \cite{dunnett1955multiple} test for $J_\mathcal{K}$ is given by
\begin{eqnarray}\label{eq-dunnett_partitioning}
    \psi_\mathcal{K}&=&\mathds{1}\{  \max(\{\hat{\theta}_i, i\in \mathcal K^c  \}\cup \{-\hat{\theta}_j, j\in \mathcal K  \}) > c_{\alpha}(\mathcal K) \}
\end{eqnarray}
where $c_{\alpha}(\mathcal K)$ represents the $(1-\alpha)$-quantile of the maximum of an $n$-variate $t$ distribution with a correlation matrix that has off-diagonal elements whose sign depends on $\mathcal K$.

To illustrate this approach, we consider the \texttt{recovery} data from \cite{bretz2016multiple}. A company conducted a study to evaluate the effectiveness of specialized heating blankets in assisting post-surgical body heat. Four types of blankets, labeled b0, b1, b2, and b3, were tested on surgical patients to measure their impact on recovery times. The b0 blanket served as the standard option already utilized in different hospitals. The main focus was to assess the recovery time, measured in minutes, of patients randomly assigned to one of the four treatments. A shorter recovery time would indicate a more effective treatment. The key question addressed in this study is whether blanket types b1, b2, or b3 modify recovery time compared to b0. To perform a formal analysis of the data, we assume an one-way layout
$$y_{ij} = \gamma + \mu_i + \epsilon_{ij}$$
where $\gamma + \mu_i$ represents the expected recovery time for blanket b$i$, $i=0,1,2,3$,
with i.i.d. errors $\epsilon_{ij} \sim N(0,\sigma^2)$.

\cite{dunnett1955multiple}'s procedure is the standard procedure used to address the many-to-one comparisons problem. In the absence of specific information about the direction of the effect, we typically formulate two-sided null hypotheses of the form $H_i: \theta_i = 0$, where $\theta_i = \mu_i - \mu_0$ and $i=1,2,3$.

Table \ref{tab:dunnett_twosided} presents the multiplicity-adjusted $p$-values for both single-step and step-down Dunnett's procedures applied to the point null hypotheses $H_i: \theta_i = 0$. Although the step-down Dunnett procedure is more powerful, the single-step Dunnett's procedure provides simultaneous confidence intervals. These intervals enable directional conclusions to be drawn without the risk of Type III errors \citep{finner2021partitioning}.  At level $\alpha = 0.1$, we conclude that blanket b2 leads to significantly lower recovery times as compared to the standard blanket b0. 
The results of one-sided Dunnett's procedures are also provided for comparison, assuming a priori specific information about the direction of the effect. 

\begin{table}
\caption{\label{tab:dunnett_twosided}Dunnett's procedures for the \texttt{recovery} data: single-step and step-down adjusted $p$-values for two-sided and one-sided hypotheses.}
\centering
\begin{tabular}{c cc cc}
\\
Hypothesis &  Test statistic & $p$-value & \multicolumn{2}{c}{Dunnett's adjusted $p$-value}  \\
 &   &  & single-step & step-down \\
\hline
\\
$H_1: \theta_1 = 0$ & $-1.330$ & 0.191 & 0.456  &  0.192\\
$H_2: \theta_2 = 0$ & $-4.656$ &  $<0.001$ & $<0.001$ & $<0.001$\\
$H_3: \theta_3 = 0$ & $-1.884$ & 0.067  & 0.884 & 0.127\\
\\
\hline
\\
$H^+_1: \theta_1 \geq 0$ & $-1.330$ & 0.096 & 0.241  &  0.096\\
$H^+_2: \theta_2 \geq 0 $ & $-4.656$ &  $<0.001$ & $<0.001$ & $<0.001$\\
$H^+_3: \theta_3 \geq 0$ & $-1.884$ & 0.034  & 0.092 & 0.064\\
\\
\hline
\\
$H^-_1: \theta_1 \leq 0$ & $-1.330$ & 0.904 & 0.996  &  0.996\\
$H^-_2: \theta_2 \leq 0 $ & $-4.656$ &  1 & 1 & 1\\
$H^-_3: \theta_3 \leq 0$ & $-1.884$ & 0.966  & 1 & 0.997\\
\end{tabular}
\end{table}

Table \ref{tab:dunnett_partitioning} shows the $p$-value for the orthant hypotheses using Dunnett's tests (\ref{eq-dunnett_partitioning}). The computation of $p$-values is performed with the \texttt{glht} function of the \texttt{multcomp} R package. Table \ref{tab:dunnett_partitioning_adjp} 
gives multiplicity adjusted $p$-values for each pair $H_i^-: \theta_i \leq 0$, $K_i: \theta_i >0$. At level $\alpha = 0.1$, we conclude that blanket b2 and b3 lead to recovery times that are at most as large as  the standard blanket b0.

\begin{table}
\caption{\label{tab:dunnett_partitioning}Partitiong procedure using Dunnett's tests:
$p$-values for orthant hypotheses for the \texttt{recovery} data.}
\centering
\begin{tabular}{cccc}
\\
Orthant hypothesis & Test statistic & Critical value $\alpha=0.1$ & $p$-value \\
\hline
\\
$K_1 \cap K_2 \cap K_3$ & $4.656$ & 1.843 & $<0.001$ \\ 
$K_1 \cap K_2 \cap H_3^-$ & $4.656$ & $1.873$ & $<0.001$ \\
$H_1^- \cap K_2 \cap K_3$ & $4.656$ & $1.867$ & $<0.001$  \\
$K_1 \cap H_2^- \cap K_3$ & $1.884$ & $1.867$ & $0.097$ \\
$K_1 \cap H_2^- \cap H_3^-$ & $1.330$ &$1.867$ & $0.259$ \\
$H_1^- \cap K_2 \cap H_3^-$  & $4.656$ & $1.866$ & $<0.001$ \\ 
$H_1^- \cap H_2^- \cap K_3$ & $1.884$ & $1.874$ & $0.098$ \\
$H_1^- \cap H_2^- \cap H_3^-$ & $-1.330$ & $1.843$ & $0.996$ \\ 
\end{tabular}
\end{table}

\begin{table}
\caption{\label{tab:dunnett_partitioning_adjp}Partitiong procedure using Dunnett's tests:
adjusted $p$-values for base hypotheses for the \texttt{recovery} data.}
\centering
\begin{tabular}{cc | cc}
\\
Hypothesis &  Adjusted $p$-value & Hypothesis &  Adjusted $p$-value  \\
\hline
&&&\\
$H^-_1: \theta_1 \leq 0$ & 0.996 & $K_1: \theta_1 > 0$ &   0.259\\
$H^-_2: \theta_2 \leq 0$ & 0.996 & $K_2: \theta_2 > 0$ &  $<0.001$\\
$H^-_3: \theta_3 \leq 0$ & 0.996 &
$K_3: \theta_3 > 0$ &  0.098\\
\end{tabular}
\end{table}

\section{Exact shortcut for the partitioning procedure with adaptive tests}\label{sec-sm-shortcut}

We first discuss the computation of the lower bound $\pl^+_\alpha$ for $n^+$ and $\pl^-_{\alpha}$ for $n^- + n^0$.
Suppose we have observed $\mathcal{S}^-$ of size $|\mathcal{S}^-|=s$, and we want to check whether 
the test in (\ref{eq-psi_K})
rejects all the hypotheses $J_{\mathcal{K}}$ with $|\mathcal{K}|=k$ at level $\alpha$:
\begin{eqnarray*}
 f_k
=
\max_{\mathcal{K}: |\mathcal{K}|=k} f(
\{2p_i, i \in \mathcal{S}^- \cap \mathcal{K}^c  \}
\cup 
\{2q_i, i \in \mathcal{S}^+ \cap \mathcal{K} \}) &\leq& \alpha.
\end{eqnarray*}
Vandermonde's convolution
\begin{eqnarray*}\label{Vandermonde}
{n \choose k} &=& 
\sum_{v=0}^{k} {s \choose k-v} {n-s \choose v}
\end{eqnarray*}
states that for each $v=0,\ldots, k$ there are ${s \choose k-v} {n-s \choose v}$ sets $\mathcal{K}$ of size $k$ such that $|\mathcal{S}^+ \cap \mathcal{K}|=v$ and $|\mathcal{S}^- \cap \mathcal{K}|=k-v$
(or equivalently, $|\mathcal{S}^- \cap \mathcal{K}^c|=s-k+v$). 
Then, the maximization problem becomes
$$ f_k
= \max_{v \in \{\max(0,k-s),\ldots,\min(k,n-s)\}} \mathop{ \max_{ \mathcal{K}: 
|\mathcal{S}^+ \cap \mathcal{K}|=v, }}_{ |\mathcal{S}^- \cap \mathcal{K}|=k-v } f(
\{2p_i, i \in \mathcal{S}^- \cap \mathcal{K}^c  \}
\cup 
\{2q_i, i \in \mathcal{S}^+ \cap \mathcal{K} \}) 
$$
For any increasing function $f$,  the maximum has solution with the $v$ largest $p$-values $q_i$ with $i \in \mathcal{S}^+$ and the $s-k+v$ largest $p$-values $p_i$ with $i \in \mathcal{S}^-$, i.e. 
\begin{eqnarray}\label{eq-algo1_max}
 f_k
= \max_{v \in \{\max(0,k-s),\ldots,\min(k,n-s)\}} f( \{2p_{(k-v+1)},\ldots,2p_{(s)} \} \cup \{2q_{(n-s-v+1)},\ldots,2q_{(n-s)} \}).    
\end{eqnarray}
In order to compute $f_k$ in (\ref{eq-algo1_max}) for $k=0,\ldots,n$, we can use a nested loop, where the number of iterations of the inner loop (i.e. the index $v$ in (\ref{eq-algo1_max}) from $\max(0,k-s)$ to $\min(k,n-s)$) depends on the value of the outer loop's index (i.e. $k$ from 0 to $n$) and the size $s$ of $\mathcal{S}^-$. The total complexity for the two loops is $O(n^2)$. The maximum number of iterations happens when $s \in \{n/2-1,n/2,n/2+1\}$ if $n$ is even, and $s \in \{(n-1)/2, (n+1)/2\}$ if $n$ is odd; the mininum number of iterations happens when  $s \in \{0,n\}$.

The procedure we have just discussed is just a special case of the following \ref{algo-1} for the derivation of the lower bounds $\pl^+_\alpha(\mathcal{I})$ and $\pl^-_{\alpha}(\mathcal{I})$ for a generic subset $\mathcal{I}$. 
\newpage
\begin{algorithm}
\SetKwInOut{KwInit}{Initialize}
\SetKwInOut{Input}{Input}
\SetKwInOut{Output}{Output}
\Input{
right-tailed $p$-values $p_1,\ldots,p_n$; combining function $f(\cdot)$; subset $\mathcal{I} \subset [n]$; level $\alpha$}
\Output{Confidence bounds $l^+_\alpha(\mathcal{I})$ and $u^+_\alpha(\mathcal{I})$, $p$-values $f_{v}(\mathcal{I})$ for $J_v(\mathcal{I}): n^+(\mathcal{I})=v$, for $v=0,\ldots,|\mathcal{I}|$}
\label{algo-1}
\KwInit{

$\mathcal{S}^-=\{i \in [n]: p_i \leq 1/2\}$, $|\mathcal{S}^-|=s$, $f_{s}(\mathcal{I})=0$ \; 

$\mathbf{a}=(a_{1},\ldots,a_{|\mathcal{S}^- \cap \mathcal{I}|})$ with $a_{1}\geq \ldots \geq a_{|\mathcal{S} \cap \mathcal{I}|}$ sorted values of $\{2p_i: i \in \mathcal{S}^- \cap \mathcal{I}\}$ \;

$\mathbf{b}=(b_{1},\ldots,b_{|\mathcal{S}^- \cap \mathcal{I}^c|})$ with $b_{1}\geq \ldots \geq b_{|\mathcal{S}^- \cap \mathcal{I}^c|}$ sorted values of $\{2p_i: i \in \mathcal{S}^- \cap \mathcal{I}^c\}$ \;

$\mathbf{c}=(c_{1},\ldots,c_{|\mathcal{S}^+ \cap \mathcal{I}|})$ with $c_{1}\geq \ldots \geq c_{|\mathcal{S}^+ \cap \mathcal{I}|}$ sorted values of $\{2q_i: i \in \mathcal{S}^+ \cap \mathcal{I}\}$ \;

$\mathbf{d}=(d_{1},\ldots,d_{|\mathcal{S}^+ \cap \mathcal{I}^c|})$ with $d_{1}\geq \ldots \geq d_{|\mathcal{S}^+ \cap \mathcal{I}^c|}$ sorted values of $\{2q_i: i \in \mathcal{S}^+ \cap \mathcal{I}^c\}$ \;
}
\For{$v\gets0$ \KwTo $|\mathcal{I}|$}
{
    \For{$u\gets0$ \KwTo $|\mathcal{I}^c|$}
    {

\For{$k\gets \max(0,v-|\mathcal{S}^-\cap \mathcal{I}|)$ \KwTo $\min(v,|\mathcal{S}^+ \cap \mathcal{I}|)$}
{
\For{$j\gets \max(0,u-|\mathcal{S}^-\cap \mathcal{I}^c|)$ \KwTo $\min(u,|\mathcal{S}^+ \cap \mathcal{I}^c|)$}
{

    $f_{k,j} = f(\{a_1,\ldots,a_{k-v+|\mathcal{S}^-\cap \mathcal{I}|}\} \cup \{b_1,\ldots,b_{j-u+|\mathcal{S}^-\cap \mathcal{I}^c|}\} \cup \{c_1,\ldots,c_k\} \cup \{d_1,\ldots,d_j\})$ \;

}
}

 $f_{v,u}(\mathcal{I}) \gets \max\{f_{k,j}\}$  \;
    
    }
    $f_{v}(\mathcal{I}) = \max\{f_{v,u}(\mathcal{I}), u = 0,\ldots,|\mathcal{I}^c| \}$ \;
}
$\pl^+_\alpha(\mathcal{I}) \gets \min(v \in \{0,\ldots,s\}: f_{v}(\mathcal{I}) > \alpha)$ \;

$\pl^-_\alpha(\mathcal{I}) \gets |\mathcal{I}| - \max(v \in \{s,\ldots,|\mathcal{I}|\}: f_{v}(\mathcal{I}) > \alpha)$ \;

\Return $\pl^+_\alpha(\mathcal{I})$, $\pl^-_\alpha(\mathcal{I})$, $f_{0}(\mathcal{I}), \ldots, f_{|\mathcal{I}|}(\mathcal{I})$ \;
\caption{Shortcut for computing the lower bounds $\pl^+_\alpha(\mathcal{I})$ and $\pl^-_\alpha(\mathcal{I})$ derived from  adaptive local tests with partitioning.}
\end{algorithm}

\begin{proposition}\label{prop-complexity}
For any $\mathcal{I} \subseteq [n]$, Algorithm \ref{algo-1} returns the lower bounds $\pl^+_\alpha(\mathcal{I})$ and $\pl^-_\alpha(\mathcal{I})$ derived from the partitioning procedure with adaptive local tests  (\ref{eq-psi_K}), with at most $O(|\mathcal{I}|^2 \cdot \max(1,|\mathcal{I}^c|^2) )$ computation. In particular, calculation of the adjusted $p$-values for $H^-_i$ and $K_i$ requires $O(n^2)$ time. 
\end{proposition}

\begin{proof}
Suppose we have observed $\mathcal{S}^-$ of size $|\mathcal{S}^-|=s$, and for any $\mathcal I \subseteq [n]$ and any $v \in \{0,\ldots,|\mathcal{I}|\}$ we want to check whether 
the test in (\ref{eq-psi_K}) rejects all the hypotheses $J_\mathcal K$ with $|\mathcal K \cap \mathcal I|=v$
 at level $\alpha$ (or, equivalently, if $J_v(\mathcal{I}): n^+(\mathcal{I})=v$ is rejected at level $\alpha$): 
$$f_{v}(\mathcal{I}) = \max_{\mathcal{K}: |\mathcal{K} \cap \mathcal{I}|=v} f(
\{2p_i, i \in \mathcal{S}^- \cap \mathcal{K}^c  \}
\cup 
\{2q_i, i \in \mathcal{S}^+ \cap \mathcal{K} \}) \leq \alpha.$$
Algorithm \ref{algo-1} computes
$$ f_{v,u}(\mathcal{I}) = \mathop{ \max_{ \mathcal{K}: |\mathcal{K} \cap \mathcal{I}|=v, }}_{ |\mathcal{K} \cap \mathcal{I}^c|=u } f(
\{2p_i, i \in \mathcal{S}^- \cap \mathcal{K}^c  \}
\cup 
\{2q_i, i \in \mathcal{S}^+ \cap \mathcal{K} \})$$
so that $f_{v}(\mathcal{I}) = \max\{f_{v,u}(\mathcal{I}), u=0,\ldots,|\mathcal{I}^c| \}$. 
The function $f(\cdot)$ in (\ref{eq-combining_f}) combines   $2p_i$ with $i\in \mathcal{S}^- \cap \mathcal{K}^c$ and $2q_i$ 
with $i \in \mathcal{S}^+ \cap \mathcal{K}$. Writing $\mathcal{K}  = (\mathcal{K} \cap \mathcal{I}) \cup (\mathcal{K} \cap \mathcal{I}^c)$ and $\mathcal{K}^c  = (\mathcal{K}^c \cap \mathcal{I}) \cup (\mathcal{K}^c \cap \mathcal{I}^c)$ gives
\begin{eqnarray*}
f_\mathcal{K}(\mathcal{I})&=&f(
\{2p_i, i \in \mathcal{S}^- \cap \mathcal{K}^c  \}
\cup 
\{2q_i, i \in \mathcal{S}^+ \cap \mathcal{K} \})\\
&=& f(
\{2p_i, i \in \mathcal{S}^- \cap \mathcal{K}^c \cap \mathcal{I} \}
\cup 
\{2p_i, i \in \mathcal{S}^- \cap \mathcal{K}^c \cap \mathcal{I}^c \}
 \\
&&\quad \, \cup
\{2q_i, i \in \mathcal{S}^+ \cap \mathcal{K} \cap \mathcal{I} \} 
\cup 
\{2q_i, i \in \mathcal{S}^+ \cap  \mathcal{K} \cap \mathcal{I}^c \} ).    
\end{eqnarray*}
Consider Vandermonde's convolutions:
$$
{|\mathcal{I}| \choose v} = 
\sum_{k=0}^{v} { |\mathcal{S}^- \cap \mathcal{I}| \choose v-k} { |\mathcal{S}^+ \cap \mathcal{I}| \choose k}, \quad 
{|\mathcal{I}^c| \choose u} = 
\sum_{j=0}^{u} { |\mathcal{S}^- \cap \mathcal{I}^c| \choose u-j} { |\mathcal{S}^+ \cap \mathcal{I}^c| \choose j}.
$$
The first convolution states that for each $k\in\{0,\ldots,v\}$, there are ${ |\mathcal{S}^- \cap \mathcal{I}| \choose v-k} { |\mathcal{S}^+ \cap \mathcal{I}| \choose k}$ sets $\mathcal{K}$ such that $|\mathcal{K} \cap \mathcal{I}|=v$ with $|\mathcal{S}^+ \cap \mathcal{K} \cap \mathcal{I}|=k$ and $|\mathcal{S}^- \cap \mathcal{K} \cap \mathcal{I}|=v-k$. Likewise, the second convolution states that for each $j \in \{0,\ldots,u\}$,  there are ${ |\mathcal{S}^- \cap \mathcal{I^c}| \choose u-j} { |\mathcal{S}^+ \cap \mathcal{I}^c| \choose j}$ sets $\mathcal{K}$ such that $|\mathcal{K} \cap \mathcal{I}^c|=u$ with $|\mathcal{S}^+ \cap \mathcal{K} \cap \mathcal{I}^c|=j$ and $|\mathcal{S}^- \cap \mathcal{K} \cap \mathcal{I}^c|=u-j$. Then, the maximization problem becomes
\begin{eqnarray*}
   && f_{v,u}(\mathcal{I}) = \mathop{ \max_{ k \in \{k_1,\ldots,k_2\}, }}_{ j \in \{j_1,\ldots,j_2\} } 
\mathop{ \max_{ \mathcal{K} : |\mathcal{S}^+ \cap \mathcal{K} \cap \mathcal{I}|=k, |\mathcal{S}^- \cap  \mathcal{K} \cap \mathcal{I}| = v-k, } }_{ |\mathcal{S}^+ \cap  \mathcal{K} \cap \mathcal{I}^c|=j, |\mathcal{S}^- \cap \mathcal{K} \cap \mathcal{I}^c| = u-j } f_\mathcal{K}(\mathcal{I})
\end{eqnarray*}
where $k_1 = \max(0,v-|\mathcal{S}^-\cap \mathcal{I}|)$, $k_2=\min(v,|\mathcal{S}^+ \cap \mathcal{I}|)$, $j_1=\max(0,u-|\mathcal{S}^-\cap \mathcal{I}^c|)$ and $j_2=\min(u,|\mathcal{S}^+ \cap \mathcal{I}^c|)$.

For any increasing function $f(\cdot)$,  the maximum has solution with largest $k$ $p$-values $q_i$ with $i \in \mathcal{S}^+ \cap \mathcal{I}$, the largest $k-v+|\mathcal{S}^-\cap \mathcal{I}|$ $p$-values $p_i$ with $i \in \mathcal{S}^- \cap \mathcal{I}^c$, the largest $j$ $p$-values $q_i$ with $i \in \mathcal{S}^+ \cap \mathcal{I}$ and the largest $j-u+|\mathcal{S}^- \cap \mathcal{I}^c|$ $p$-values $p_i$ with $i \in \mathcal{S}^- \cap \mathcal{I}$: 
$$f_{v,u}(\mathcal{I}) = 
\mathop{ \max_{ k \in \{k_1,\ldots,k_2\} }}_{ j \in \{j_1,\ldots,j_2\} } 
f(\{a_1,\ldots,a_{k-v+|\mathcal{S}\cap \mathcal{I}|}\} \cup \{b_1,\ldots,b_{j-u+|\mathcal{S}\cap \mathcal{I}^c|}\} \cup \{c_1,\ldots,c_k\} \cup \{d_1,\ldots,d_j\})
$$
where 
$a_{1}\geq \ldots \geq a_{|\mathcal{S}^- \cap \mathcal{I}|}$, 
$b_{1}\geq \ldots \geq b_{|\mathcal{S}^- \cap \mathcal{I}^c|}$,
$c_{1}\geq \ldots \geq c_{|\mathcal{S}^+ \cap \mathcal{I}|}$ and 
$d_{1}\geq \ldots \geq d_{|\mathcal{S}^+ \cap \mathcal{I}^c|}$
denote the sorted values of $\{2p_i: i \in \mathcal{S}^- \cap \mathcal{I}\}$, $\{2p_i: i \in \mathcal{S}^- \cap \mathcal{I}^c\}$, 
$\{2q_i: i \in \mathcal{S}^+ \cap \mathcal{I}\}$ and 
$\{2q_i: i \in \mathcal{S}^+ \cap \mathcal{I}^c\}$, respectively.

Algorithm \ref{algo-1} evaluates $f_{v,u}(\mathcal{I})$ with a nested loop. The outer loop executes $\min(v,|\mathcal{S}^+ \cap \mathcal{I}|) - \max(0,v-|\mathcal{S}^- \cap \mathcal{I}|)$ times. Every time the outer loop executes, the inner loop executes $\min(v,|\mathcal{S}^+ \cap \mathcal{I}|) - \max(0,u-|\mathcal{S} \cap \mathcal{I}^c|)$ times. As a result, the complexity for evaluating $f_{v,u}(\mathcal{I})$ is $O(|\mathcal{I}| |\mathcal{I}^c|)$.

The complexity for computing  $f_{v}(\mathcal{I})$ for $v=0,\ldots, |\mathcal{I}|$ is $O(|\mathcal{I}|^2 |\mathcal{I}^c|^2)$ because it requires to compute $f_{v,u}(\mathcal{I})$ for $u=0,\ldots, |\mathcal{I}^c|$ and $v=0,\ldots, |\mathcal{I}|$. 
If $\mathcal{I} = \{i\}$, it takes $O(n^2)$ to compute the adjusted $p$-values $\bar{p}_i = f_{0}(\{i\})$ and $\bar{q}_i = f_{1}(\{i\})$ for $H^-_i$ and $K_i$, respectively.

\end{proof}

\section{Adaptive lower bounds for $n^+$ and $n^-+n^0$}\label{Appendix-overall bounds}
We  derive  lower bounds $l_{\alpha}^+$ and an  $l_{\alpha}^-$ for $n^+$ and $n^-+n^0$ such that 
\begin{eqnarray}\label{eq-coverage_plus-overall}
\mP_{\theta}\Big(
l_{\alpha}^{+}
\leq 
n^{+}, \ l_{\alpha}^{-}
\leq 
n^{-}+n^0 
\mid \mathcal S
 \Big) \geq 1- \alpha.
\end{eqnarray}
These bounds may be tighter than taking $\mathcal{I}=[n]$ with adaptive partitioning.  Interestingly, the  gap between these bounds and the bounds for  $\mathcal{I}=[n]$ in the adaptive partitioning procedure  is minimal, so adaptive partitioning is recommended if interest lies in positive and non-positive discoveries. This finding is based on a wide range of data generations examined (omitted for brevity). Typically, the analyst is  interested in positive and non-positive discoveries in addition to the lower bounds for $n^+$ and $n^-+n^0$. Since the cost of further inferences is minimal,  we suggest using adaptive partitioning in order to find the lower bounds for $n^+$ and $n^-+n^0$ and provide discoveries. Nevertheless,  procedure \ref{algo-bounds-overall} below should be used if interest is only in lower bounds for $n^+$ and $n^-+n^0$, since it is computationally much simpler and the bounds are at least as tight as with $\mathcal{I}=[n]$ in  procedure \ref{algo-1}, as formalized in Proposition \ref{prop-overallbounds-ap}.

We  briefly review tests of partial conjunction (PC) hypotheses, in order to set in context the novel Procedure \ref{algo-bounds-overall}. Let  $H^{r/n}: n^+ \leq r-1$, be
 the PC null hypothesis that at most $r-1$ hypotheses among $H^-_1,\ldots,H^-_n$ are false;  and $K^{r/n}:  n^-+n^0 \leq r-1$,  
the PC null hypothesis that at most $r-1$ hypotheses among $K_1,\ldots,K_n$ are false. For $r=1$,  $H^{r/n}$ is the global null hypothesis that none of the parameters are positive. For $r=2$, rejection of $H^{r/n}$ leads to establishing minimal replicability in the positive direction \citep{Benjamini09, Jaljuli2022}. 

Since $H^{r/n}$ is false if and only if every intersection hypothesis of size $n-r+1$ is false \citep{Benjamini09}, a valid $p$-value $p^{r/n}$ for $H^{r/n}$ is the largest intersection hypothesis $p$-value, over all intersections of $n-r+1$ null hypotheses:  $$p^{r/n} = \max_{\{\mathcal I: \mathcal I\subseteq[n], |\mathcal I| = n-r+1\}} p_{\mathcal I}, $$ where $p_{\mathcal I}$ is the $p$-value for the intersection hypothesis $\cap_{i\in \mathcal I} H^-_i$. 
Similarly,   a valid $p$-value $q^{r/n}$ for $K^{r/n}$ is $$q^{r/n} = \max_{\{\mathcal I: \mathcal I\subseteq[n], |\mathcal I| = n-r+1\}} q_{\mathcal I}, $$ where $q_{\mathcal I}$ is the $p$-value for the intersection hypothesis $\cap_{i\in \mathcal I} K_i$.

For  $f$ that satisfies the monotonicity  and symmetry condition (A0),
 $p^{r/n} = f(p_{(r)}, \ldots, p_{(n)})$ is a valid $p$-value, satisfying \begin{eqnarray*}
\sup_{\theta \in H^{r/n}} \mP_\theta(p^{r/n} \leq x)\leq x \quad \forall x \in [0,1].
\end{eqnarray*}
 The inequality is an equality, i.e., $p^{r/n}$ is uniformly distributed,  for the least favorable parameter configuration (LFC) $\theta_{LFC} \in H^{r/n}$, for  which the $p$-values corresponding to $r-1$ parameters  are zero (almost surely), and the $p$-values corresponding to the remaining $n-r+1$ parameters are  uniformly distributed. For example, in the normal means problem, the LFC configuration for testing $H^{r/n}$ is that $r-1$ parameters are infinite and $n-r+1$ parameters are zero; for testing $K^{r/n}$ the LFC is that $r-1$ parameters are minus infinity and $n-r+1$ parameters are zero.  

 For example, the PC $p$-values  using Fisher's combining method \citep{fisher1934statistical} are: 
\begin{eqnarray*}
p^{r/n} = \mP \Big( \chi^2_{2(n-r+1)} \geq -2 \sum_{k=r}^{n} \log(p_{(k)})\Big), \quad q^{r/n} = \mP \Big( \chi^2_{2(n-r+1)} \geq -2 \sum_{k=r}^{n} \log(q_{(k)})\Big)
\end{eqnarray*}
where $p_{(1)} \leq \ldots \leq p_{(n)}$ and $q_{(1)} \leq \ldots \leq q_{(n)}$ denote the sorted values of $p_1,\ldots,p_n$ and $q_1,\ldots,q_n$, respectively. Note that $1-p_{(k)} = q_{(n-k+1)}$ for continuous test statistics.
 
 By considering only hypotheses in $\mathcal S^-$ for $H^{r/n}$, we can avoid including $p$-values that are stochastically much larger than uniform when their null hypotheses are true. Therefore, as in \S~\ref{sec-setup}, we shall restrict ourselves to the directions guided by the data, so we shall use for testing $H^{r/n}$ 

\begin{equation}\label{eq-condPCp}
p^{r/n} = 
  \max_{\{\mathcal I: \mathcal I\subseteq[n], |\mathcal I| = n-r+1\}} f(\{2p_i: i\in\mathcal I\cap \mathcal S^-\}) =  
\begin{cases}
f\left(2p_{(r)}, \ldots, 2p_{(|\mathcal S^-|)}\right) & \text{if } r \leq |\mathcal S^-|,\\
1 & \text{otherwise}.
\end{cases}
\end{equation}
and for testing $K^{r/n}$ 

\begin{equation}\label{eq-condPCq}
q^{r/n} =  \max_{\{\mathcal I: \mathcal I\subseteq[n], |\mathcal I| = n-r+1\}} f(\{2q_i: i\in\mathcal I\cap \mathcal S^+\}) =  
\begin{cases}
f\left(2q_{(r)}, \ldots, 2q_{(|\mathcal S^+|)}\right) & \text{if } r \leq |\mathcal S^+|,\\
1 & \text{otherwise}.
\end{cases}
\end{equation}

\begin{algo}[Adaptive PC testing]\label{algo-bounds-overall} $ $
 \begin{enumerate}
 \item[Step 1] Apply Step 1 of procedure \ref{prop-conditionalCT}.
        \item[Step 2]  Test in order $\{H^{r/n}: r=1, \ldots, |\mathcal S^-|\}$, using $p^{r/n}$ in \eqref{eq-condPCp}, at level $\alpha$. Stop at the first non-rejection, $p^{r/n}>\alpha$.  Let $l^+_{\alpha}$ be the number of rejections, with $l^+_{\alpha} \in \{0,\ldots,|\mathcal S^-|\}$. If $l^+_{\alpha}=n$, return $l^+_{\alpha} = n, l^-_{\alpha}=0$, otherwise go to the next step.
\item[Step 3]  Test in order $\{K^{r/n}: r=1, \ldots, n-|\mathcal S^-|\}$, using $q^{r/n}$ in \eqref{eq-condPCq}, at level $\alpha$. Stop at the first non-rejection,  $q^{r/n}>\alpha$.  Let $l^-_{\alpha}$ be the number of rejections, with $l^-_{\alpha} \in \{0,\ldots,n-|\mathcal S^-|\}$.
 Return $l^+_{\alpha} $ and $l^-_\alpha$ (with $l^+_{\alpha} +l^-_{\alpha}\leq n$). 

 \end{enumerate}
 \end{algo}

Note that the bounds of the procedure will be the same if the testing in order is continued until $n$ in each step. This is so because
 $p^{(|\mathcal S^-|+1)/n}>\alpha$ and  $q^{(n-|\mathcal S^-|+1)/n}>\alpha$.

 The guaranteed coverage is formalized in the following proposition. 
\begin{proposition}\label{prop -overall conditional bounds}
 Let $\left\lbrace p^{r/n}: r\in[n] \right \rbrace $ be valid conditional PC $p$-values for $\{H^{r/n}: r\in[n]\}$, and  let $\left\lbrace q^{r/n}: r\in[n] \right \rbrace $ be valid conditional PC $p$-values for $\{K^{r/n}: r\in[n]\}$.  Then $l_{\alpha}^+$ and $l_{\alpha}^-$ satisfy \eqref{eq-coverage_plus-overall}. Furthermore, the unconditional coverage is 
\begin{eqnarray}\label{eq-coverage_plus-overall-unconditional}
\mP_{\theta}\Big(
l_{\alpha}^{+}
\leq 
n^{+}, \ l_{\alpha}^{-}
\leq 
n^{-}+n^0 
 \Big) \geq 1- (1-2^{-n})\alpha.
\end{eqnarray}
 
\end{proposition}
\begin{proof}
Suppose that $\theta$,  the true parameter vector, has  $n^+=t$ positive coordinates.  We have $l_{\alpha}^+ \leq t \leq n-l_{\alpha}^-$  if and only if $p^{(t+1)/n}> \alpha$ and $q^{(n-t+1)/n}> \alpha$.
The result follows since conditional on $\mathcal S^-$, it is only possible to make an error in one direction. Specifically, if $t>|\mathcal S^-|$, then  $p^{(|\mathcal S^-|+1)/n}>\alpha$, since it is not possible to reject $H^{(|\mathcal S^-|+1)/n}$, so it is not possible to err with regard to the lower bound. 
Therefore, if  $t>|\mathcal S^-|$, then $$\mP_{\theta}(t\notin[l_{\alpha}^+,n-l_{\alpha}^-]\mid \mathcal S) = \mP_{\theta}(t>n-l_{\alpha}^-\mid \mathcal S)\leq \mP_{\theta}(q^{(n-t+1)/n}\leq \alpha \mid \mathcal S)\leq \alpha. $$
Similarly, if $t<|\mathcal S^-|$, then  $q^{(n-|\mathcal S^-|+1)/n}>\alpha$, since it is not possible to reject $K^{(n-|\mathcal S^-|+1)/n}$, so it is not possible to err with regard to the upper bound.  Therefore, if  $t<|\mathcal S^-|$, then $$\mP_{\theta}(t\notin[l_{\alpha}^+,n-l_{\alpha}^-]\mid \mathcal S) = \mP_{\theta}(t<l_{\alpha}^+\mid \mathcal S)\leq \mP_{\theta}(p^{(t+1)/n}\leq \alpha \mid \mathcal S)\leq \alpha. $$
If $t = |\mathcal S^-|$, then since the lower bound is at most $ |\mathcal S^-|$ and the upper bound is at least $n- |\mathcal S^+|$, it is not possible to make an error on either bound. 
Therefore, the unconditional error of non-covering $n^+$ is
\begin{eqnarray}
 && \mP_{\theta}(t\notin[l_{\alpha}^+,n-l_{\alpha}^-]) = \mE\left( \mP_{\theta}(t\notin[l_{\alpha}^+,n-l_{\alpha}^-]\mid \mathcal S) \right)\nonumber \\ 
 && =\mE_{\theta}\left \lbrace \mI(t>|\mathcal S^-|)\mP_{\theta}(t>n-l_{\alpha}^-\mid \mathcal S^-)+ 
 \mI(t< |\mathcal S^-|)\mP_{\theta}(t<l_{\alpha}^+\mid \mathcal S^-)
 \right \rbrace \nonumber \\
 && \leq \alpha\mE_{\theta}\left(\mI(t> |\mathcal S^-|)+\mI(t< |\mathcal S^-|) \right) = \alpha \mP_{\theta}(t\neq |\mathcal S^-|)  \leq \alpha(1-2^{-n}), \nonumber
\end{eqnarray}
 where the  last inequality follows from \eqref{eq-sm-unconditional-gain}.
\end{proof}

\begin{proposition}\label{prop-overallbounds-ap}
Algorithm \ref{algo-1} returns the lower bounds $\pl^+_\alpha$ and $\pl^-_\alpha$. 
These lower bounds are at most as good as those obtained by the Procedure \ref{algo-bounds-overall}.
\end{proposition}

\begin{proof}
It is sufficient to note that
if $k<s$, then the index $v$ 
in (\ref{eq-algo1_max})
starts at $\max(0,k-s)=0$ by computing the PC conditional $p$-value $p^{k+1/n} = f( \{2p_{(k+1)},\ldots,2p_{(s)} \})$. 
Then $p^{k+1/n} > \alpha$ implies $$f_k=\max_{\mathcal{K}: |\mathcal{K}|=k} f(
\{2p_i, i \in \mathcal{S}^- \cap \mathcal{K}^c  \}
\cup 
\{2q_i, i \in \mathcal{S}^+ \cap \mathcal{K} \}) \geq p^{k+1/n} > \alpha$$ for any $k \in \{0,\ldots,s-1\}$, i.e. the lower bound $\ell^+_{\alpha}$ for $n^+$ from Algorithm \ref{algo-1} is smaller than or equal to the lower bound of Procedure \ref{algo-bounds-overall}.
Likewise, if $k>s$ then the index $v$ in (\ref{eq-algo1_max}) starts at $\max(0,k-s)=k-s$ by computing the PC conditional $p$-value $q^{n-k+1/n} = f( \{2q_{(n - k +1)},\ldots,2q_{(n-s)} \})$, thus $q^{n-k+1/n} > \alpha$ implies $$f_k=\max_{\mathcal{K}: |\mathcal{K}|=k} f(
\{2p_i, i \in \mathcal{S}^- \cap \mathcal{K}^c  \}
\cup 
\{2q_i, i \in \mathcal{S}^+ \cap \mathcal{K} \}) \geq q^{n-k+1/n} > \alpha$$ for any $k \in \{s+1,\ldots,n\}$, i.e. the lower bound $\pl^-_{\alpha}$ for $n^- + n_0$ from Algorithm \ref{algo-1} is less than or equal to the upper bound of Procedure \ref{algo-bounds-overall}.
\end{proof}

\begin{remark}
If $n^0>0$  then the probability that the lower bounds from Procedure \ref{algo-bounds-overall} do not cover at least one of $n^+, n^-$ may exceed $\alpha$. If $n^0 = n$ and all $p$ values are uniform then $\mP(p^{1/n}\leq \alpha \cup q^{1/n} \leq \alpha  )\approx 2\alpha$; as $n^0$ decreases the probability that the lower bounds do not cover at least one parameter decreases from $2\alpha$ to $\alpha$ (for $n_0=0$). 
\end{remark}

\begin{remark}
For a combination function $f()$, the test for qualitative interactions in \cite{Zhao2019} is rejected at level $\alpha$ if and only if $l_{\alpha}^+\geq 1$ and $l_{\alpha}^-\geq 1$ in the above procedure.  Therefore, if the assumption $n^0=0$ is reasonable, then the above procedure  complements nicely a conclusion that there is qualitative interaction, by providing with $(1-\alpha)$ confidence the  (interval) estimate of the parameter tested, $n^+$. More generally, a level $\alpha$ test of a generalized qualitative interaction null hypothesis that $n^+< a$ or $n^+>b$, for predefined $1\leq a<b\leq n-1 $, has the following rejection rule:  reject if $l_{\alpha}^+ \geq a$ and $n-l_{\alpha}^- \leq b$. To see that this is an $\alpha$ level test, consider the null value  $\theta$ such that $n^+(\theta) \notin[a,b]$. Without loss of generality, suppose $n^+(\theta)>b$. Then the probability of  falsely rejecting the generalized qualitative interaction true null hypothesis is $$\mP_{\theta} (l_{\alpha}^+\geq a \textrm{ and } n-l_{\alpha}^- \leq b) \leq  \mP_{\theta} ( n-l_{\alpha}^- \leq b)\leq \mP_{\theta} ( n-l_{\alpha}^- < n^+)\leq  \alpha.$$
\end{remark}

\subsection{A note on general confidence bounds for $n^+$}\label{subsec-nonadaptive PC testing}

If $p^{r/n}$ is a valid $p$-value for testing $H^{r/n}$ for $r=1,\ldots,n$, then 
\begin{eqnarray}\label{lower_PC}
l^+_{\alpha}(p)&=& \max\{l \in \{0,\ldots,n\}: p^{r/n} \leq \alpha \,\,\mathrm{for}\,\, r = 0,\ldots, l\}
\end{eqnarray}
satisfies $\mP_\theta \big( l_\alpha(p) \leq  n^+ \big) \geq 1-\alpha$ for all $\theta \in \Theta$, where $p^{0/n}\equiv 0$ since $H^{0/n}: n^+ < 0$ is always false. 
Analogously, 
if  $q^{r/n}$ is a valid  $p$-value for testing $K^{r/n}$ for $r=1,\ldots,n$, then
\begin{eqnarray}\label{upper_PC}
l_{\alpha}^-(p)&=& \max\{u \in \{0,\ldots,n\}: q^{r/n} \leq \alpha \,\,\mathrm{for}\,\, r = 0,\ldots,u\}
\end{eqnarray}
satisfies $\mP_\theta \big( n^+ \leq n-l^-_{\alpha}(p) \big) \geq 1-\alpha$ for all $\theta \in \Theta$, where $q^{0/n}\equiv 0$ since $K^{0/n}: n^+ > n$ is always false. For notational simplicity, we shall often write $l^+_{\alpha}$ and $l^-_{\alpha}$ instead of $l^+_{\alpha}(p)$ and $l^-_{\alpha}(p)$, but of course these bounds are functions of the $p$-value vector $p$. 

A straightforward application of the Bonferroni inequality shows that if level $\alpha/2$ is used for each bound, i.e.,  $l^+_{\alpha/2}$ in (\ref{lower_PC}) and $l^-_{\alpha/2}$ in (\ref{upper_PC}), then  
$\mP_\theta \big( l^+_{\alpha/2} \leq n^+ \leq n-l^-_{\alpha/2} \big) \geq 1-\alpha$. These bounds where used in \cite{Jaljuli2022} in order to complement meta-analyses in systematic reviews. 

The correction of using $\alpha/2$ in each direction (instead of $\alpha$, as in the adaptive PC testing procedure ) is, however, conservative. Intuitively, the correction should be less severe since in a given configuration, the probability of erring by exceeding one bound is much larger than the probability of erring by exceeding the other bound. To see this, note that $$\mP_{\theta}(n^+ \notin [l^+_{\alpha/2}, n-l^-_{\alpha/2}] = \mP_{\theta}(n^+<l^+_{\alpha/2})+\mP_{\theta}(n^+>n-l^-_{\alpha/2})$$ 
 has value $\alpha/2$ for the following least favorable parameter configurations (LFCs) when testing PC null hypotheses: 
the positive parameter value has $p_i=0$ (almost surely)  and the non-positive parameter value has a $p$-value with a uniform distribution; or the positive parameter value has a $p$-value which is (practically) uniformly distributed
 and the non-positive parameter value has $q_i=0$ (almost surely). To see this, note that 
 for the LFC with $p_i=0$  for $n^+$ parameters, $p^{(n^++1)/n}\sim U(0,1)$ and $q^{(n-n^++1)/n}=1$ almost surely.  For the LFC with $q_i = 0$ for $n-n^+$ parameters, $q^{(n-n^++1)/n}\sim U(0,1)$ and $p^{(n^++1)/n}=1$ almost surely. Non-coverage can occur if the lower bound is violated, so $p^{(n^++1)/n}\leq \alpha/2$, or if the upper bound is violated, so $q^{(n-n^++1)/n}\leq \alpha/2$.  Therefore  $$\mP_{LFC}(n^+\notin [l^+_{\alpha/2}, n-l^-_{\alpha/2}])  =  \mP(U\leq \alpha/2) = \alpha/2. $$

We conjecture that  the coverage guarantee is  typically $(1-\alpha)$ if the lower bounds are $l^+_{\alpha}$ and $l^-_{\alpha}$. In particular, whenever the test statistics are continuous, from one dimensional exponential families. 
Without loss of generality, suppose the first $t$ coordinates are positive, i.e., $\theta_1,\ldots,\theta_t>0$ and $\theta_{t+1},\ldots,\theta_n\leq 0$, and $n^+ =t$. Our conjecture is thus  that the solution to the following optimization problem is $\alpha$ for a large class of valid PC $p$-values:
\begin{eqnarray}\label{opt-problem}
\label{opt} \max_{\theta} && \mP_{\theta}(p^{(t+1)/n}\leq \alpha)+ \mP_{\theta} (q^{(n-t+1)/n}\leq \alpha) \nonumber\\
\mbox{s.t. } &&  \theta_i > 0, i=1,\ldots, t, \nonumber \\ 
 &&  \theta_j \leq 0, j=t+1,\ldots, n. \nonumber
\end{eqnarray}
To see that this is not true in general for unconditional combination tests (i.e., using local tests that do not condition on the vector of signs $\mathcal S$), consider the following stylized example. Let $n=2$ and $n^+=1$ be such that $\theta_1=0$ and $\theta_2$ is positive. Assume that the distribution of a $p$-value from  $H_i$, $x_i$,  has the following distribution: $\mP(x_i=\alpha) = \alpha, \mP(x_i=1) = 1-\alpha.  $ Similarly,  the  distribution of a $p$-value  from  $K_i$, $y_i=1-x_i$,  has distribution: $\mP(y_i=\alpha) = \alpha, \mP(y_i=1) = 1-\alpha.  $ These $p$-values are valid since $$\mP_{H_i}(p_i\leq a)\leq a, \ \mP_{K_i}(q_i\leq a)\leq a, \ \forall a\in [0,1].$$  The unconditional  PC $p$-value (i.e., it is derived from a local test that does not condition on the vector of signs $\mathcal S$) for $H^{2/2}$ and $K^{2/2}$  is, respectively, $p^{2/2}=\max(p_1,p_2)$ which has distribution $\max(1-x_1, x_2)$, and $q^{2/2}=\max(q_1,q_2)$ which has distribution $\max(x_1,1-x_2)$. Therefore:
\begin{eqnarray}
 && \mP_{\theta}(1 \notin [l^+_{\alpha}, n-l^-_{\alpha}] = \mP_{\theta}(1<l^+_{\alpha})+\mP_{\theta}(1<l^-_{\alpha})\nonumber \\
 && =  \mP_{\theta}(\max(p^{1/2}, p^{2/2})<\alpha)+\mP_{\theta}(\max(q^{1/2}, q^{2/2})<\alpha)\nonumber \\
 && = \mP((1-x_1,x_2) = (0,\alpha))+\mP((x_1,1-x_2) = (\alpha,0)) = 2\times \alpha \times (1-\alpha)>\alpha.\nonumber
\end{eqnarray}

\section{Applications: enhancing meta-analysis}

Figure \ref{Figure:schooltopk} shows simultaneous $95\%$ confidence intervals for $n^+(\mathcal{I}_k)$ with Fisher's combining function, where $\mathcal{I}_k$ are the indexes of the top $k=|\mathcal{I}_k|$ schools with largest (in absolute value) estimated effects $|y_i|$.
For example, among the top 38 schools, we observe at least 9 positive effects and at least 1 non-positive effect, or among the top 51 schools, we observe at least 10 positive effects and at least 2 non-positive effects, and so on.

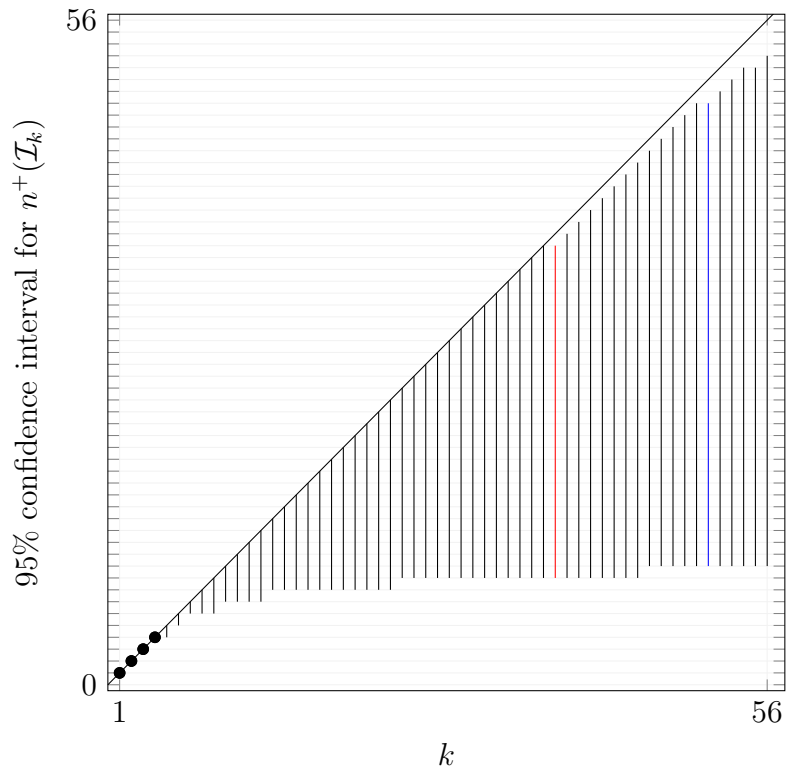
\begin{figure}
\centering
\begin{tikzpicture}[scale=1],
\begin{axis}[
 grid=major,
    grid style={line width=.25pt, draw=gray!10},
scale only axis=true,
height=9cm,
width=9cm,  
ylabel = {$95\%$ confidence interval for $n^+(\mathcal{I}_k)$},
xlabel = {$k$},
  xmin=0,
  xmax=57.5, 
    ymin=-0.5,
  ymax=56.5, 
ytick ={0,...,56},
yticklabels ={0,,,,,,,,,,,,,,,,,,,,,,,,,,,,,,,,,,,,,,,,,,,,,,,,,,,,,,  ,,56},
xtick ={1,56},
]
\addplot [color=black, solid, mark=*] coordinates { ( 1,1 ) ( 1,1 )}; 
 \addplot [color=black, solid, mark=*] coordinates { ( 2,2 ) ( 2,2 )}; 
 \addplot [color=black, solid, mark=*] coordinates { ( 3,3 ) ( 3,3 )}; 
 \addplot [color=black, solid, mark=*] coordinates { ( 4,4 ) ( 4,4 )}; 
 \addplot [color=black, solid] coordinates { ( 5,4 ) ( 5,5 )}; 
 \addplot [color=black, solid] coordinates { ( 6,5 ) ( 6,6 )}; 
 \addplot [color=black, solid] coordinates { ( 7,6 ) ( 7,7 )}; 
 \addplot [color=black, solid] coordinates { ( 8,6 ) ( 8,8 )}; 
 \addplot [color=black, solid] coordinates { ( 9,6 ) ( 9,9 )}; 
 \addplot [color=black, solid] coordinates { ( 10,7 ) ( 10,10 )}; 
 \addplot [color=black, solid] coordinates { ( 11,7 ) ( 11,11 )}; 
 \addplot [color=black, solid] coordinates { ( 12,7 ) ( 12,12 )}; 
 \addplot [color=black, solid] coordinates { ( 13,7 ) ( 13,13 )}; 
 \addplot [color=black, solid] coordinates { ( 14,8 ) ( 14,14 )}; 
 \addplot [color=black, solid] coordinates { ( 15,8 ) ( 15,15 )}; 
 \addplot [color=black, solid] coordinates { ( 16,8 ) ( 16,16 )}; 
 \addplot [color=black, solid] coordinates { ( 17,8 ) ( 17,17 )}; 
 \addplot [color=black, solid] coordinates { ( 18,8 ) ( 18,18 )}; 
 \addplot [color=black, solid] coordinates { ( 19,8 ) ( 19,19 )}; 
 \addplot [color=black, solid] coordinates { ( 20,8 ) ( 20,20 )}; 
 \addplot [color=black, solid] coordinates { ( 21,8 ) ( 21,21 )}; 
 \addplot [color=black, solid] coordinates { ( 22,8 ) ( 22,22 )}; 
 \addplot [color=black, solid] coordinates { ( 23,8 ) ( 23,23 )}; 
 \addplot [color=black, solid] coordinates { ( 24,8 ) ( 24,24 )}; 
 \addplot [color=black, solid] coordinates { ( 25,9 ) ( 25,25 )}; 
 \addplot [color=black, solid] coordinates { ( 26,9 ) ( 26,26 )}; 
 \addplot [color=black, solid] coordinates { ( 27,9 ) ( 27,27 )}; 
 \addplot [color=black, solid] coordinates { ( 28,9 ) ( 28,28 )}; 
 \addplot [color=black, solid] coordinates { ( 29,9 ) ( 29,29 )}; 
 \addplot [color=black, solid] coordinates { ( 30,9 ) ( 30,30 )}; 
 \addplot [color=black, solid] coordinates { ( 31,9 ) ( 31,31 )}; 
 \addplot [color=black, solid] coordinates { ( 32,9 ) ( 32,32 )}; 
 \addplot [color=black, solid] coordinates { ( 33,9 ) ( 33,33 )}; 
 \addplot [color=black, solid] coordinates { ( 34,9 ) ( 34,34 )}; 
 \addplot [color=black, solid] coordinates { ( 35,9 ) ( 35,35 )}; 
 \addplot [color=black, solid] coordinates { ( 36,9 ) ( 36,36 )}; 
 \addplot [color=black, solid] coordinates { ( 37,9 ) ( 37,37 )}; 
 \addplot [color=red, solid] coordinates { ( 38,9 ) ( 38,37 )}; 
 \addplot [color=black, solid] coordinates { ( 39,9 ) ( 39,38 )}; 
 \addplot [color=black, solid] coordinates { ( 40,9 ) ( 40,39 )}; 
 \addplot [color=black, solid] coordinates { ( 41,9 ) ( 41,40 )}; 
 \addplot [color=black, solid] coordinates { ( 42,9 ) ( 42,41 )}; 
 \addplot [color=black, solid] coordinates { ( 43,9 ) ( 43,42 )}; 
 \addplot [color=black, solid] coordinates { ( 44,9 ) ( 44,43 )}; 
 \addplot [color=black, solid] coordinates { ( 45,9 ) ( 45,44 )}; 
 \addplot [color=black, solid] coordinates { ( 46,10 ) ( 46,45 )}; 
 \addplot [color=black, solid] coordinates { ( 47,10 ) ( 47,46 )}; 
 \addplot [color=black, solid] coordinates { ( 48,10 ) ( 48,47 )}; 
 \addplot [color=black, solid] coordinates { ( 49,10 ) ( 49,48 )}; 
 \addplot [color=black, solid] coordinates { ( 50,10 ) ( 50,49 )}; 
 \addplot [color=blue, solid] coordinates { ( 51,10 ) ( 51,49 )}; 
 \addplot [color=black, solid] coordinates { ( 52,10 ) ( 52,50 )}; 
 \addplot [color=black, solid] coordinates { ( 53,10 ) ( 53,51 )}; 
 \addplot [color=black, solid] coordinates { ( 54,10 ) ( 54,52 )}; 
 \addplot [color=black, solid] coordinates { ( 55,10 ) ( 55,52 )}; 
 \addplot [color=black, solid] coordinates { ( 56,10 ) ( 56,53 )}; 
   \addplot [color=black, solid] coordinates { ( 0,0 ) ( 57,57 )};
\end{axis}
\end{tikzpicture} 
\caption{Simultaneous $95\%$ confidence intervals for $n^+(\mathcal{I}_k)$ with Fisher's combining function, where $\mathcal{I}_k$ are the indexes of the top $k=|\mathcal{I}_k|$ schools with largest (in absolute value) standardized effects. For example, within the top $38$ schools, we have at least 9 positive effects and at least 1 non-positive effect (red interval); or within the top $51$ schools, we have at least 10 positive effects and at least 2 non-positive effect (blue interval). Dots represent degenerate intervals $[k,k]$.}
\label{Figure:schooltopk}
\end{figure}

\section{Additional Simulation results}

For the simulation setting of $n=50$ parameters, described in \S~\ref{sec - sim} of the main text, we provide the following additional results. 

First, for the FDR controlling procedure of \cite{guo2015stepwise}, GR-FDR, we provide the actual coverage guarantee in the settings considered in Figure \ref{fig-SM-GRFDR-coverage}. Specifically, we provide the estimated probability (based on 2000 simulation runs for each data generation) that the number of positive and non-positive discoveries does not exceed the true number of positive and non-positive parameters, respectively, in the discovery set. The coverage can be much lower than the 0.95 coverage guarantee, since this procedure only provides the guarantee that the FDP is at most $\alpha$ in expectation. This is in contrast to all other procedures that provide the 0.95  coverage guarantee.

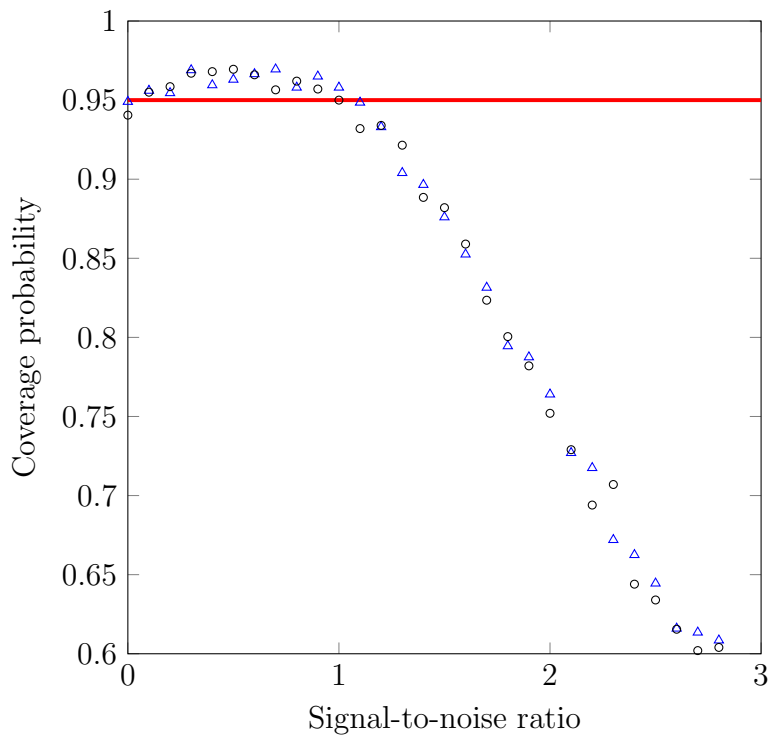
\begin{figure}[htbp]
\centering
\begin{tikzpicture}[scale=1]
	\begin{axis}[
	xmin = 0,
	xmax = 3,
	xtick ={0,...,3},
xticklabels ={0,1,2,3},
	ymin = 0.6,
	ymax = 1,
	xlabel=Signal-to-noise ratio,
		ylabel= Coverage probability,
		height=10cm,
		width=10cm
	]
	\addplot[only marks, mark size=2pt, color=blue,  mark=triangle] coordinates {
( 0,0.949 ) ( 0.1,0.956 ) ( 0.2,0.9545 ) ( 0.3,0.969 ) ( 0.4,0.9595 ) ( 0.5,0.963 ) ( 0.6,0.9665 ) ( 0.7,0.9695 ) ( 0.8,0.958 ) ( 0.9,0.965 ) ( 1,0.958 ) ( 1.1,0.9485 ) ( 1.2,0.933 ) ( 1.3,0.904 ) ( 1.4,0.8965 ) ( 1.5,0.876 ) ( 1.6,0.8525 ) ( 1.7,0.8315 ) ( 1.8,0.7945 ) ( 1.9,0.7875 ) ( 2,0.764 ) ( 2.1,0.727 ) ( 2.2,0.7175 ) ( 2.3,0.672 ) ( 2.4,0.6625 ) ( 2.5,0.6445 ) ( 2.6,0.616 ) ( 2.7,0.6135 ) ( 2.8,0.6085 ) ( 2.9,0.5755 ) ( 3,0.5795 )
};

	\addplot[only marks, mark size=1.5pt, color=black,  mark=o] coordinates {
( 0,0.9405 ) ( 0.1,0.955 ) ( 0.2,0.9585 ) ( 0.3,0.967 ) ( 0.4,0.968 ) ( 0.5,0.9695 ) ( 0.6,0.966 ) ( 0.7,0.9565 ) ( 0.8,0.962 ) ( 0.9,0.957 ) ( 1,0.95 ) ( 1.1,0.932 ) ( 1.2,0.934 ) ( 1.3,0.9215 ) ( 1.4,0.8885 ) ( 1.5,0.882 ) ( 1.6,0.859 ) ( 1.7,0.8235 ) ( 1.8,0.8005 ) ( 1.9,0.782 ) ( 2,0.752 ) ( 2.1,0.729 ) ( 2.2,0.694 ) ( 2.3,0.707 ) ( 2.4,0.644 ) ( 2.5,0.634 ) ( 2.6,0.6155 ) ( 2.7,0.602 ) ( 2.8,0.604 ) ( 2.9,0.5735 ) ( 3,0.5765 )
};

\addplot[mark=none, color=red, line width=1.5] coordinates {(0,.95) (3,.95)};
	\end{axis}
\end{tikzpicture} 
\caption{\label{fig-SM-GRFDR-coverage} For the GR-FDR procedure with discovery set $\mathcal R$, the coverage probability of $n^+(\mathcal R)$ and $n^-(\mathcal R)+n^0(\mathcal R)$  versus the signal-to-noise ratio, for the following settings: $n^+=n^-=15$ (black circles); $n^+=30$ and $n^-=0$ (blue triangles). The horizontal red line is the desired coverage probability of 0.95.}
\end{figure}

Second, for  DCT and partitioning, we compare four combining methods: Fisher, ALRT, Simes, and mSimes. Figure \ref{fig-SM-sim} shows that mSimes dominates Simes: it provides tighter bounds and more discoveries. For the bounds, Fisher is best and ALRT is a close second. For discoveries, this relation is reversed  (with ALRT, there are slightly more discoveries than with Fisher), but they are both  much worse than Simes.   

Finally, we note that we performed additional simulations, omitted for brevity, varying the parameter values within each data generation settings and considering more configurations of $n^+$ and $n^-$. The qualitative conclusions above regarding the relative performance of the different combining methods remained unchanged. 

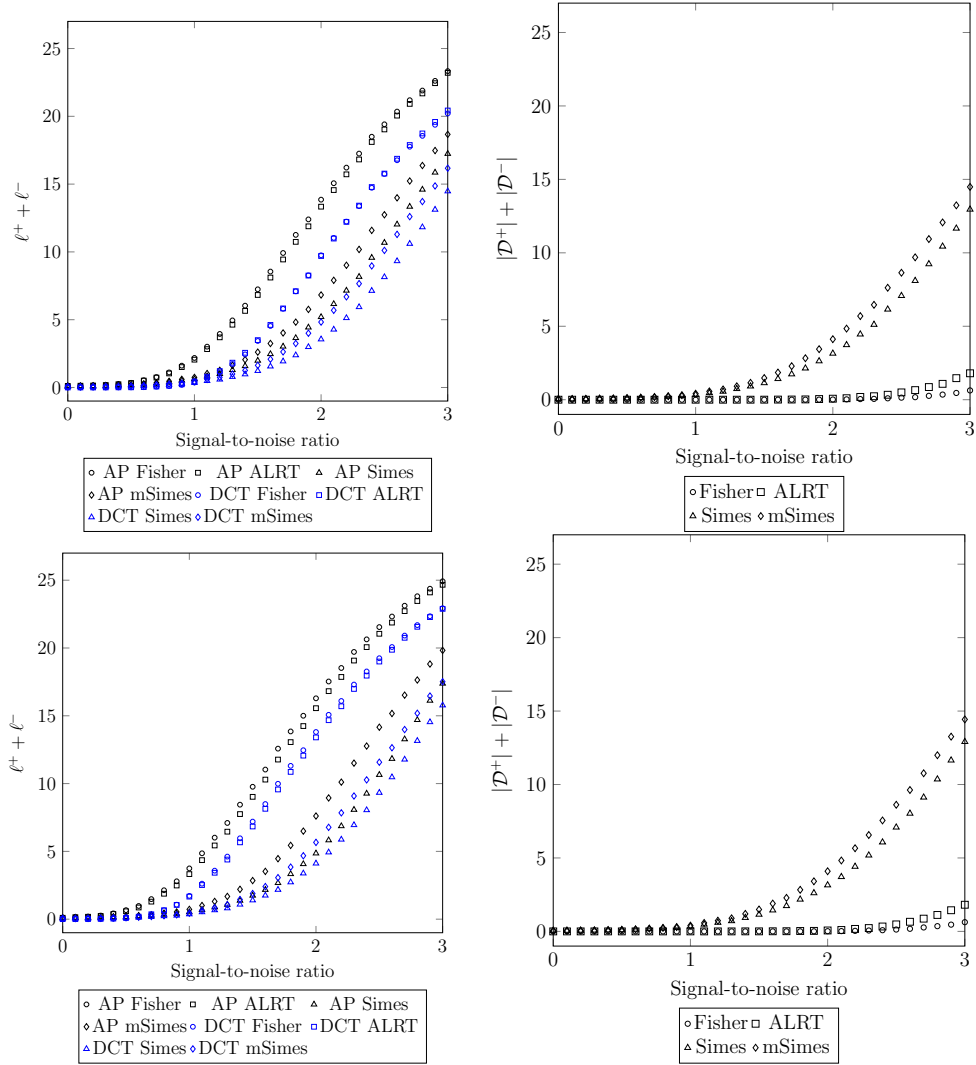
\begin{figure}[htbp]
\centering
\begin{tikzpicture}[scale=.6]
	\begin{axis}[
	xmin = 0,
	xmax = 3,
	xtick ={0,...,3},
xticklabels ={0,1,2,3},
	ymin = -1,
	ymax = 27,
	xlabel=Signal-to-noise ratio,
		ylabel= $\ell^+ + \ell^-$,
		height=10cm,
		width=10cm,
		legend style={at={(0.5,-0.15)},
		anchor=north,legend columns=3}
	]
	\addplot[only marks, mark size=1.5pt, color=black,  mark=o] coordinates {
( 0,0.1235 ) ( 0.1,0.135000000000005 ) ( 0.2,0.154500000000006 ) ( 0.3,0.188000000000002 ) ( 0.4,0.253 ) ( 0.5,0.344999999999999 ) ( 0.6,0.534500000000001 ) ( 0.7,0.767499999999998 ) ( 0.8,1.11 ) ( 0.9,1.6095 ) ( 1,2.178 ) ( 1.1,3.0115 ) ( 1.2,3.94900000000001 ) ( 1.3,4.948 ) ( 1.4,6.035 ) ( 1.5,7.24 ) ( 1.6,8.5495 ) ( 1.7,9.9195 ) ( 1.8,11.256 ) ( 1.9,12.3975 ) ( 2,13.854 ) ( 2.1,15.066 ) ( 2.2,16.216 ) ( 2.3,17.2505 ) ( 2.4,18.495 ) ( 2.5,19.4095 ) ( 2.6,20.349 ) ( 2.7,21.188 ) ( 2.8,21.9005 ) ( 2.9,22.6145 ) ( 3,23.332 )
};
 \addlegendentry{AP Fisher}

 	\addplot[only marks, mark size=1.5pt, color=black,  mark=square] coordinates {
( 0,0.115000000000002 ) ( 0.1,0.131 ) ( 0.2,0.143999999999998 ) ( 0.3,0.1785 ) ( 0.4,0.238999999999997 ) ( 0.5,0.329000000000001 ) ( 0.6,0.509 ) ( 0.7,0.731500000000004 ) ( 0.8,1.0365 ) ( 0.9,1.5045 ) ( 1,2.0515 ) ( 1.1,2.8295 ) ( 1.2,3.69900000000001 ) ( 1.3,4.6295 ) ( 1.4,5.648 ) ( 1.5,6.82400000000001 ) ( 1.6,8.1075 ) ( 1.7,9.441 ) ( 1.8,10.741 ) ( 1.9,11.8865 ) ( 2,13.3355 ) ( 2.1,14.562 ) ( 2.2,15.721 ) ( 2.3,16.818 ) ( 2.4,18.106 ) ( 2.5,19.033 ) ( 2.6,20.051 ) ( 2.7,20.9155 ) ( 2.8,21.688 ) ( 2.9,22.4445 ) ( 3,23.204 )
};
 \addlegendentry{AP ALRT}

  	\addplot[only marks, mark size=2pt, color=black,  mark=triangle] coordinates {
( 0,0.100000000000001 ) ( 0.1,0.0954999999999941 ) ( 0.2,0.115499999999997 ) ( 0.3,0.128999999999998 ) ( 0.4,0.159500000000001 ) ( 0.5,0.194499999999998 ) ( 0.6,0.2395 ) ( 0.7,0.308 ) ( 0.8,0.394500000000001 ) ( 0.9,0.521000000000001 ) ( 1,0.622499999999995 ) ( 1.1,0.830000000000005 ) ( 1.2,1.0065 ) ( 1.3,1.2955 ) ( 1.4,1.5835 ) ( 1.5,1.984 ) ( 1.6,2.4705 ) ( 1.7,3.024 ) ( 1.8,3.65300000000001 ) ( 1.9,4.432 ) ( 2,5.2015 ) ( 2.1,6.161 ) ( 2.2,7.15000000000001 ) ( 2.3,8.167 ) ( 2.4,9.573 ) ( 2.5,10.6795 ) ( 2.6,12.028 ) ( 2.7,13.339 ) ( 2.8,14.5985 ) ( 2.9,15.855 ) ( 3,17.24 )
};
 \addlegendentry{AP Simes}

   	\addplot[only marks, mark size=2pt, color=black,  mark=diamond] coordinates {
( 0,0.0974999999999966 ) ( 0.1,0.0999999999999943 ) ( 0.2,0.115500000000004 ) ( 0.3,0.133000000000003 ) ( 0.4,0.168999999999997 ) ( 0.5,0.201000000000001 ) ( 0.6,0.261000000000003 ) ( 0.7,0.340499999999999 ) ( 0.8,0.457000000000001 ) ( 0.9,0.600500000000004 ) ( 1,0.744500000000002 ) ( 1.1,1.0265 ) ( 1.2,1.2655 ) ( 1.3,1.6325 ) ( 1.4,2.0515 ) ( 1.5,2.608 ) ( 1.6,3.246 ) ( 1.7,4.023 ) ( 1.8,4.819 ) ( 1.9,5.7595 ) ( 2,6.824 ) ( 2.1,7.9055 ) ( 2.2,9.015 ) ( 2.3,10.181 ) ( 2.4,11.5885 ) ( 2.5,12.7385 ) ( 2.6,13.989 ) ( 2.7,15.23 ) ( 2.8,16.3715 ) ( 2.9,17.475 ) ( 3,18.667 )
};
 \addlegendentry{AP mSimes}

	\addplot[only marks, mark size=1.5pt, color=blue,  mark=o] coordinates {
( 0,0.00150000000000006 ) ( 0.1,0.00100000000000477 ) ( 0.2,0.00300000000000011 ) ( 0.3,0.00400000000000489 ) ( 0.4,0.00650000000000261 ) ( 0.5,0.00999999999999801 ) ( 0.6,0.0205000000000055 ) ( 0.7,0.0590000000000046 ) ( 0.8,0.104999999999997 ) ( 0.9,0.192 ) ( 1,0.369999999999997 ) ( 1.1,0.714500000000001 ) ( 1.2,1.1245 ) ( 1.3,1.731 ) ( 1.4,2.439 ) ( 1.5,3.4325 ) ( 1.6,4.5395 ) ( 1.7,5.802 ) ( 1.8,7.0905 ) ( 1.9,8.27249999999999 ) ( 2,9.749 ) ( 2.1,11.034 ) ( 2.2,12.233 ) ( 2.3,13.3745 ) ( 2.4,14.731 ) ( 2.5,15.7495 ) ( 2.6,16.759 ) ( 2.7,17.759 ) ( 2.8,18.559 ) ( 2.9,19.3755 ) ( 3,20.211 )
};
 \addlegendentry{DCT Fisher}

 	\addplot[only marks, mark size=1.5pt, color=blue,  mark=square] coordinates {
( 0,0.00250000000000483 ) ( 0.1,0.00300000000000011 ) ( 0.2,0.00300000000000011 ) ( 0.3,0.00500000000000256 ) ( 0.4,0.0115000000000052 ) ( 0.5,0.0139999999999958 ) ( 0.6,0.0309999999999988 ) ( 0.7,0.0734999999999957 ) ( 0.8,0.131 ) ( 0.9,0.237500000000004 ) ( 1,0.433499999999995 ) ( 1.1,0.795500000000004 ) ( 1.2,1.2115 ) ( 1.3,1.83 ) ( 1.4,2.5605 ) ( 1.5,3.4935 ) ( 1.6,4.612 ) ( 1.7,5.8315 ) ( 1.8,7.10149999999999 ) ( 1.9,8.255 ) ( 2,9.6985 ) ( 2.1,10.9845 ) ( 2.2,12.216 ) ( 2.3,13.414 ) ( 2.4,14.783 ) ( 2.5,15.777 ) ( 2.6,16.871 ) ( 2.7,17.882 ) ( 2.8,18.737 ) ( 2.9,19.5875 ) ( 3,20.4295 )
};
 \addlegendentry{DCT ALRT}

  	\addplot[only marks, mark size=2pt, color=blue,  mark=triangle] coordinates {
( 0,0.0520000000000067 ) ( 0.1,0.0414999999999992 ) ( 0.2,0.052500000000002 ) ( 0.3,0.0720000000000027 ) ( 0.4,0.0790000000000006 ) ( 0.5,0.103999999999999 ) ( 0.6,0.128999999999998 ) ( 0.7,0.174500000000002 ) ( 0.8,0.221499999999999 ) ( 0.9,0.279500000000006 ) ( 1,0.351500000000001 ) ( 1.1,0.485999999999997 ) ( 1.2,0.592500000000001 ) ( 1.3,0.787999999999997 ) ( 1.4,0.977499999999999 ) ( 1.5,1.227 ) ( 1.6,1.557 ) ( 1.7,1.9315 ) ( 1.8,2.3825 ) ( 1.9,2.9895 ) ( 2,3.555 ) ( 2.1,4.2775 ) ( 2.2,5.1255 ) ( 2.3,5.9355 ) ( 2.4,7.135 ) ( 2.5,8.1465 ) ( 2.6,9.3235 ) ( 2.7,10.603 ) ( 2.8,11.8245 ) ( 2.9,13.1185 ) ( 3,14.4735 )
};
 \addlegendentry{DCT Simes}

   	\addplot[only marks, mark size=2pt, color=blue,  mark=diamond] coordinates {
( 0,0.0510000000000019 ) ( 0.1,0.0429999999999993 ) ( 0.2,0.0534999999999997 ) ( 0.3,0.072499999999998 ) ( 0.4,0.0819999999999936 ) ( 0.5,0.107500000000002 ) ( 0.6,0.139499999999998 ) ( 0.7,0.192 ) ( 0.8,0.247999999999998 ) ( 0.9,0.324999999999996 ) ( 1,0.426500000000004 ) ( 1.1,0.597499999999997 ) ( 1.2,0.750999999999998 ) ( 1.3,0.994 ) ( 1.4,1.259 ) ( 1.5,1.631 ) ( 1.6,2.0855 ) ( 1.7,2.627 ) ( 1.8,3.2435 ) ( 1.9,4.002 ) ( 2,4.8235 ) ( 2.1,5.6935 ) ( 2.2,6.6865 ) ( 2.3,7.6595 ) ( 2.4,8.966 ) ( 2.5,10.111 ) ( 2.6,11.292 ) ( 2.7,12.601 ) ( 2.8,13.713 ) ( 2.9,14.8645 ) ( 3,16.172 )
};
 \addlegendentry{DCT mSimes}

	\end{axis}
\end{tikzpicture} ~ \begin{tikzpicture}[scale=0.65]
	\begin{axis}[
	xmin = 0,
	xmax = 3,
	xtick ={0,...,3},
xticklabels ={0,1,2,3},
	ymin = -1,
	ymax = 27,
	xlabel=Signal-to-noise ratio,
		ylabel= $|\mathcal{D}^+| + |\mathcal{D}^-|$,
		height=10cm,
		width=10cm,
		legend style={at={(0.5,-0.15)},
		anchor=north,legend columns=2}
	]
	\addplot[only marks, mark size=1.5pt, color=black,  mark=o] coordinates {
( 0,0 ) ( 0.1,0 ) ( 0.2,0 ) ( 0.3,0 ) ( 0.4,0 ) ( 0.5,0 ) ( 0.6,0 ) ( 0.7,0 ) ( 0.8,0 ) ( 0.9,0 ) ( 1,0 ) ( 1.1,0 ) ( 1.2,0 ) ( 1.3,0 ) ( 1.4,5e-04 ) ( 1.5,0.0015 ) ( 1.6,5e-04 ) ( 1.7,0.001 ) ( 1.8,0.0035 ) ( 1.9,0.006 ) ( 2,0.0135 ) ( 2.1,0.0235 ) ( 2.2,0.044 ) ( 2.3,0.0565 ) ( 2.4,0.0955 ) ( 2.5,0.1405 ) ( 2.6,0.172 ) ( 2.7,0.264 ) ( 2.8,0.3495 ) ( 2.9,0.46 ) ( 3,0.6325 )
};
 \addlegendentry{Fisher}

	\addplot[only marks, mark size=2pt, color=black,  mark=square] coordinates {
( 0,0 ) ( 0.1,0 ) ( 0.2,0 ) ( 0.3,0 ) ( 0.4,0 ) ( 0.5,0 ) ( 0.6,0 ) ( 0.7,0 ) ( 0.8,0 ) ( 0.9,0 ) ( 1,5e-04 ) ( 1.1,0 ) ( 1.2,5e-04 ) ( 1.3,5e-04 ) ( 1.4,0.001 ) ( 1.5,0.003 ) ( 1.6,0.0085 ) ( 1.7,0.0125 ) ( 1.8,0.021 ) ( 1.9,0.045 ) ( 2,0.069 ) ( 2.1,0.102 ) ( 2.2,0.189 ) ( 2.3,0.237 ) ( 2.4,0.336 ) ( 2.5,0.4975 ) ( 2.6,0.6535 ) ( 2.7,0.8615 ) ( 2.8,1.08 ) ( 2.9,1.468 ) ( 3,1.7995 )
};
 \addlegendentry{ALRT}

	\addplot[only marks, mark size=2pt, color=black,  mark=triangle] coordinates {
( 0,0.052 ) ( 0.1,0.0415 ) ( 0.2,0.0515 ) ( 0.3,0.0705 ) ( 0.4,0.078 ) ( 0.5,0.103 ) ( 0.6,0.1275 ) ( 0.7,0.169 ) ( 0.8,0.217 ) ( 0.9,0.269 ) ( 1,0.34 ) ( 1.1,0.4635 ) ( 1.2,0.5605 ) ( 1.3,0.741 ) ( 1.4,0.916 ) ( 1.5,1.141 ) ( 1.6,1.4215 ) ( 1.7,1.741 ) ( 1.8,2.14 ) ( 1.9,2.6295 ) ( 2,3.1495 ) ( 2.1,3.7305 ) ( 2.2,4.455 ) ( 2.3,5.114 ) ( 2.4,6.151 ) ( 2.5,7.076 ) ( 2.6,8.0985 ) ( 2.7,9.237 ) ( 2.8,10.437 ) ( 2.9,11.6615 ) ( 3,12.956 )
};
 \addlegendentry{Simes}

	\addplot[only marks, mark size=2pt, color=black,  mark=diamond] coordinates {
( 0,0.051 ) ( 0.1,0.043 ) ( 0.2,0.0525 ) ( 0.3,0.0725 ) ( 0.4,0.081 ) ( 0.5,0.107 ) ( 0.6,0.1365 ) ( 0.7,0.185 ) ( 0.8,0.24 ) ( 0.9,0.308 ) ( 1,0.409 ) ( 1.1,0.5685 ) ( 1.2,0.6995 ) ( 1.3,0.9155 ) ( 1.4,1.1485 ) ( 1.5,1.4765 ) ( 1.6,1.85 ) ( 1.7,2.285 ) ( 1.8,2.8255 ) ( 1.9,3.443 ) ( 2,4.121 ) ( 2.1,4.841 ) ( 2.2,5.694 ) ( 2.3,6.4585 ) ( 2.4,7.6205 ) ( 2.5,8.6435 ) ( 2.6,9.694 ) ( 2.7,10.9385 ) ( 2.8,12.0695 ) ( 2.9,13.233 ) ( 3,14.487 )
};
 \addlegendentry{mSimes}

	\end{axis}
\end{tikzpicture} 
\begin{tikzpicture}[scale=.6]
	\begin{axis}[
	xmin = 0,
	xmax = 3,
	xtick ={0,...,3},
xticklabels ={0,1,2,3},
	ymin = -1,
	ymax = 27,
	xlabel=Signal-to-noise ratio,
		ylabel= $\ell^+ + \ell^-$,
		height=10cm,
		width=10cm,
		legend style={at={(0.5,-0.15)},
		anchor=north,legend columns=3}
	]
	\addplot[only marks, mark size=1.5pt, color=black,  mark=o] coordinates {
( 0,0.105499999999999 ) ( 0.1,0.160000000000004 ) ( 0.2,0.1845 ) ( 0.3,0.265499999999996 ) ( 0.4,0.424999999999997 ) ( 0.5,0.664999999999999 ) ( 0.6,0.972500000000004 ) ( 0.7,1.4795 ) ( 0.8,2.1325 ) ( 0.9,2.7925 ) ( 1,3.7435 ) ( 1.1,4.8585 ) ( 1.2,6.0135 ) ( 1.3,7.0975 ) ( 1.4,8.44450000000001 ) ( 1.5,9.765 ) ( 1.6,11.0345 ) ( 1.7,12.5775 ) ( 1.8,13.846 ) ( 1.9,15.0035 ) ( 2,16.288 ) ( 2.1,17.53 ) ( 2.2,18.5255 ) ( 2.3,19.7095 ) ( 2.4,20.6365 ) ( 2.5,21.542 ) ( 2.6,22.3205 ) ( 2.7,23.1185 ) ( 2.8,23.816 ) ( 2.9,24.372 ) ( 3,24.919 )
};
 \addlegendentry{AP Fisher}

 	\addplot[only marks, mark size=1.5pt, color=black,  mark=square] coordinates {
( 0,0.101999999999997 ) ( 0.1,0.146499999999996 ) ( 0.2,0.178000000000004 ) ( 0.3,0.238500000000002 ) ( 0.4,0.381500000000003 ) ( 0.5,0.583500000000001 ) ( 0.6,0.861499999999999 ) ( 0.7,1.293 ) ( 0.8,1.875 ) ( 0.9,2.48350000000001 ) ( 1,3.3245 ) ( 1.1,4.354 ) ( 1.2,5.4365 ) ( 1.3,6.453 ) ( 1.4,7.75150000000001 ) ( 1.5,9.02200000000001 ) ( 1.6,10.2935 ) ( 1.7,11.7715 ) ( 1.8,13.059 ) ( 1.9,14.243 ) ( 2,15.561 ) ( 2.1,16.818 ) ( 2.2,17.864 ) ( 2.3,19.074 ) ( 2.4,20.065 ) ( 2.5,21.0485 ) ( 2.6,21.8725 ) ( 2.7,22.72 ) ( 2.8,23.4635 ) ( 2.9,24.098 ) ( 3,24.6615 )
};
 \addlegendentry{AP ALRT}

  	\addplot[only marks, mark size=2pt, color=black,  mark=triangle] coordinates {
( 0,0.100499999999997 ) ( 0.1,0.104999999999997 ) ( 0.2,0.1145 ) ( 0.3,0.122999999999998 ) ( 0.4,0.152999999999999 ) ( 0.5,0.175000000000004 ) ( 0.6,0.220500000000001 ) ( 0.7,0.268999999999998 ) ( 0.8,0.341499999999996 ) ( 0.9,0.402499999999996 ) ( 1,0.518500000000003 ) ( 1.1,0.698499999999996 ) ( 1.2,0.881 ) ( 1.3,1.071 ) ( 1.4,1.3725 ) ( 1.5,1.7445 ) ( 1.6,2.1705 ) ( 1.7,2.6725 ) ( 1.8,3.33199999999999 ) ( 1.9,4.068 ) ( 2,4.8545 ) ( 2.1,5.82299999999999 ) ( 2.2,6.8655 ) ( 2.3,8.072 ) ( 2.4,9.2675 ) ( 2.5,10.6485 ) ( 2.6,11.829 ) ( 2.7,13.2745 ) ( 2.8,14.6955 ) ( 2.9,16.1195 ) ( 3,17.3795 )
};
 \addlegendentry{AP Simes}

   	\addplot[only marks, mark size=2pt, color=black,  mark=diamond] coordinates {
( 0,0.102000000000004 ) ( 0.1,0.109999999999999 ) ( 0.2,0.115499999999997 ) ( 0.3,0.131999999999998 ) ( 0.4,0.170999999999999 ) ( 0.5,0.1935 ) ( 0.6,0.259499999999996 ) ( 0.7,0.3215 ) ( 0.8,0.444000000000003 ) ( 0.9,0.532000000000004 ) ( 1,0.721499999999999 ) ( 1.1,1.00449999999999 ) ( 1.2,1.3075 ) ( 1.3,1.677 ) ( 1.4,2.1995 ) ( 1.5,2.8465 ) ( 1.6,3.5245 ) ( 1.7,4.46149999999999 ) ( 1.8,5.4425 ) ( 1.9,6.493 ) ( 2,7.60749999999999 ) ( 2.1,8.941 ) ( 2.2,10.1065 ) ( 2.3,11.507 ) ( 2.4,12.768 ) ( 2.5,14.1575 ) ( 2.6,15.1695 ) ( 2.7,16.522 ) ( 2.8,17.6385 ) ( 2.9,18.817 ) ( 3,19.823 )
};
 \addlegendentry{AP mSimes}

	\addplot[only marks, mark size=1.5pt, color=blue,  mark=o] coordinates {
( 0,0.00200000000000244 ) ( 0.1,0.0034999999999954 ) ( 0.2,0.00750000000000028 ) ( 0.3,0.0174999999999983 ) ( 0.4,0.035499999999999 ) ( 0.5,0.0805000000000007 ) ( 0.6,0.162999999999997 ) ( 0.7,0.351999999999997 ) ( 0.8,0.651499999999999 ) ( 0.9,1.067 ) ( 1,1.717 ) ( 1.1,2.6185 ) ( 1.2,3.5785 ) ( 1.3,4.6155 ) ( 1.4,5.964 ) ( 1.5,7.1955 ) ( 1.6,8.49 ) ( 1.7,9.9855 ) ( 1.8,11.3085 ) ( 1.9,12.46 ) ( 2,13.804 ) ( 2.1,15.0715 ) ( 2.2,16.0915 ) ( 2.3,17.3025 ) ( 2.4,18.2685 ) ( 2.5,19.2445 ) ( 2.6,20.072 ) ( 2.7,20.9145 ) ( 2.8,21.6885 ) ( 2.9,22.342 ) ( 3,22.929 )
};
 \addlegendentry{DCT Fisher}

 	\addplot[only marks, mark size=1.5pt, color=blue,  mark=square] coordinates {
( 0,0.00300000000000011 ) ( 0.1,0.00450000000000017 ) ( 0.2,0.0114999999999981 ) ( 0.3,0.0220000000000056 ) ( 0.4,0.0420000000000016 ) ( 0.5,0.0904999999999987 ) ( 0.6,0.1845 ) ( 0.7,0.359499999999997 ) ( 0.8,0.660499999999999 ) ( 0.9,1.05 ) ( 1,1.663 ) ( 1.1,2.5175 ) ( 1.2,3.4225 ) ( 1.3,4.3955 ) ( 1.4,5.6665 ) ( 1.5,6.8365 ) ( 1.6,8.1395 ) ( 1.7,9.5665 ) ( 1.8,10.8615 ) ( 1.9,12.0615 ) ( 2,13.406 ) ( 2.1,14.6765 ) ( 2.2,15.7035 ) ( 2.3,16.98 ) ( 2.4,17.96 ) ( 2.5,19.001 ) ( 2.6,19.8685 ) ( 2.7,20.752 ) ( 2.8,21.564 ) ( 2.9,22.2785 ) ( 3,22.887 )
};
 \addlegendentry{DCT ALRT}

  	\addplot[only marks, mark size=2pt, color=blue,  mark=triangle] coordinates {
( 0,0.0510000000000019 ) ( 0.1,0.0534999999999997 ) ( 0.2,0.0599999999999952 ) ( 0.3,0.069500000000005 ) ( 0.4,0.088000000000001 ) ( 0.5,0.103000000000002 ) ( 0.6,0.131999999999998 ) ( 0.7,0.18 ) ( 0.8,0.228000000000002 ) ( 0.9,0.272500000000001 ) ( 1,0.365499999999997 ) ( 1.1,0.521000000000001 ) ( 1.2,0.666499999999999 ) ( 1.3,0.8095 ) ( 1.4,1.08300000000001 ) ( 1.5,1.3945 ) ( 1.6,1.7495 ) ( 1.7,2.1735 ) ( 1.8,2.7255 ) ( 1.9,3.38 ) ( 2,4.1115 ) ( 2.1,4.9315 ) ( 2.2,5.8755 ) ( 2.3,6.94450000000001 ) ( 2.4,8.04649999999999 ) ( 2.5,9.3275 ) ( 2.6,10.473 ) ( 2.7,11.7795 ) ( 2.8,13.155 ) ( 2.9,14.5375 ) ( 3,15.7735 )
};
 \addlegendentry{DCT Simes}

   	\addplot[only marks, mark size=2pt, color=blue,  mark=diamond] coordinates {
( 0,0.0535000000000068 ) ( 0.1,0.0500000000000043 ) ( 0.2,0.0609999999999999 ) ( 0.3,0.069500000000005 ) ( 0.4,0.0899999999999963 ) ( 0.5,0.106999999999999 ) ( 0.6,0.146000000000001 ) ( 0.7,0.199000000000005 ) ( 0.8,0.258000000000003 ) ( 0.9,0.322000000000003 ) ( 1,0.442500000000003 ) ( 1.1,0.645499999999998 ) ( 1.2,0.842000000000006 ) ( 1.3,1.064 ) ( 1.4,1.4485 ) ( 1.5,1.889 ) ( 1.6,2.4055 ) ( 1.7,3.0605 ) ( 1.8,3.8415 ) ( 1.9,4.678 ) ( 2,5.6625 ) ( 2.1,6.7735 ) ( 2.2,7.8415 ) ( 2.3,9.084 ) ( 2.4,10.273 ) ( 2.5,11.5765 ) ( 2.6,12.644 ) ( 2.7,13.984 ) ( 2.8,15.186 ) ( 2.9,16.4555 ) ( 3,17.512 )
};
 \addlegendentry{DCT mSimes}

	\end{axis}
\end{tikzpicture} ~ \begin{tikzpicture}[scale=0.65]
	\begin{axis}[
	xmin = 0,
	xmax = 3,
	xtick ={0,...,3},
xticklabels ={0,1,2,3},
	ymin = -1,
	ymax = 27,
	xlabel=Signal-to-noise ratio,
		ylabel= $|\mathcal{D}^+| + |\mathcal{D}^-|$,
		height=10cm,
		width=10cm,
		legend style={at={(0.5,-0.15)},
		anchor=north,legend columns=2}
	]
	\addplot[only marks, mark size=1.5pt, color=black,  mark=o] coordinates {
( 0,0 ) ( 0.1,0 ) ( 0.2,0 ) ( 0.3,0 ) ( 0.4,0 ) ( 0.5,0 ) ( 0.6,0 ) ( 0.7,0 ) ( 0.8,0 ) ( 0.9,0 ) ( 1,0 ) ( 1.1,0 ) ( 1.2,0 ) ( 1.3,0 ) ( 1.4,5e-04 ) ( 1.5,0 ) ( 1.6,0.001 ) ( 1.7,0.003 ) ( 1.8,0.0045 ) ( 1.9,0.007 ) ( 2,0.0145 ) ( 2.1,0.019 ) ( 2.2,0.037 ) ( 2.3,0.0615 ) ( 2.4,0.0815 ) ( 2.5,0.1265 ) ( 2.6,0.1775 ) ( 2.7,0.2685 ) ( 2.8,0.356 ) ( 2.9,0.46 ) ( 3,0.6275 )
};
 \addlegendentry{Fisher}

	\addplot[only marks, mark size=2pt, color=black,  mark=square] coordinates {
( 0,0 ) ( 0.1,0 ) ( 0.2,0 ) ( 0.3,0 ) ( 0.4,0 ) ( 0.5,0 ) ( 0.6,0 ) ( 0.7,0 ) ( 0.8,0 ) ( 0.9,0 ) ( 1,0 ) ( 1.1,0 ) ( 1.2,0 ) ( 1.3,0 ) ( 1.4,0.0025 ) ( 1.5,0.003 ) ( 1.6,0.0085 ) ( 1.7,0.0125 ) ( 1.8,0.0275 ) ( 1.9,0.0345 ) ( 2,0.063 ) ( 2.1,0.0945 ) ( 2.2,0.1685 ) ( 2.3,0.227 ) ( 2.4,0.333 ) ( 2.5,0.4765 ) ( 2.6,0.6545 ) ( 2.7,0.8675 ) ( 2.8,1.1195 ) ( 2.9,1.4405 ) ( 3,1.8135 )
};
 \addlegendentry{ALRT}

	\addplot[only marks, mark size=2pt, color=black,  mark=triangle] coordinates {
( 0,0.05 ) ( 0.1,0.0535 ) ( 0.2,0.058 ) ( 0.3,0.0685 ) ( 0.4,0.085 ) ( 0.5,0.0995 ) ( 0.6,0.1285 ) ( 0.7,0.1725 ) ( 0.8,0.212 ) ( 0.9,0.2565 ) ( 1,0.3375 ) ( 1.1,0.478 ) ( 1.2,0.6015 ) ( 1.3,0.7105 ) ( 1.4,0.9325 ) ( 1.5,1.152 ) ( 1.6,1.4415 ) ( 1.7,1.737 ) ( 1.8,2.1785 ) ( 1.9,2.613 ) ( 2,3.149 ) ( 2.1,3.7135 ) ( 2.2,4.414 ) ( 2.3,5.1815 ) ( 2.4,6.0695 ) ( 2.5,7.0845 ) ( 2.6,8.0285 ) ( 2.7,9.125 ) ( 2.8,10.359 ) ( 2.9,11.648 ) ( 3,12.91 )
};
 \addlegendentry{Simes}

	\addplot[only marks, mark size=2pt, color=black,  mark=diamond] coordinates {
( 0,0.052 ) ( 0.1,0.05 ) ( 0.2,0.0595 ) ( 0.3,0.067 ) ( 0.4,0.0885 ) ( 0.5,0.103 ) ( 0.6,0.14 ) ( 0.7,0.19 ) ( 0.8,0.2395 ) ( 0.9,0.2975 ) ( 1,0.394 ) ( 1.1,0.575 ) ( 1.2,0.7285 ) ( 1.3,0.8915 ) ( 1.4,1.1915 ) ( 1.5,1.4905 ) ( 1.6,1.895 ) ( 1.7,2.2835 ) ( 1.8,2.83 ) ( 1.9,3.4105 ) ( 2,4.095 ) ( 2.1,4.8275 ) ( 2.2,5.6625 ) ( 2.3,6.5615 ) ( 2.4,7.553 ) ( 2.5,8.6195 ) ( 2.6,9.625 ) ( 2.7,10.7655 ) ( 2.8,11.993 ) ( 2.9,13.257 ) ( 3,14.431 )
};
 \addlegendentry{mSimes}

	\end{axis}
\end{tikzpicture} 
\caption{\label{fig-SM-sim} 
The average sum of the lower bounds on positive and nonpositive parameters (left column) and  number of individual hypotheses rejected (right column), versus the signal-to-noise ratio,  
for the following methods: AP with combining functions Fisher (AP FISHER), ALRT (AP ALRT), Simes (AP SIMES), modified Simes (AP mSIMES); DCT with combining functions Fisher (DCT FISHER), ALRT (DCT ALRT), Simes (DCT SIMES), modified Simes (DCT mSIMES);. $n^+=n^-=15$ in the first row; $n^+=30$ and $n^-=0$ in the second row.  In the right column, the absence of DCT procedures  is due to the fact that they coincide with the AP procedures. 
} 
\end{figure}

\end{document}